\documentclass[11pt,letterpaper]{article}

\usepackage{amsmath,amssymb,amsthm,mathabx,mathrsfs}
\usepackage{graphicx,xcolor}
\usepackage{fullpage}
\usepackage[sort,round]{natbib}
\usepackage{multirow}
\usepackage{booktabs}
\usepackage{float}

\DeclareMathOperator{\Cov}{Cov}

\DeclareMathOperator{\rk}{rank}
\DeclareMathOperator{\tr}{tr}

\newcommand{\field}[1]{\mathbb{#1}}
\newcommand{\R}{\field{R}}

\newcommand{\E}{\field{E}}

\newcommand{\h}[1]{\boldsymbol{#1}}
\newcommand{\diag}{\mbox{diag}}

\newcommand{\hZ}{{\h{Z}}}
\newcommand{\hzero}{{\h{0}}}

\newcommand{\vect}{\mathrm{vec}}
\newcommand{\hx}{\h{x}}
\newcommand{\be}{\h{e}}
\newcommand{\hA}{\h{A}}

\newcommand{\hC}{\h{C}}
\newcommand{\hE}{\h{E}}
\newcommand{\hX}{\h{X}}

\newcommand{\hW}{\h{W}}
\newcommand{\hG}{\h{G}}
\newcommand{\hI}{\h{I}}
\newcommand{\hH}{\h{H}}
\newcommand{\hJ}{\h{J}}
\newcommand{\hL}{\h{L}}
\newcommand{\hM}{\h{M}}
\newcommand{\hV}{\h{V}}
\newcommand{\hR}{{\h{R}}}
\newcommand{\halpha}{\h{\alpha}}

\newcommand{\hgamma}{\h{\gamma}}
\newcommand{\hQ}{\h{Q}}
\newcommand{\hP}{\h{P}}
\newcommand{\hF}{\h{F}}
\newcommand{\hS}{\h{S}}
\newcommand{\hU}{\h{U}}
\newcommand{\hz}{\h{z}}
\newcommand{\PPP}{\mathcal{P}}

\theoremstyle{plain}
\newtheorem{theorem}{Theorem}
\newtheorem{corollary}[theorem]{Corollary}

\newtheorem{lemma}[theorem]{Lemma}

\theoremstyle{definition}

\theoremstyle{definition}
\newtheorem{remark}{Remark}

\title{Reduced-Rank Autoregressive Models for Matrix Time Series}

\author{
    Han Xiao$^a$\footnote{Han Xiao is Professor,
    Department of Statistics, Rutgers
    University, Piscataway, NJ. E-mail:
    hxiao@stat.rutgers.edu. Yuefeng Han is Assistant Professor, Department of Applied and Computational Mathematics and Statistics, University of Notre Dame, IN. Email: yuefeng.han@nd.edu.
    Rong Chen is Professor, Department of
    Statistics, Rutgers University, Piscataway, NJ. E-mail:
    rongchen@stat.rutgers.edu. Ama Ampadu-Kissi is PhD student, Department of Statistics, Rutgers
    University, Piscataway, NJ. E-mail: aka117@scarletmail.rutgers.edu. 
    Han Xiao is the
    corresponding author.}, Yuefeng Han$^b$, Rong Chen$^a$ and Ama Ampadu-Kissi$^a$ \\
    $^a$Rutgers University, $^b$University of Notre Dame
}
\date{}

\begin{document}

\maketitle

\noindent {\bf Abstract.} Matrix time series is a series of matrix data observed over time. Analytical tools for such time series is needed in many applications in finance, economics, engineering and many other fields. To avoid the use of vectorization of the matrices which loses the column and row information, and the vector autoregression framework in traditional time series analysis, \cite{chen2021autoregressive} proposed the Matrix Autoregressive (MAR) Model. The model maintains and utilizes the matrix structure, leading to a substantial dimensional reduction and admitting explicit interpretations, comparing with the vector autoregressive model on the vectorized data. However, the MAR model still encounters difficulties in dealing with large dimensional matrix time series as the coefficient matrices in MAR models are also large. In this paper we propose to achieve further dimension reduction through reduced-rank constraints of the coefficient matrices in the MAR model. Estimation and rank determination procedures are studied. Theoretical investigation and empirical examples show that the reduced-rank constraint can achieve higher statistical efficiency than the MAR model.

\bigskip

\noindent {\bf Keywords:} Forecasting; Matrix time series; Rank determination; Reduced-rank regression

\newpage

\section{Introduction}
\label{sec:intro}

\nocite{anderson:2013,horn2012matrix,chen2012sparse,chen2021autoregressive,chen2019factor}

Observations in matrix and tensor (multi-dimensional array) forms have been generated and collected more and more abundantly in many fields including biological/medical research, economics, engineering, finance, signal processing, social sciences etc. In response to the urgent need of analytical tools for analyzing such type of data in various applications, many optimization and statistical methods/procedures have been proposed and studied. Similar to the use of matrix decomposition for analysis of vector observations, tensor decomposition and estimation methods play a principal role in analyzing matrix/tensor data \citep{cichocki2015tensor,cichocki2009nonnegative,de2000best,de2000multilinear,sidiropoulos2017tensor,anandkumar2014tensor,de2008tensor}.

In many applications, the matrices are observed through time, and hence form a matrix-valued time series. Although it is possible and perhaps convenient to treat time as another mode and apply the tensor methods to such a three-way tensor, the time dimension is intrinsically different, and the temporal dependence requires careful modeling and analysis to aid practitioners on acquiring a diagnostic understanding of the dynamics and making reliable forecasts. It has been witnessed that multilinear models can reduce the dimension and improve the estimation stability for matrix/tensor data \citep{ding2018matrix,raskutti2019convex, zhao:2014,zhou:2013}.
For dependent data, there has been many recent works on factor models of matrix/tensor time series, see \cite{wang:2018,chen2019factor,han2020tensor,han2020number,gao2021two} among others. On the other hand, \cite{hoff2015multilinear} pioneered in suggesting the multilinear model for longitudinal tensor data. \cite{chen2021autoregressive} proposed the matrix autoregressive model (MAR), which retains the matrix form of the data and specifies the autoregressive relationship through a bilinear matrix product.
Besides the interpretations adherent to its matrix and bilinear form, the MAR also reduces the model complexity significantly, comparing to the approach of concatenating the matrix observation into a long vector and then using the traditional vector autoregressive (VAR) model \citep{hannan:1970,lutkepohl:2005,tsay:2014,tiao:1981}.

When the matrix observations are themselves of large dimensions, the MAR model still involves a large number of parameters. It is desirable and sometimes necessary to reduce the dimension even further. In this paper, we propose the reduced-rank matrix autoregressive model (RRMAR), which assumes the form of MAR, but requires in addition that the coefficient matrices have ranks smaller than their dimensions.

We note that another natural approach to reducing the model complexity is to impose sparsity on the MAR. This approach is closely related to recent works on sparse VAR models \citep{loh2011high,basu2015regularized,melnyk2016estimating,davis2016sparse,han:2015,kock:2015,lin2017regularized,nicholson2017varx}. In addition, \cite{basu2019low} and \cite{lin2020regularized} considered additional low rank constraints on the coefficient matrices, \cite{hall2018learning} introduced the generalized VAR model, and \cite{han2020sparse} focused on the nonlinear sparse VAR model. \cite{ghosh2019high} and \cite{ghosh2020strong} considered the high dimensional VAR from a Bayesian perspective. 

In contrast to the aforementioned works based on sparsity, the thrust of the present paper hinges upon the low rank structure of the coefficient matrices in MAR. As will be elaborated in Section~\ref{sec:rrmar}, the low rank matrices in RRMAR continue to admit natural interpretations, and lead to a greater dimension reduction compared to MAR. It also relates to and provides a generative mechanism for the dynamic matrix factor models of \cite{wang:2018}, and has a close connection to the hierarchical factor models.

The proposed model and estimation procedure are related to the reduced-rank regression \citep{anderson:1951,izenman:1975,reinsel:1998}. 
We consider two estimators: one based on least squares (RR.LS) and one based on maximum likelihood (RR.CC). It is worth pointing out that they correspond to two different algorithms for vector reduced-rank regression: one minimizing the trace of the sample covariance matrix, and the other the determinant, where the latter also corresponds to the canonical correlation analysis (see for example \cite{reinsel:1998} for more details). The likelihood-based RR.CC is indeed the maximum likelihood estimator if the covariance tensor of the error matrix has the form of a tensor product of two covariance matrices. Even when this assumption does not hold, the RR.CC nevertheless can be viewed as an estimator obtained together with a regularized estimation of the covariance tensor, and can still potentially lead to superior performances over the RR.LS.

Here we shall emphasize two significant differences between our model and the classical reduced-rank regression. First, the observations are in matrix form, and the model takes a bilinear form. Second, the algorithms requires running reduced-rank least squares/maximum likelihood iteratively. We develop central limit theorems for the estimators of the coefficient matrices, as well as their singular vectors. The bilinear form of the matrix model also makes the analysis substantially different from the vector case. 

The estimation procedures depend on the ranks of the two coefficient matrices in the RRMAR model. We propose to use information criterion based procedures to identify the ranks of these matrices. Since two ranks are to be determined, a thorough search over all possible pairs of ranks can be very costly, so we also introduce procedures to select the two ranks separately. Asymptotic consistency of these selection procedures are established.

The rest of this article is organized as follows. The RRMAR model is introduced in Section~\ref{sec:rrmar}, together with its basic properties, interpretations and connections with other models. In Section~\ref{sec:est} we propose two estimators, RR.LS and RR.CC. Asymptotic distributions are provided for both of them in Section~\ref{sec:asymp}, as well as corresponding estimators of the leading singular vectors of the coefficient matrices. The model/rank selection procedures based on information criterion are introduced in Section~\ref{sec:rank}, with their consistency properties. We use an extensive numerical study and an example in finance to demonstrate the performances of the proposed models and estimators in Section~\ref{sec:num}. All the proofs are collected in the Appendix. 

\subsection{Notations}
We gather the notations and the definitions of some special matrices in this section.

We use $\|\cdot\|_F$ to denote the Frobenius norm of a matrix, and $\rho(\cdot)$ the spectral radius. We use $\otimes$ to denote the Kronecker product, and $\circ$ the (point-wise) Hadamard product of two matrices. The notation $\h{1}_k$ stands for a $k$-dimensional vector with all entries equal to one. For any matrix $\h{M}$, we use $\h{M}[i,]$ and $\h{M}[,j]$ to denote its $i$-th row and $j$-th column respectively. The column space of $\hM$ is denoted by $\mathrm{col}(\hM)$. The matrix vectorization, denoted by $\vect(\cdot)$, turns a matrix into a vector by stacking its columns.

For any positive integer $p$, let $\be_{p,j}\in\mathbb{R}^p$ be the $j$-th base vector whose $j$-the entry is 1, and others zero. For any two positive integers $p$ and $q$, let $\hJ_{p,q}$ be the $(pq)\times(pq)$ permutation matrix defined as
\begin{equation}
    \label{eq:J}
    \hJ_{p,q} = \left[\hI_q\otimes\be_{p,1},\hI_q\otimes\be_{p,2},\ldots,\hI_q\otimes\be_{p,p}\right].
\end{equation}
The permutation $\hJ_{p,q}$ does the following: for any $p\times q$ matrix $\hM$, $   \hJ_{p,q}\vect(\hM')=\vect(\hM).$
In other words, $\hJ_{p,q}$ connects the vectorization of a matrix and its transpose. 

Let $\hL_p$ be the $p^2\times p$ matrix whose $j$-th column is given by $\be_{p,j}\otimes\be_{p,j}$, i.e.
\begin{equation}
    \label{eq:L}
    \hL_p=[\be_{p,1}\otimes\be_{p,1},\be_{p,2}\otimes\be_{p,2},\ldots,\be_{p,p}\otimes\be_{p,p}].
\end{equation}
For any $p\times p$ matrix $\hM=(m_{jk})$, the following operation extracts its diagonals:
\begin{equation*}
    \hL_p'\vect(\hM)=(m_{11},\ldots,m_{pp})',
\end{equation*}
and furthermore,
\begin{equation*}
    \vect^{-1}[\hL_p\hL_p'\vect(\hM)]=\diag(\hM),
\end{equation*}
where $\diag(\hM)$ is the $p\times p$ diagonal matrix keeping $\hM$'s diagonal elements.

\section{Reduced-Rank MAR Model}
\label{sec:rrmar}

The {\it reduced-rank matrix autoregressive model} (RRMAR) takes the
form
\begin{equation}
  \label{eq:marrr}
  \hX_t=\hA_1\hX_{t-1}\hA_2'+\hE_t,
\end{equation}
where $\h{A}_i$ are $d_i\times d_i$ autoregressive coefficient
matrices {of ranks $k_i \leq d_i$}, and $\h{E}_t\in\R^{d_1\times d_2}$ is a matrix
white noise. It {is} 
the same as the MAR model proposed
by \cite{chen2021autoregressive},
except for the additional low rank assumption that $\rk(\hA_i)=k_i \leq d_i$, for
$i=1,2$. It is worth observing that the number of parameters to
determine $\hA_i$ under the rank constraint is $d_i^2-(d_i-k_i)^2=(2d_i-k_i)k_i$ \citep[see for example][]{camba2003tests, reinsel:1998}, as
opposed to $d_i^2$ for the unconstrained $\hA_i$, and the former can
be much smaller if $k_i\ll d_i$. We also assume that $\|\hA_1\|_F=1$,
so that $\hA_1$ and $\hA_2$ are identified up to a sign change. To guarantee that the model \eqref{eq:marrr} is causal and stationary, we require that $\rho(\hA_1)\cdot\rho(\hA_2)<1$.

To better understand the implication of the low rank assumption, we
write $\hA_i=\hA_{il}\hA_{ic}'$, where $\hA_{il}$ and $\hA_{ic}$ are
both $d_i\times k_i$ full rank matrices. The model \eqref{eq:marrr} is
then written as
\begin{equation}
  \label{eq:marrr_fac}
  \hX_t=\hA_{1l}\,\boxed{\hA_{1c}'\hX_{t-1}\hA_{2c}}\,\hA_{2l}'+\hE_t.
\end{equation}
The boxed part $\hF_{t}:=\hA_{1c}'\hX_{t-1}\hA_{2c}$ is a
$k_1\times k_2$ matrix, which can be viewed as a composite and much
smaller version of the $d_1\times d_2$ matrix $\hX_{t-1}$. The conditional
expectation of $\hX_t$ given $\hX_{t-1}$ is then given by loading on
$\hF_{t}$ from left by $\hA_{1l}$, and from right by
$\hA_{2l}'$. The RRMAR model therefore provides a generating mechanism
for the matrix factor model $\hX_t=\hA_{l1}\hF_t\hA_{l2}' +\hE_t$, introduced in \cite{wang:2018}, in which the factor process $\hF_t$ is assumed to be latent and unobserved. In the RRMAR model in \eqref{eq:marrr_fac}, $\hF_t$ depends on $\hX_{t-1}$ hence is observed given the parameters. Due to this connection, we call $\hA_{ic}$ the composition matrix, and $\hA_{il}$ the loading matrix.

The RRMAR model is also related to the hierarchical factor models in the econometrics literature \citep{moench2013dynamic}. Specifically, let $\hF^*_t = \hA_{1c}'\hX_{t-1}\hA_{2}'$, then
\begin{equation*}
    \hX_t = \hA_{1l}\hF^*_t + \hE_t,
\end{equation*}
which means that the $j$-th column of $\hX_t$ follows a factor model with loading $\hA_{1l}$, and factors $\hF^*_t[,j]$, the $j$-th column of $\hF^*_t$. In the next layer, we have
\begin{equation*}
    \hF_t^{*'} = \hA_{2l} \hF_t',
\end{equation*}
which says that the $d_2\times k_1$ factor matrix $\hF_t^{*'}$ is further driven by a smaller factor matrix $\hF_t'$ {(defined by the boxed part in \eqref{eq:marrr_fac})}, the $j$-th column of $\hF_t^{*'}$ corresponding to the loading $\hA_{2l}$, and factors $\hF_t'[,j]$. Therefore, the model \eqref{eq:marrr} also gives a generating mechanism of a special instance of the hierarchical factor model.

\nocite{giannone2008nowcasting,diebold2008global}

\begin{remark}
{Again we emphasize that RRMAR is strictly an AR model, with the conditional mean of $\hX_t$ given the past information solely depending on $\hX_{t-1}$. As in all traditional AR models, we assume the noise process $\hE_t$ is white (in time), though the elements in $\hE_t$ are allowed to have (strong) correlations among them. It is possible to extend the model to have an ARMA form to allow non-white error processes, but such an ARMA model is extremely difficult to analyze and typically less useful in practice even for vector time series, due to various ambiguities. Although RRMAR can be written in a factor model form, strictly speaking, it is not a factor model, as it is {\it generative} in which $\hF_t$ is based on the past information $\hX_{t-1}$. On the other hand, a typical factor model is often {\it descriptive}, with the factor process $\hF_t$ being latent and its estimator is typically a linear combination of the current observation $\hX_t$.}
\end{remark}

There are many potential extensions of model \eqref{eq:marrr}. The first is the extension to RRMAR$(p)$ model:
\begin{equation}
    \label{eq:marrrp}
    \hX_t=\sum_{j=1}^p\hA_{j1}\hX_{t-j}\hA_{j2}'+\hE_t,
\end{equation}
where all $\hA_{ji}$ are of low ranks. The RRMAR$(p)$ model also extends the reduced-rank vector autoregressive models \citep{velu1986reduced, camba2003tests,al2019testing}. The second extension is more subtle. Although only the lag-1 observation $\hX_{t-1}$ is involved on the right hand side of \eqref{eq:marrr}, there can be multiple terms in the form 
\begin{equation}
  \label{eq:marrrk}
  \hX_t=\sum_{j=1}^r\hA_{1}^{(j)}\hX_{t-1}\left(\hA_{2}^{(j)}\right)'+\hE_t.
\end{equation}
To see this extension more clearly, we take vectorization on both sides of \eqref{eq:marrrk}:
\begin{equation*}
    \vect(\hX_t) = \left(\sum_{j=1}^r \hA_{2}^{(j)}\otimes\hA_{1}^{(j)}\right) \vect(\hX_{t-1}) + \hE_t.
\end{equation*}
It is seen from the preceding equation that the RRMAR model \eqref{eq:marrr} amounts to restricting the coefficient matrix of the VAR(1) model to the form of a Kronecker product, and the model \eqref{eq:marrrk} is more flexible by representing the coefficient matrix as a sum of $r$ Kronecker products. 

The extension to AR($p$) in \eqref{eq:marrrp} is quite straightforward, and the alternating estimation algorithms and rank determination procedures proposed below can be easily extended for this model, with heavier computational cost.
It can help to make the model more parsimonious by requiring (i) the matrices $\hA_{ji}$ to have the same ranks across $1\leq j\leq p$, or (ii) the rank of $\hA_{ji}$ decreases as the lag $j$ increases, or (iii) the matrices $\hA_{ji}$ have the same column/row spaces across $j$. 
The extension to \eqref{eq:marrrk} is more intricate,
as certain identifiability constraints are needed. We shall focus on model \eqref{eq:marrr} in the main part of the paper, and treat the extension \eqref{eq:marrrp} briefly in Appendix II, but leave the extension \eqref{eq:marrrk} for future studies. Finally, we add that the two extensions \eqref{eq:marrrp} and \eqref{eq:marrrk} can be combined to give a more comprehensive model.

\section{Estimation}
\label{sec:est}

Suppose a matrix time series $\{\hX_t\}$ of length $T$ is observed. To
estimate the coefficient matrices $\hA_i$, we propose to use the
alternating reduced-rank regression, updating one, while holding the
other fixed. Specifically, suppose $\hA_2$ is given, we discuss how
to estimate $\hA_1$. Recall that $\hA[,j]$ denotes the $j$-th column
of a matrix $\hA$. We also make the convention that $\hA'[,j]$ denotes the
$j$-th column of $\hA'$, i.e. the $j$-th row of $\hA$ as a column
vector. The $j$-th column of the model equation \eqref{eq:marrr} is
\begin{equation*}
  \hX_t[,j]=\hA_1\,\boxed{\hX_{t-1}\hA_2'[,j]} + \hE_t[,j].
\end{equation*}
Since $\hA_2$ is fixed, the preceding equation can be viewed as the
reduced-rank regression involving $(T-1)d_2$ sample units, where each
column $\hX_t[,j]$ is a response vector, the boxed vector is the
covariate, and $\hA_1$ is the coefficient matrix. In Section~\ref{sec:ils} we consider the estimation of $\hA_1$ by least squares. On the other hand, under normality, the
classical reduced-rank regression minimizes the determinant of the
sample covariance matrix of the error vectors, under the rank
constraint, which is related and in fact equivalent to canonical
correlation analysis \citep{anderson:2013,reinsel:1998}. In Section~\ref{sec:icc} we introduce a special covariance structure of $\hE_t$, under which we seek to estimate $\hA_1$ by the Gaussian MLE.

{In this section we focus on the estimation of the coefficient matrices given the ranks $k_1$ and $k_2$. The determination of the ranks will be discussed in Section 5.}

\subsection{Alternating least squares}
\label{sec:ils}

The least squares estimators are solutions of
\begin{equation}
    \label{eq:rrmin}
    \min_{\hA_1,\,\hA_2: \; \rk(\hA_1)=k_1,\,\rk(\hA_2)=k_2}\sum_{t=2}^T\left\|\hX_t-\hA_1{\hX_{t-1}\hA_2'}\right\|_F^2 
\end{equation}
We denote the least squares estimators by $\hat\hA_i^{\hbox{\tiny ls}}$, and will refer to them as the RR.LS estimators.
Suppose $\hA_2$ is given, the problem \eqref{eq:rrmin} becomes minimizing the trace of the sample covariance
matrix {of the residuals} under the rank constraint on $\hA_1$:
\begin{align*}
  & \min_{\hA_1: \; \rk(\hA_1)=k_1}\sum_{t=2}^T\left\|\hX_t-\hA_1{\hX_{t-1}\hA_2'}\right\|_F^2 \\ 
  \Longleftrightarrow\quad &
  \min_{\hA_1: \; \rk(\hA_1)=k_1}\tr\left[\sum_{t=2}^T\sum_{j=1}^{d_2}\left(\hX_t[,j]-\hA_1{\hX_{t-1}\hA_2'[,j]}\right)\left(\hX_t[,j]-\hA_1{\hX_{t-1}\hA_2'[,j]}\right)'\right].\nonumber
\end{align*}
Let $\hS_{xx}=\sum_t \hX_{t-1}\hA_2'\hA_2\hX_{t-1}'$,
$\hS_{yx}=\sum_t \hX_t\hA_2\hX_{t-1}'$, and
$\hU:=[U_1, U_2,\ldots, U_{k_1}]$, where
$U_j$ is the $j$-th leading normalized eigenvector of
$\hS_{yx}\hS_{xx}^{-1}\hS_{xy}$. Then $\hA_1$ can be updated as
\begin{equation*}
  \check \hA_1^{\hbox{\tiny ls}}=\hU\hU'\hS_{yx}\hS_{xx}^{-1},
\end{equation*}
see for example Equation (2.15) of \cite{reinsel:1998}. Given $\hA_1$, an update of $\hA_2$ can be similarly obtained. We therefore use the alternating least squares to find the minimizer of \eqref{eq:rrmin}.

\subsection{Alternating canonical correlation analysis}
\label{sec:icc}

The classical reduced-rank regression has also been situated under normality, leading to the Gaussian MLE of the coefficient matrix. To introduce the MLE for the RRMAR model, we need to assume that the covariance matrix $\Sigma_e$ of $\vect(\hE_t)$ takes the form of a product
\begin{equation}
\label{eq:ocov}
    \Sigma_e=\Sigma_2\otimes\Sigma_1,
\end{equation}
where $\Sigma_1$ and $\Sigma_2$ are $d_1\times d_1$ and $d_2\times d_2$ positive definite matrices respectively. This is equivalent to assuming $\hE_t = \Sigma_1^{1/2}\h{Z}_t\Sigma_2^{1/2}$, where $\h{Z}$ has iid standard normal entries. This assumption allows us to separate the row and column dependence within the error matrix, with $\Sigma_1$ and $\Sigma_2$ corresponding to the row-wise and column-wise correlations among the entries of $\hE_t$ respectively. This type of covariance model has been proposed and studied in the literature as the ``transposable" \citep{allen2010transposable}, ``array normal" \citep{hoff2011separable}, ``separable" \citep{tsiligkaridis2013covariance,zhou2014gemini}, and ``Kronecker product" \citep{hafner2020estimation,linton2019estimation} covariance structure. We refer the readers to \cite{hoff2011separable} and \cite{linton2019estimation} for a more detailed account on the history of the separable covariance matrix. \cite{chen:2019a} also considered the MAR model under this covariance structure. 

\begin{remark}
{To introduce the Gaussian MLE, we have assumed that $\hZ_t$ in the expression $\hE_t=\Sigma_1^{1/2}\hZ_t\Sigma_2^{1/2}$ has iid standard normal entries. Note that it is also possible to impose additional structures on $\hZ_t$. For example, $\hZ_t$ can have a rank one structure $\hZ_t=\hz_{1t}\hz_{2t}'$ where $\hz_{1t}$ and $\hz_{1t}$ are  random vectors of dimensions $d_1$ and $d_2$, respectively, with all elements independent and of unit-variance. In this case, $\hE_t$ is a rank-one error matrix, with separated row noises and column noises. One difficulty of using this structure is that it implies that 
$\hX_t-\hA_1\hX_{t-1}\hA_2'$ is of rank one for all $t$. Additional error terms may be needed. We leave such an extension to the future research.} 
\end{remark}
Under the normality {and error structure \eqref{eq:ocov}}, the log likelihood of the RRMAR model is (up to some additive constants)
\begin{equation}
\label{eq:loglik}
    -(T-1)(d_2\log|\Sigma_1|+d_1\log|\Sigma_2|)
    -\sum_{t=2}^T\tr\left[ \Sigma_1^{-1}(\hX_t-\hA_1\hX_{t-1}\hA_2')\Sigma_2^{-1}(\hX_t-\hA_1\hX_{t-1}\hA_2')'\right].
\end{equation}
We will introduce an alternating algorithm to find the MLE, maximizing \eqref{eq:loglik} alternatively over one pair $(\hA_i,\Sigma_i)$ while holding the other fixed. As will be seen, each iteration can be viewed as a reduced-rank regression, which is equivalent to the canonical correlation analysis \citep{reinsel:1998}.
We therefore denote the minimizer of \eqref{eq:loglik} by
$\hat\hA_i^{\hbox{\tiny cc}}$ and $\hat\Sigma_i$, and refer to them as the RR.CC estimators.

We now describe how to estimate $\hA_1$ and $\Sigma_1$ when $\hA_2$ and $\Sigma_2$ are known. Under assumption \eqref{eq:ocov}, we can rewrite the model as
\begin{equation*}
    \left(\hX_t\Sigma_2^{-1/2}\right)[,j] = \hA_1 \left(\hX_{t-1}\hA_2'\Sigma_2^{-1/2}\right)[,j] + \left(\hE_t\Sigma_2^{-1/2}\right)[,j].
\end{equation*}
Note that the columns of the transformed error matrix $\hE_t\Sigma_2^{-1/2}$ are iid $N(\hzero,\Sigma_1)$. If we let $\h{y}_{tj}=\left(\hX_t\Sigma_2^{-1/2}\right)[,j]$, \: $\hx_{tj}=\left(\hX_{t-1}\hA_2'\Sigma_2^{-1/2}\right)[,j]$ and $\h{\epsilon}_{tj}=\left(\hE_t\Sigma_2^{-1/2}\right)[,j]$, then the preceding equation can be viewed as a reduced-rank regression with i.i.d. errors:
\begin{equation}
\label{eq:rrr2}
    \h{y}_{tj} = \hA_1\hx_{tj} + \h{\epsilon}_{tj}, \quad 2\leq t\leq T,\;1\leq j\leq d_2.
\end{equation}
The MLE of $\hA_1$ based on \eqref{eq:rrr2} with i.i.d. normal errors has been well studied in the classical reduced-rank regression. Here we only define necessary notations to introduce the final expression of the MLE. We refer the readers to the classical texts \cite{anderson:2013} and \cite{reinsel:1998} for more details.
Let 
\begin{align*}
    \tilde\hS_{xx} & =\sum_t\sum_j \hx_{tj}\hx_{tj}'=\sum_t \hX_{t-1}\hA_2'\Sigma_2^{-1}\hA_2\hX_{t-1}',\\
    \tilde\hS_{yx} & =\sum_t\sum_j \h{y}_{tj}\hx_{tj}'=\sum_t \hX_t\Sigma_2^{-1}\hA_2\hX_{t-1}'.
\end{align*}
The least squares estimator (with no rank constraint) of $\hA_1$ based on \eqref{eq:rrr2} is then given by $\tilde\hA_1 = \tilde\hS_{yx}\tilde\hS_{xx}^{-1}$. To get the MLE of $\hA_1$ under the constraint $\rk(\hA_1)=k_1$, let
\begin{equation*}
    \tilde\Sigma_{\epsilon\epsilon} = \sum_t\sum_j\left(\h{y}_{tj}-\tilde\hA_1\hx_{tj}\right)\left(\h{y}_{tj}-\tilde\hA_1\hx_{tj}\right)' = \sum_t\left(\hX_t-\tilde\hA_1\hX_{t-1}\hA_2'\right)\Sigma_2^{-1}\left(\hX_t-\tilde\hA_1\hX_{t-1}\hA_2'\right)'.
\end{equation*}
Take $\tilde\hU:=[\tilde U_1, \tilde U_2,\ldots, \tilde U_{k_1}]$, where $\tilde U_j$ is the $j$-th leading unit eigenvector of
$\tilde\Sigma_{\epsilon\epsilon}^{-1/2}\tilde\hS_{yx}\tilde\hS_{xx}^{-1}\tilde\hS_{xy}\tilde\Sigma_{\epsilon\epsilon}^{-1/2}$. Then $\hA_1$ is 
updated as
\begin{equation}
\label{eq:mle-A1}
  \check \hA_1^{\hbox{\tiny cc}}=\tilde\Sigma_{\epsilon\epsilon}^{1/2} \tilde \hU \tilde \hU'\tilde\Sigma_{\epsilon\epsilon}^{-1/2}\tilde\hS_{yx}\tilde\hS_{xx}^{-1}.
\end{equation}
For the derivation of this update, see for example Equation (2.15) of \cite{reinsel:1998}.
Subsequently, the covariance matrix $\Sigma_1$ is updated as
\begin{equation*}
    \check\Sigma_1 = \frac{1}{T-1} \sum_t\left(\hX_t-\check\hA_1^{\hbox{\tiny cc}}\hX_{t-1}\hA_2'\right)\Sigma_2^{-1}\left(\hX_t-\check\hA_1^{\hbox{\tiny cc}}\hX_{t-1}\hA_2'\right)'.
\end{equation*}
Given $\hA_1$ and $\Sigma_1$, an update of $\hA_2$ and $\Sigma_2$ can be similarly obtained. Therefore, we use the alternating algorithm to find the minimizer of \eqref{eq:loglik}.

\begin{remark}
Note that each step of the algorithms reduces the corresponding objective functions, hence the algorithms are likely to converge to a local/global minimum.
To ensure the estimators do not oscillating among equivalent solutions, proper constraints need to be enforced. We have assumed that $\hA_1$ is always normalized so that $\|\hA_1\|_F=1$, and will require the same for its estimates in the algorithm. Similarly, for the MLE we require  $\|\Sigma_1\|_F=1$ and $\|\hat\Sigma_1\|_F=1$. On the other hand, since the noise covariance matrix is full rank and the ranks $k_1$ and $k_2$ are the correct ranks, there is no identifiability issue related to rank deficiency.
\end{remark}

\begin{remark}
Since the objective functions are not convex, care 
needs to be taken to reach the global minimum. Using multiple random initial values is one possible approach.
Another approach is to use
the projection estimator of $\hA_1$ and $\hA_2$ in \cite{chen2021autoregressive} (by ignoring the rank constraints), as the initial values of both alternating algorithms. The projection estimator of $\hA_1$ and $\hA_2$ are derived from a global (one-step) estimation of $\hA_2\otimes\hA_1$, though they are not efficient.
\end{remark}

\section{Asymptotics}
\label{sec:asymp}

The asymptotic analysis is substantially different from the classical
reduced-rank regression, due to the alternating nature of the
estimation. For example, the gradient condition for the LSE $\hat\hA_1^{\hbox{\tiny ls}}$ is
$\hat\hA_1^{\hbox{\tiny
    ls}}=\hat\hU\hat\hU'\hat\hS_{yx}\hat\hS_{xx}^{-1}$, where
$\hat \hU$, $\hat\hS_{xx}$ and $\hat\hS_{yx}$ are defined as the
$\hU$, $\hS_{xx}$ and $\hS_{yx}$ in Section~\ref{sec:est} with the
modification that all $\hA_2$ therein need to be replaced by
$\hat\hA_2^{\hbox{\tiny ls}}$. In other words, the asymptotic
behaviors of $\hat\hA_1^{\hbox{\tiny ls}}$ and
$\hat\hA_2^{\hbox{\tiny ls}}$ are intertwined. 

The asymptotics for $\hA_i^{\hbox{\tiny ls}}$ and $\hA_i^{\hbox{\tiny cc}}$ involve heavy notations. First of all,
recall that we assume $\|\hA_1\|_F=1$ for the parameter identifiability, so we rescale
$\hat\hA_i^{\hbox{\tiny ls}}$ and $\hat\hA_i^{\hbox{\tiny cc}}$ so that
$\|\hat\hA_1^{\hbox{\tiny ls}}\|_F=1$ and $\|\hat\hA_1^{\hbox{\tiny cc}}\|_F=1$. 
In the table below, we list notations that appear in both Theorem~\ref{thm:ls} and Theorem~\ref{thm:cc}, but with different definitions in these theorems.
\begin{center}
\begin{tabular}{l|l|l}
\hline
    Notations & Theorem~\ref{thm:ls} & Theorem~\ref{thm:cc}  \\\hline
    $\Gamma_1$ & $\E(\hX_t\hA_2'\hA_2\hX_t')$ & $\E(\hX_t\hA_2'\Sigma_2^{-1}\hA_2\hX_t')$ \\
    $\Gamma_2$ & $\E(\hX_t'\hA_1'\hA_1\hX_t)$ & $\E(\hX_t'\hA_1'\Sigma_1^{-1}\hA_1\hX_t)$\\
    $\mathbb{P}_i$ & \hbox{orthogonal projection to} $\mathrm{col}(\hA_i)$ & \hbox{orthogonal projection to} $\mathrm{col}(\Sigma_i^{-1/2}\hA_i)$\\
    $\PPP_i$ & $\mathbb{P}_i$ & $\Sigma_i^{-1/2}\mathbb{P}_i\Sigma_i^{1/2}$ \\
    $\halpha_1$ & $\vect(\hA_1)$ & Same \\
    $\hgamma_1$ & $(\halpha_1',\hzero')'$ & Same \\
    $\hW_t$ & $[(\hA_2\hX_t')\otimes\hI_{d_1},\hI_{d_2}\otimes(\hA_1\hX_t)]'$ & {Same} \\
    $\hH$ & $\E(\hW_t\hW_t')+\hgamma_1\hgamma_1'$ & $\E(\hW_t\Sigma_e^{-1}\hW_t')+\hgamma_1\hgamma_1'$ \\\hline
\end{tabular}
\end{center}
Define
\begin{equation*}
  \hQ_t:=\begin{pmatrix}
          \hX_t\hA_2'\otimes \PPP_1+[\Gamma_1\hA_1'(\hA_1\Gamma_1\hA_1')^+\hA_1\hX_t\hA_2']\otimes(\hI-\PPP_1) \\
          \PPP_2\otimes \hX_t'\hA_1'+(\hI-\PPP_2)\otimes [\Gamma_2\hA_2'(\hA_2\Gamma_2\hA_2')^+\hA_2\hX_t'\hA_1']
        \end{pmatrix},
\end{equation*}
where $\hM^+$ denotes the Moore-Penrose inverse of $\hM$. Note that $\hQ_t$ will appear in both Theorem~\ref{thm:ls} and Theorem~\ref{thm:cc}. Although it seems to have the same definition in both theorems, the two versions actually differ because the $\Gamma_i$ and $\PPP_i$ involved have different definitions. 

\begin{theorem}
\label{thm:ls}
  Assume that $\{\hE_t\}$ are i.i.d. with mean zero and finite second
  moments. Assume that
  $0<\rk{\hA_i}=k_i\leq d_i$, $\rho(\hA_1)\rho(\hA_2)<1$, where $\rho(\cdot)$ denotes the spectral radius of a matrix, and
  $\Sigma_e$ is non-singular. 
  Then
  \begin{equation*}
    \sqrt{T}\begin{pmatrix}
      \vect\left[\hat\hA_1^{\hbox{\tiny ls}}-\hA_1\right] \\
      \vect\left[(\hat\hA_2^{\hbox{\tiny ls}})'-\hA_2'\right]
    \end{pmatrix}
    \Rightarrow N(\hzero, \Xi^{\hbox{\tiny ls}}),
  \end{equation*}
  where
  \begin{equation}
  \label{eq:cov_ls}
    \Xi^{\hbox{\tiny ls}} :=\hH^{-1}\E(\hQ_t\Sigma_e\hQ_t')\hH^{-1}.
  \end{equation}
\end{theorem}

\begin{theorem}
\label{thm:cc}
  Assume that $\{\hE_t\}$ are i.i.d. with mean zero and finite second
  moments. Assume that
  $0<\rk{\hA_i}=k_i\leq d_i$, $\rho(\hA_1)\rho(\hA_2)<1$, and
  $\Sigma_e$ is of the form \eqref{eq:ocov}, and is non-singular. 
  Then
  \begin{equation*}
    \sqrt{T}\begin{pmatrix}
      \vect\left[\hat\hA_1^{\hbox{\tiny cc}}-\hA_1\right] \\
      \vect\left[(\hat\hA_2^{\hbox{\tiny cc}})'-\hA_2'\right]
    \end{pmatrix}
    \Rightarrow N(\hzero, \Xi^{\hbox{\tiny cc}}),
  \end{equation*}
  where
  \begin{equation}
  \label{eq:cov_cc}
    \Xi^{\hbox{\tiny cc}} :=\hH^{-1}\E(\hQ_t\Sigma_e^{-1}\hQ_t')\hH^{-1}.
  \end{equation}
\end{theorem}

Besides the usual regularity conditions, we also assume in Theorem~\ref{thm:ls} and Theorem~\ref{thm:cc} that certain matrices have distinct eigenvalues. While this assumption holds for generic positive definite matrices, it is even easier to be fulfilled by the aforementioned matrices since they involve $\Gamma_i$.

As mentioned before, the RRMAR model is a MAR model with additional low rank constraints on the coefficient matrices $\hA_i$. If the estimation is carried out without the rank constraints, then the procedure of \cite{chen2021autoregressive} applies and so does its asymptotic result ({e.g.} Theorem~4 {in \cite{chen2021autoregressive}}). In fact, if the estimation of $\hA_i$ is given by the MLE without imposing the low rank constraints, the asymptotic covariance matrix would take the same form as $\Xi^{\hbox{\tiny cc}}$, by setting $\mathcal P_i=\hI$ in the definition of $\hQ_t$. Denote this covariance matrix by $\tilde \Xi$. The following theorem asserts that the MLE $\hat\hA_i^{\hbox{\tiny cc}}$ under the RRMAR model are asymptotically more efficient.
\begin{theorem}
\label{thm:efficiency}
Under the assumptions of Theorem~\ref{thm:cc}, it holds that $\tilde\Xi\succeq \Xi^{\hbox{\tiny cc}}$, i.e. the difference $\tilde\Xi - \Xi^{\hbox{\tiny cc}}$ is positive semi-definite.
\end{theorem}

\begin{remark}
Since $\Sigma_e$ in Theorem~1 can be arbitrary, as long as it is non-singular, the LSE does not correspond to the MLE, and thus there is no similar result to Theorem~\ref{thm:efficiency} regarding the comparison of the LSE under the RRMAR model and the MAR model without rank constraints. On the other hand, if the entries of $\hE_t$ are IID, we can show a similar result to Theorem~\ref{thm:efficiency} for the LSE, which is not presented here since it is too special. 
\end{remark}

\newcommand{\hD}{\h{D}}
\newcommand{\hd}{\h{d}}
\newcommand{\hone}{\h{1}}

\medskip
We now consider the asymptotics 
of the composition and loading matrices $\hA_{il}$ and $\hA_{ic}$ in \eqref{eq:marrr_fac}. The following discussion works the same for either $\hat\hA_i^{\hbox{\tiny ls}}$ or $\hat\hA_i^{\hbox{\tiny cc}}$. Therefore, we will use the unified notations $\hat\hA_i$ and $\Xi$, dropping the superscripts $^{\hbox{\tiny ls}}$ and $^{\hbox{\tiny cc}}$. Since $\hA_{ic}$ and $\hA_{il}$ cannot be identified as seen from $\hA_{il}\hA_{ic}'=\hA_{il}\hM\hM^{-1}\hA_{ic}'$ for any invertible $k_i\times k_i$ matrix $\hM$, we consider instead the singular value decomposition (SVD) of $\hA_i$.  Write $\hA_{i}={\hU}_i\hD_i\hV_i'$, where both $\hU_i$ and $\hV_i$ are $d_i\times k_i$ ortho-normal matrices. Denote the $j$-th diagonal element of $\hD_i$ by $d_{ij}$, and define $\hd_i=(d_{i1},\ldots,d_{i,k_i})'$. Comparing \eqref{eq:marrr_fac}, we see that $\hV_i$ corresponds to the composition matrix $\hA_{ic}$, $\hU_i$ corresponds to the loading matrix $\hA_{il}$, and $\hD_i$ can be absorbed into either $\hA_{il}$ or $\hA_{ic}$. 
Let $\hat\hA_i =\hat\hU_i\hat\hD_i(\hat\hV_i)'$ be the SVD of $\hat\hA_i$. Since $\hat\hA_i\hat\hA_i' = \hat\hU_i\hat\hD_i^2(\hat\hU_i)'$, the asymptotic distribution of $\hat\hU_i$ can be obtained based on that of $\hat\hA_i\hat\hA_i'$. Similarly, the asymptotic distribution of $\hat\hV_i$ can be derived from that of $\hat\hA_i'\hat\hA_i$.
Note that the asymptotic covariance matrix of $\vect(\hat\hA_i)$ (for $i=1,2$) is a submatrix of $\Xi$ and can be extracted from \eqref{eq:cov_ls} or \eqref{eq:cov_cc}. Following that, we let $\Xi_{i1}$ be the asymptotic covariance matrix of $\vect(\hat\hA_i\hat\hA_i')$, which can be obtained through the expansion $$\hat\hA_i\hat\hA_i' = \hA_i\hA_i'+(\hat\hA_i-\hA_i)\hA_i'+\hA_i(\hat\hA_i-\hA_i)'+o_P(T^{-1/2}).$$  More specifically, when $i=1$,
\begin{equation*}
    \Xi_{11}=\left[\hA_1\otimes\hI_{d_1}+(\hI_{d_1}\otimes\hA_1)\hJ_{d_1,d_1}\right]\left\{ \Xi[1:d_1^2,1:d_1^2] \right\}\left[\hA_1\otimes\hI_{d_1}+(\hI_{d_1}\otimes\hA_1)\hJ_{d_1,d_1}\right]',
\end{equation*}
where $\Xi[1:d_1^2,1:d_1^2]$ is the upper left $d_1^2\times d_1^2$ block of $\Xi^{\hbox{\tiny ls}}$ or $\Xi^{\hbox{\tiny cc}}$, and the matrix $\hJ_{d_1,d_1}$ is defined in \eqref{eq:J}.
The asymptotic covariance matrix of $\hat\hA_1'\hat\hA_1$, denoted by $\Xi_{12}$, has a similar expression. When $i=2$, the matrices $\Xi_{21}$ and $\Xi_{22}$, related to $\hat\hA_2$, are also defined similarly.

Define the matrix $\hR_{i1}$ as
\begin{equation*}
    \hR_{i1}=(\hI_{k_i}\otimes\hU_i,\hI_{k_i}\otimes \hU_i^{\perp})
    \begin{pmatrix}
    (\hD_i^2\otimes\hI_{k_i}-\hI_{k_i}\otimes\hD_i^2+\hL_{k_i}\hL_{k_i}')^{-1}(\hI_{k_i^2}-\hL_{k_i}\hL_{k_i}')(\hU_i'\otimes\hU_i') \\
    (\hD_i^{-2}\hU_i')\otimes(\hU_i^\perp)'
    \end{pmatrix},
\end{equation*}
and define $\hR_{i2}$ similarly, but replacing $\hU_i$ with $\hV_i$.

Note that even when $\hA_i$ has distinct singular values, the columns of $\hU_i$ and $\hV_i$ are only identified up to sign changes. We shall adopt the following convection to identify $\hU_i$: the first nonzero element of each column of $\hU_i$ is positive. Since $\hA_1$ and $\hA_2$ are also only identified up to sign changes, we make one more requirement to identify $\hV_i$: the first nonzero element of the first column of $\hV_2$ is positive. Subsequently, we also require the estimators $\hat\hU_i$ and $\hat\hV_i$ to satisfy these identifiability conditions.

Now, as a consequence of Theorem~\ref{thm:ls} and Theorem~\ref{thm:cc}, we have the following result regarding $\hat\hU_i$ and $\hat\hV_i$. 
\begin{corollary}
\label{cor:svd}
Assume the conditions of Theorem~\ref{thm:ls} or Theorem~\ref{thm:cc} hold, and that the singular values of $\hA_1$ are distinct, and so are those of $\hA_2$. For each of $i=1,2$, it holds that
\begin{equation*}
   \sqrt{T}\vect(\hat\hU_i-\hU_i)\Rightarrow
   N\left(\hzero, \hR_{i1} \Xi_{i1} \hR_{i1}'\right),
\end{equation*}
and
\begin{equation*}
   \sqrt{T}\vect(\hat\hV_i-\hV_i)\Rightarrow
   N\left(\hzero, \hR_{i2} \Xi_{i2} \hR_{i2}'\right),
\end{equation*}
and
\begin{equation*}
    \sqrt{T}(\hat\hd_{i} -\hd_i) \Rightarrow N\left(\hzero,\frac{1}{4}\hD_i^{-1}\hL_{k_i}'(\hU_i'\otimes\hU_i')\Xi_{i1}(\hU_i\otimes\hU_i)\hL_{k_i}\hD_i^{-1}\right).
\end{equation*}
\end{corollary}

\begin{remark}
Let $\hG_i$ be the $k_i\times k_i$ matrix defined as
\begin{equation*}
    \hG_i[j,k] = \left\{\begin{array}{ll}
    (d^2_{ik}-d^2_{ij})^{-1}     & \hbox{when } j\neq k,  \\
    0     & \hbox{when } j=k.
    \end{array}\right.
\end{equation*} 
Furthermore, let $\hU_i^\perp$ be the $d_i\times (d_i-k_i)$ orthonormal matrix such that $(\hU_i,\hU_i^\perp)$ is an orthogonal matrix. The asymptotic distribution of $\hat\hU_i$ can also be obtained from the following equation
\begin{equation*}
    \hat\hU_i-\hU_i = (\hU_i,\hU_i^{\perp})
    \begin{pmatrix}
    \hG_i\circ \left[\hU_i'(\hat\hA_i\hat\hA_i' - \hA_i\hA_i')\hU_i\right] \\
    (\hU_i^{\perp})'(\hat\hA_i\hat\hA_i' - \hA_i\hA_i')\hU_i\hD_i^{-2}
    \end{pmatrix}
    + o_P\left(\frac{1}{\sqrt{T}}\right),
\end{equation*}
where $\circ$ denotes the entry-wise Hadamard or Schur product of two matrices.
\end{remark}

\begin{remark}
The joint distribution of $\hat\hU_i$ and $\hat\hV_i$ can also be derived. However, we choose not to spell the details out here for two reasons: the notations are already very complicated, and more importantly, there does not seem to be any direct applications of such a joint distribution.
\end{remark}

\section{Identification of the Rank}
\label{sec:rank}

In most applications the ranks of $\hA_i$ are unknown, and it is
important to determine them
from the data. This problem has been considered for multivariate reduced-rank regression by \cite{anderson:1951} and \cite{anderson:2013}, and for reduced-rank autoregressive model by \cite{kohn1979asymptotic}, \cite{reinsel:1998}, \cite{tiao1989model} and \cite{tsay1985use}, among others. For high dimensional reduced-rank regression based on independent samples, penalized least squares can select the ranks along with the estimation, where the penalty is based on nuclear norm \citep{yuan:2007,negahban:2011}, $\ell_0$ norm
\citep{bunea:2011}, or Schatten-$q$ quasi-norm
\citep{rohde:2011}. \cite{basu2019low} and \cite{lin2020regularized} considered low rank VAR and tensor models,
combining least squares estimation with the nuclear norm
penalty.

We propose to use an information criterion to select the ranks. For a given pair of ranks $(r_1,r_2)$, it is defined as
\begin{equation}
\label{eq:bic}
\begin{aligned}
    \mathrm{EBIC}(r_1,r_2) = & \log\left[\frac{1}{Td_1d_2}\sum_{t=2}^T\|\hX_t-\hat\hA_1^{\hbox{\tiny ls}}\hX_{t-1}(\hat\hA_2^{\hbox{\tiny ls}})'\|_F^2\right] \\ 
    & + \frac{1}{Td_1d_2}\cdot[\log(Td_2)\cdot r_1(2d_1-r_1)+\log(Td_1)\cdot r_2(2d_2-r_2)].
\end{aligned}
\end{equation}
This can be viewed as an extended version of the Bayesian Information Criterion \citep{schwarz1978estimating}, so we use the acronym EBIC. Here the likelihood is calculated for the model where the entries of $\h{E}_t$ are iid $N(0,\sigma^2)$, so it is best viewed as a “quasi”-likelihood. It is not precisely derived according to the posterior probability under the Bayesian framework \citep{haughton1988choice}. Instead, we combine the quasi-log-likelihood and a penalty term where the numbers of parameters are multiplied by the logarithm of the sample sizes. Instead of simply counting the number of parameters, the effective number of parameters \citep{mukherjee2015degrees,yuan2016degrees} can also be used in the EBIC. However, we choose the current version for simplicity. The selected pair of ranks $(\hat k_1,\hat k_2)$ minimizes the EBIC over all the pairs $(r_1,r_2)$ such that $1\leq r_1\leq r_{1\!\max}$ and $1\leq r_2\leq r_{2\!\max}$, where $r_{1\!\max}$ and $r_{2\!\max}$ are pre-determined maximum ranks of $\hA_1$ and $\hA_2$, respectively. If no information is available and the dimensions $d_1$ and $d_2$ are not too large, we can simply use $r_{1\!\max}=d_1$ and $r_{2\!\max}=d_2$. The following Theorem \ref{thm:bic} confirms that the EBIC in \eqref{eq:bic} does achieve the consistency. Its empirical performances are also outstanding, as will be shown in Section~\ref{sec:simulation}. 

Since there are two ranks to be determined, a direct search via \eqref{eq:bic} over all possible pairs of ranks can be very costly when both $r_{1\!\max}$ and $r_{2\!\max}$ are large. We also consider selecting these two ranks separately. Specifically, the selected ranks are
\[
\hat{r}_1=\arg\min_{r_1}\mathrm{EBIC}(r_1,r_{2\!\max})\quad\quad \hat{r}_2=\arg\min_{r_2}\mathrm{EBIC}(r_{1\!\max},r_{2}).
\]

\begin{theorem}\label{thm:bic}
  Assume that $\{\hE_t\}$ are i.i.d. with mean zero and finite second moments. 
  Also assume that $0<\rk{\hA_i}=k_i\leq d_i$, $\rho(\hA_1)\rho(\hA_2)<1$, and
  $\Sigma_e$ is non-singular. Then both the joint $\mathrm{EBIC}$ and the separate $\mathrm{EBIC}_i$ select the true ranks consistently, given that $k_i\leq r_{i\!\max}$.
\end{theorem}

\begin{remark}
{Theorem~\ref{thm:bic} continues to hold if the penalty {term} in EBIC \eqref{eq:bic} {(the second term)} is scaled by any positive constant $c$. {The finite sample performance of EBIC may depend on such a constant.} The resampling based tuning method introduced by \cite{hallin2007} can be {adopted} 
here to determine $c$. The benefit of using a data-driven penalty can be significant when the dimension is high.}
\end{remark}

\begin{remark}
{If the second term in \eqref{eq:bic} is replaced by {$2(Td_1d_2)^{-1}[r_1(2d_1-r_1)+r_2(2d_2-r_2)]$}, the criterion is similar to AIC. Although it may not be consistent, it often works well in small sample \citep{shao1997asymptotic,brockwell:1991}. } 
\end{remark}

\section{Numerical Studies}
\label{sec:num}

\subsection{Simulations}
\label{sec:simulation}

In this section, we investigate the finite sample performance of the proposed estimation and rank determination procedures for the RRMAR models under various simulation setups. The simulation study consists of three parts. The first part is designed to compare the empirical behavior of the proposed alternating least square estimator $\hat\hA_i^{\hbox{\tiny ls}}$, labelled as RR.LS in the figures, and the alternating MLE $\hat\hA_i^{\hbox{\tiny cc}}$, labelled as RR.CC. The true ranks are taken as known. The least squares estimator without rank constraints (labelled as LSE) in \cite{chen2021autoregressive} is also included as a benchmark for comparison. In the second part, we report the coverage probabilities of the confidence intervals constructed based on Theorem~\ref{thm:ls}, Theorem~\ref{thm:cc} and Corollary~\ref{cor:svd}. The third part examines the rank determination based on the EBIC proposed in Section~\ref{sec:rank}. We also experiment with rank selection by rolling forecasting.

For given dimensions $d_i$ and ranks $k_i$, the observed
data $\hX_t$ are simulated according to model \eqref{eq:marrr}. The matrix $\hA_1$ is generated according to $\hA_1=\hQ_1\Lambda\hQ_2'$, where the entries of the $k_1\times k_2$ diagonal matrix $\Lambda$ are sampled from the uniform distribution over the interval $[0.5,1.5]$, and the $d_1\times k_1$ orthonormal matrices $\hQ_1$ and $\hQ_2$ are generated randomly from the Haar distribution. 
The matrix $\hA_2$ is generated in the same way. The two matrices $\hA_1$ and $\hA_2$ are then rescaled so that $\rho:=\rho(\hA_1)\rho(\hA_2)<1$ and $||\hA_1||_F=1$. 
Throughout all the simulation studies, two different settings of the covariance structure of the innovation matrix $\hE_t$ are considered:
\begin{enumerate}
\item[(I)] The covariance matrix $\Sigma_e=\Cov(\vect(\hE_t))$ is randomly generated according to $\Sigma_e= \h{Q}\Lambda\h{Q}'$, where the entries of the diagonal matrix $\Lambda$ are equally spaced over $[1,10]$, and $\h{Q}$ is a random orthogonal matrix generated from the Haar distribution.
\item [(II)] The covariance matrix $\Sigma_e$ takes the form \eqref{eq:ocov}, where each of $\Sigma_1$ and $\Sigma_2$ is generated in the same way as the $\Sigma_e$ in Setting I, except that the diagonal entries of $\Lambda$ are equally spaced over $[1,5]$.
\end{enumerate}
For a particular simulation setting with multiple repetitions, the
matrices $\hA_i$ and $\Sigma_e$ are fixed.

\begin{remark}
We randomize the population parameters in the simulation, while controlling the key parameters (e.g. 
$\rho(\hA_1)\rho(\hA_2)$ and the entries of the diagonal matrix $\Lambda$ in the noise covariance structure). The main reason is that the theoretical results show that the estimation performance mainly depend on these controlled parameters. Due to the large number of ``free parameters", it is difficult to construct a specific 
design of the parameter set, and the individual parameter is of less importance and interests. Also one won't be able to cover all possible combinations. 
In a way, the randomly generated population parameter set can be viewed as a ``representative" set.
\end{remark}

In the first experiment, for each configuration of sample size $T$, dimensions $d_i$ and ranks $k_i$, we repeat the simulation 100 times, and show a box plot of the estimation error
\begin{equation*}
  \log(\|\hat\hA_2\otimes\hat\hA_1-\hA_2\otimes\hA_1\|_F^2).
\end{equation*}
The spectral radius is fixed at $\rho=.75$. Figure~\ref{fig:simu1} and Figure~\ref{fig:simu2} use the {boxplots\footnote{{Each boxplot shows the distribution of the log errors from 100 repeated experiments in each setting. The box in the middle shows the range of $1$-st quartile to $3$-rd quantile, and the stems extends to the minimum of 1.5 interquartile range (IQR) and the maximum/minimum of the data. The (rare) dots show the observations outside the 1.5 IRQ range (often considered as outliers).}}} to compare LSE, RR.LS and RR.CC, under Settings (I) and (II) respectively.  It is seen from both figures that the advantage of RR.LS and RR.CC over LSE gets bigger as the dimensions grow higher. On the other hand, for fixed dimensions, this advantage becomes smaller as the ranks increase. 
From Figure~\ref{fig:simu1} we also find that, under Setting (I) of the error covariance matrix, RR.CC performs similarly as RR.LS does, even though the covariance matrix $\Sigma_e$ 
does not have the form \eqref{eq:ocov}, which is assumed for RR.CC. On the other hand, Figure~\ref{fig:simu2} clearly demonstrates the advantage of RR.CC over RR.LS under setting (II), when $\Sigma_e$ does bear the form \eqref{eq:ocov}. {In addition, it is seen from Figure 2 that the differences (in log error) between LSE and RR.LS/RR.CC remain roughly the same for difference sample sizes. This confirms the results in Theorem 3, which states that the gain in efficiency (in terms of reduction in variance of the estimators) is proportional to $T$, hence the absolute difference between the full model in Chen, Xiao and Yang (2019) and the reduced-rank model becomes smaller as $T$ increases, but their ratio or log difference remain roughly the same.}

\begin{figure}[h]
    \centering
    \includegraphics[width=2.1in]{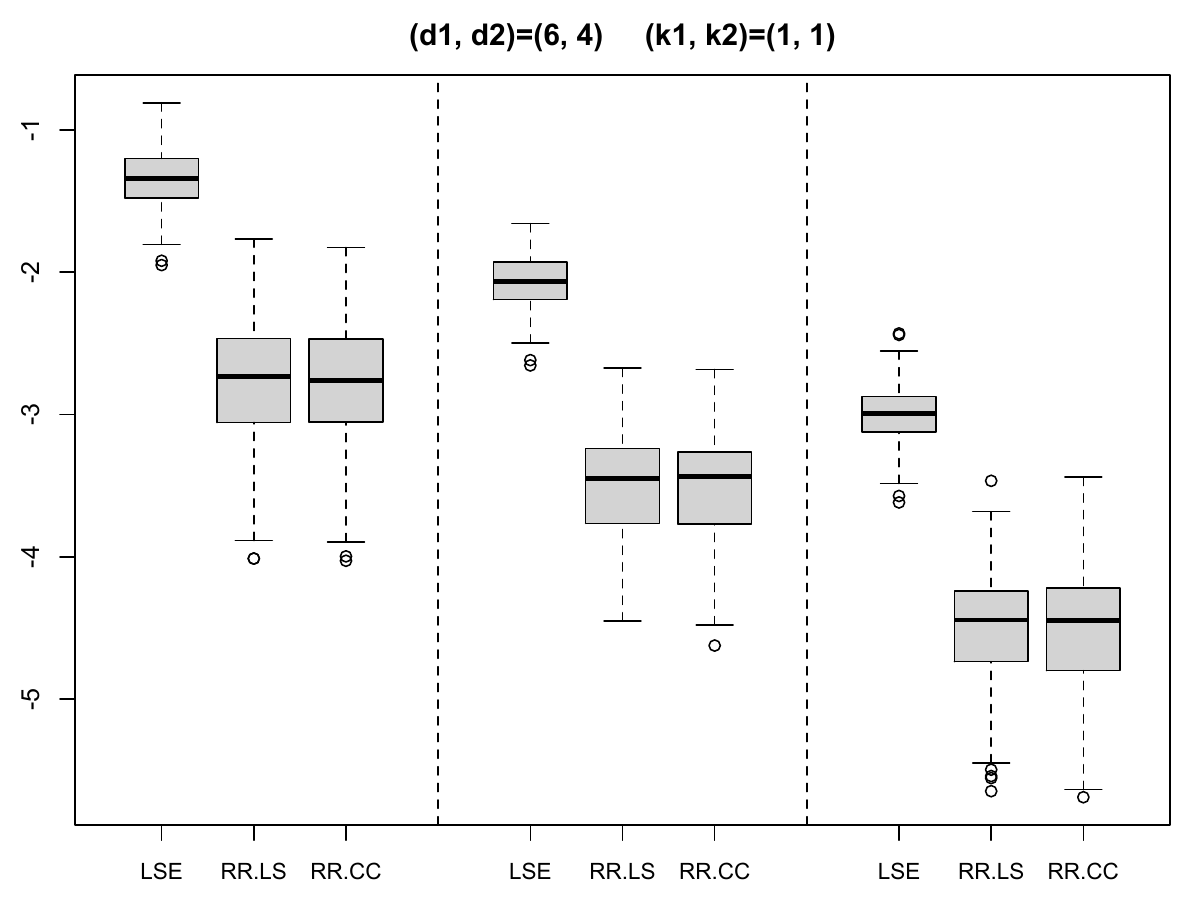}~~\includegraphics[width=2.1in]{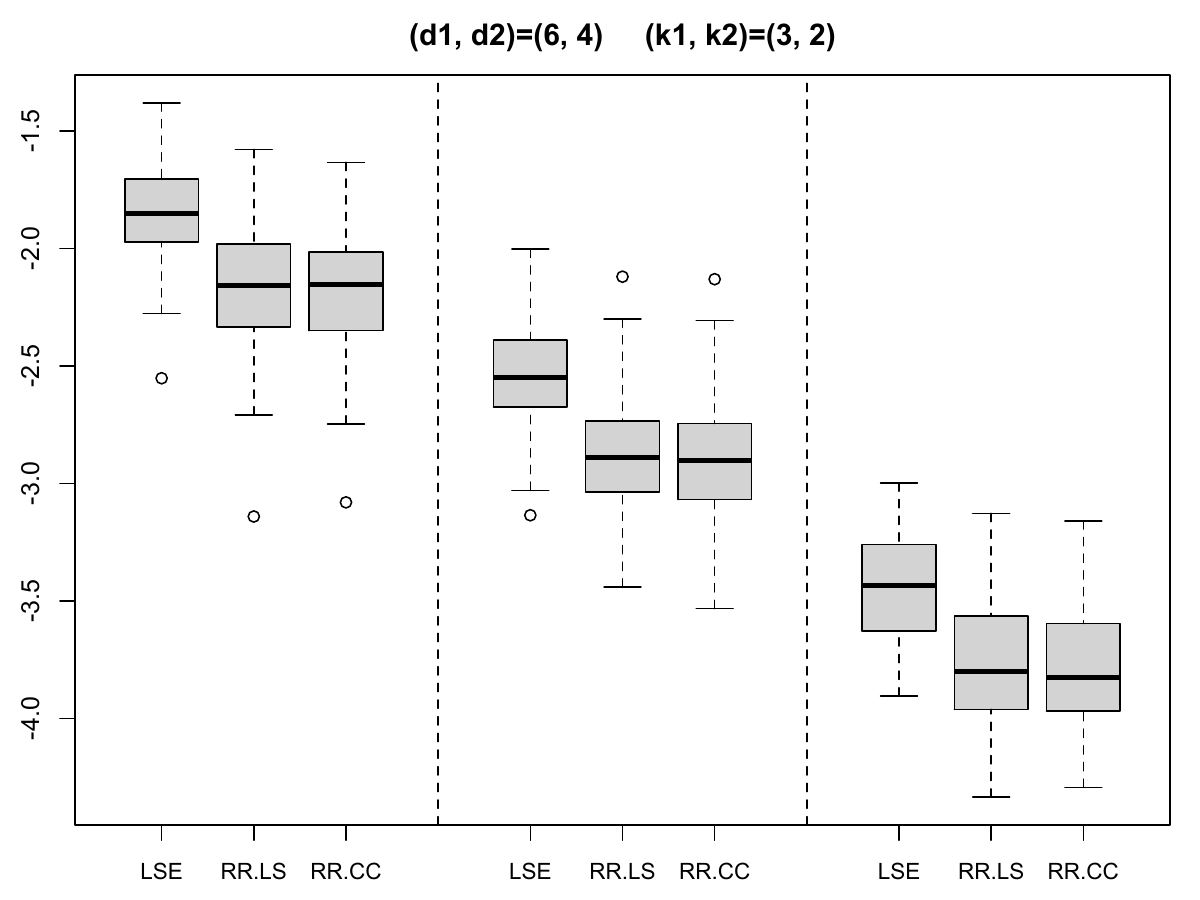}~~\includegraphics[width=2.1in]{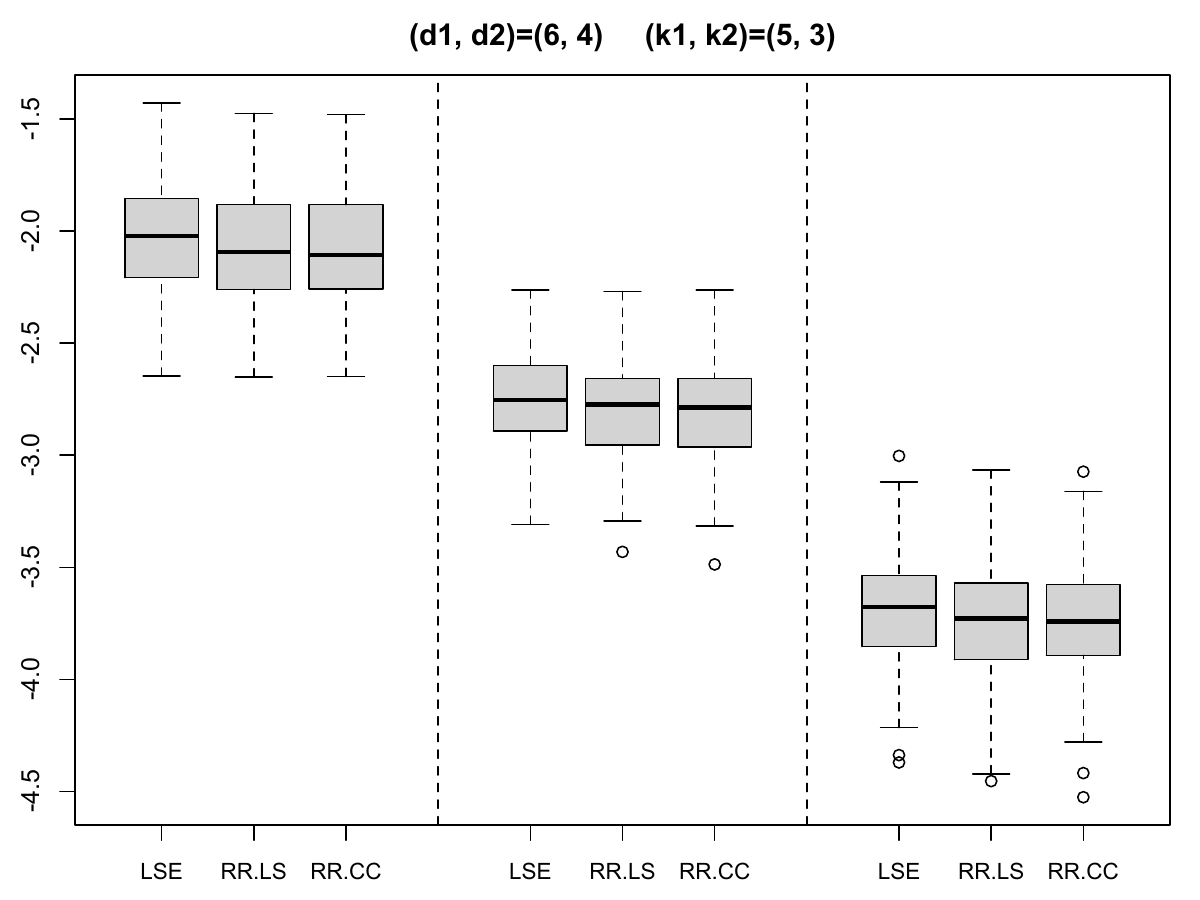}\\
    \includegraphics[width=2.1in]{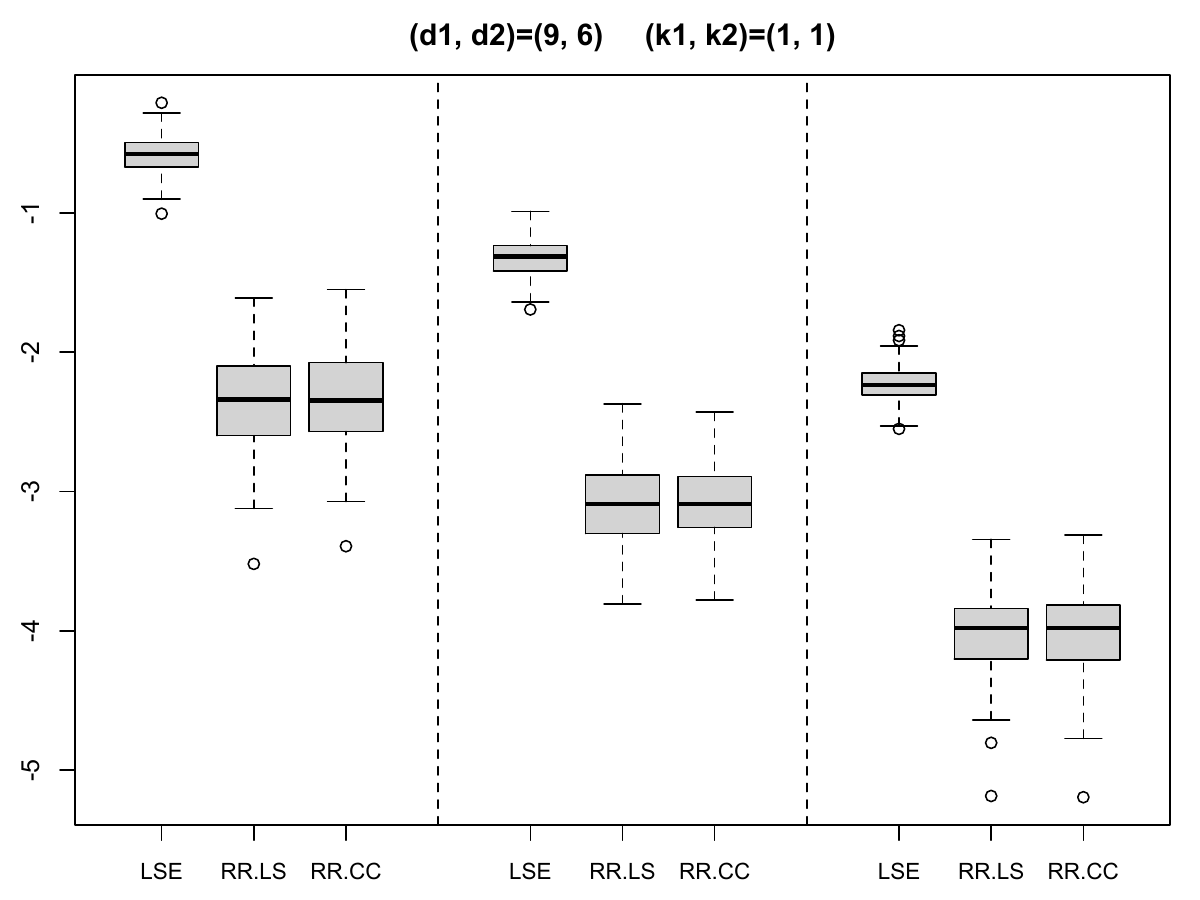}~~\includegraphics[width=2.1in]{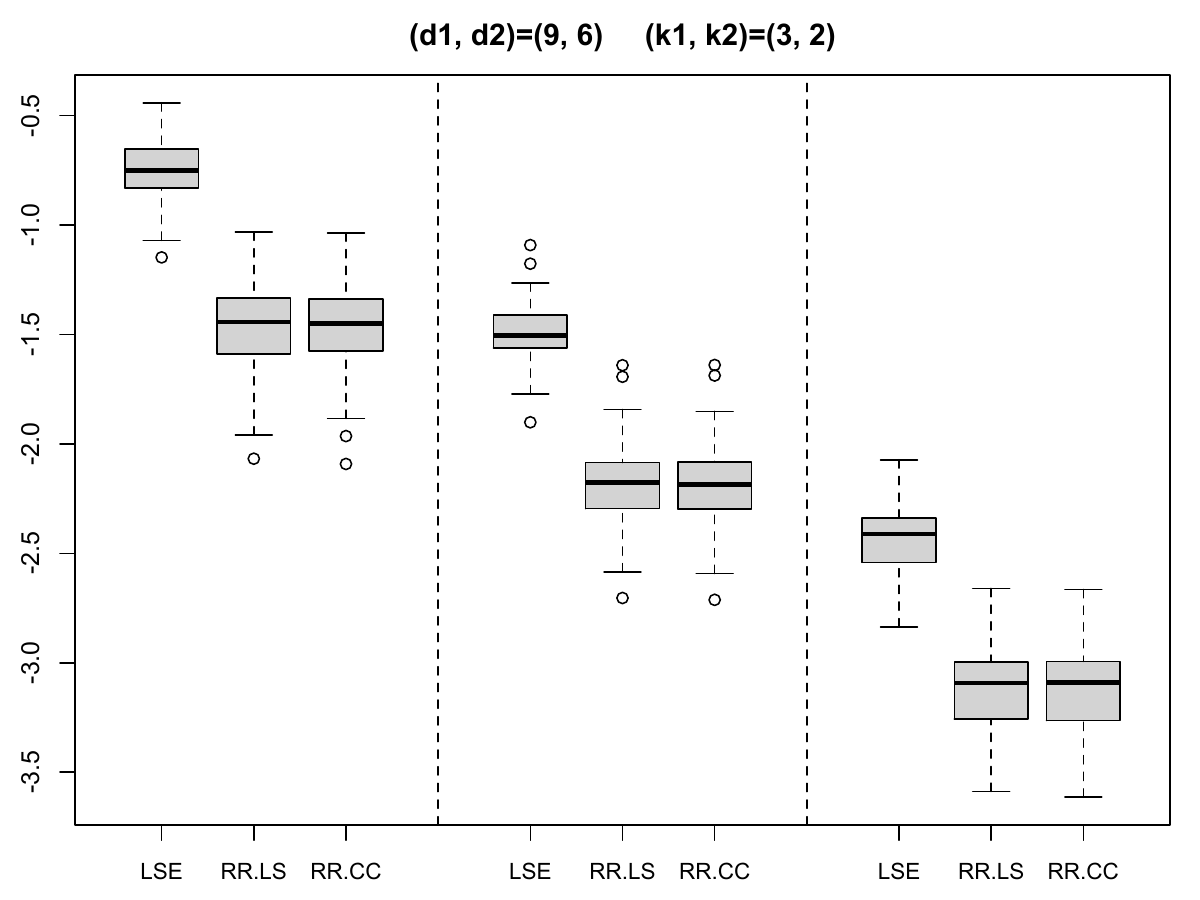}~~\includegraphics[width=2.1in]{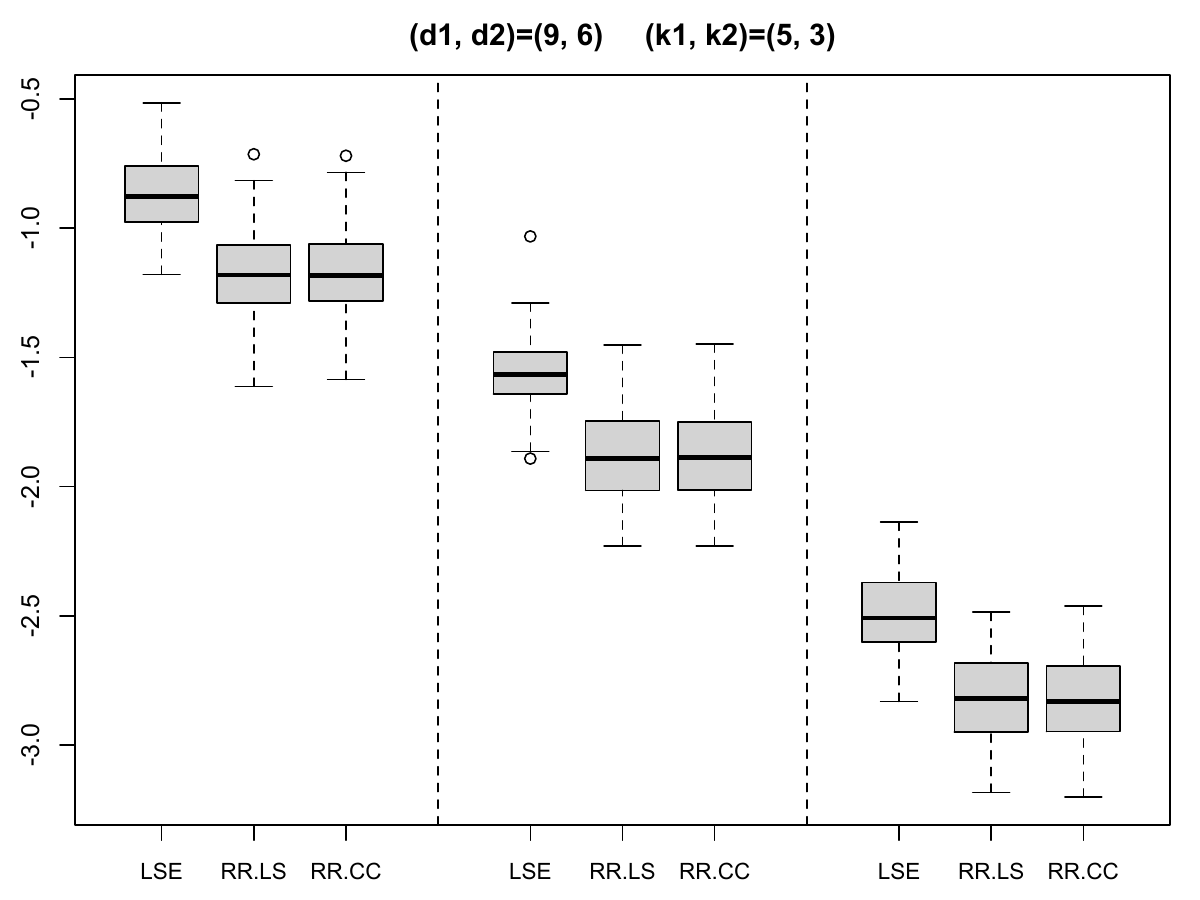}\\
    \includegraphics[width=2.1in]{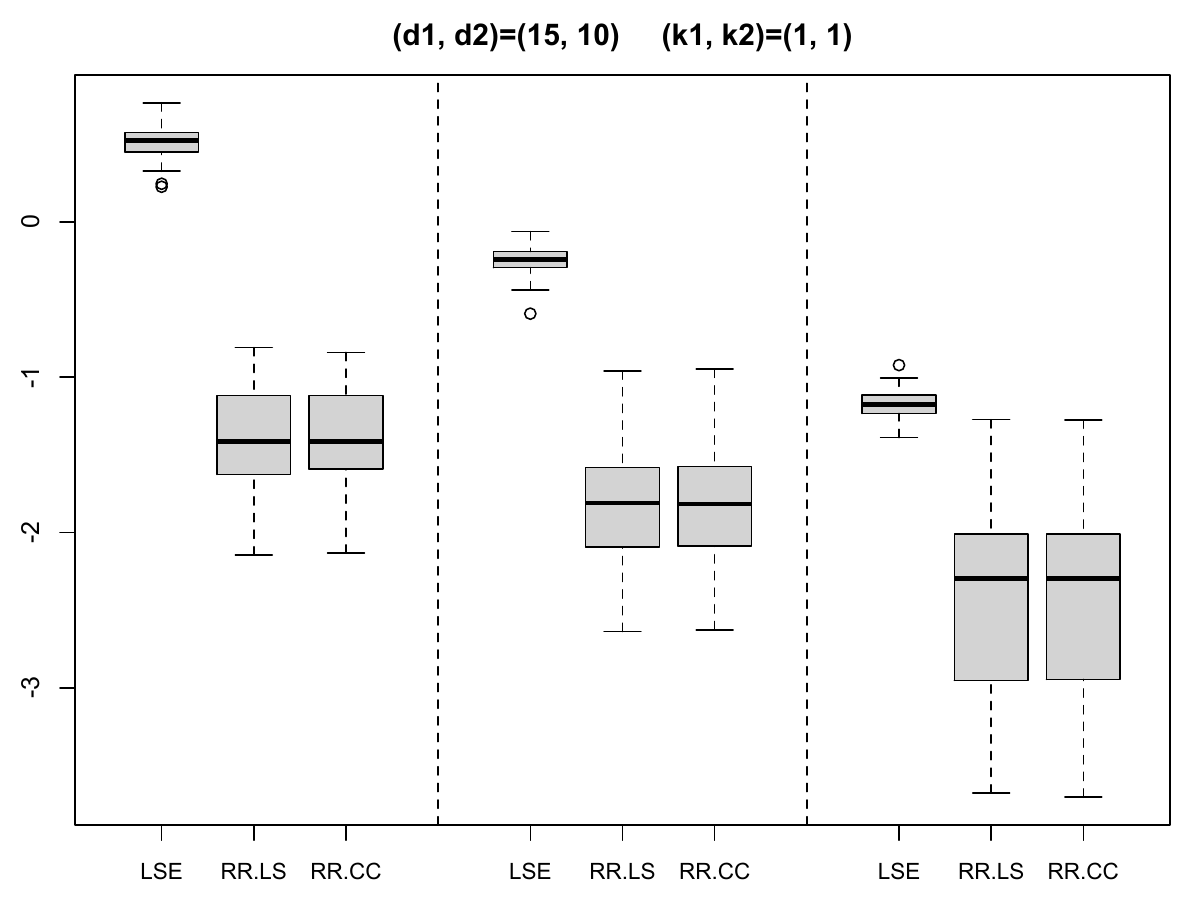}~~\includegraphics[width=2.1in]{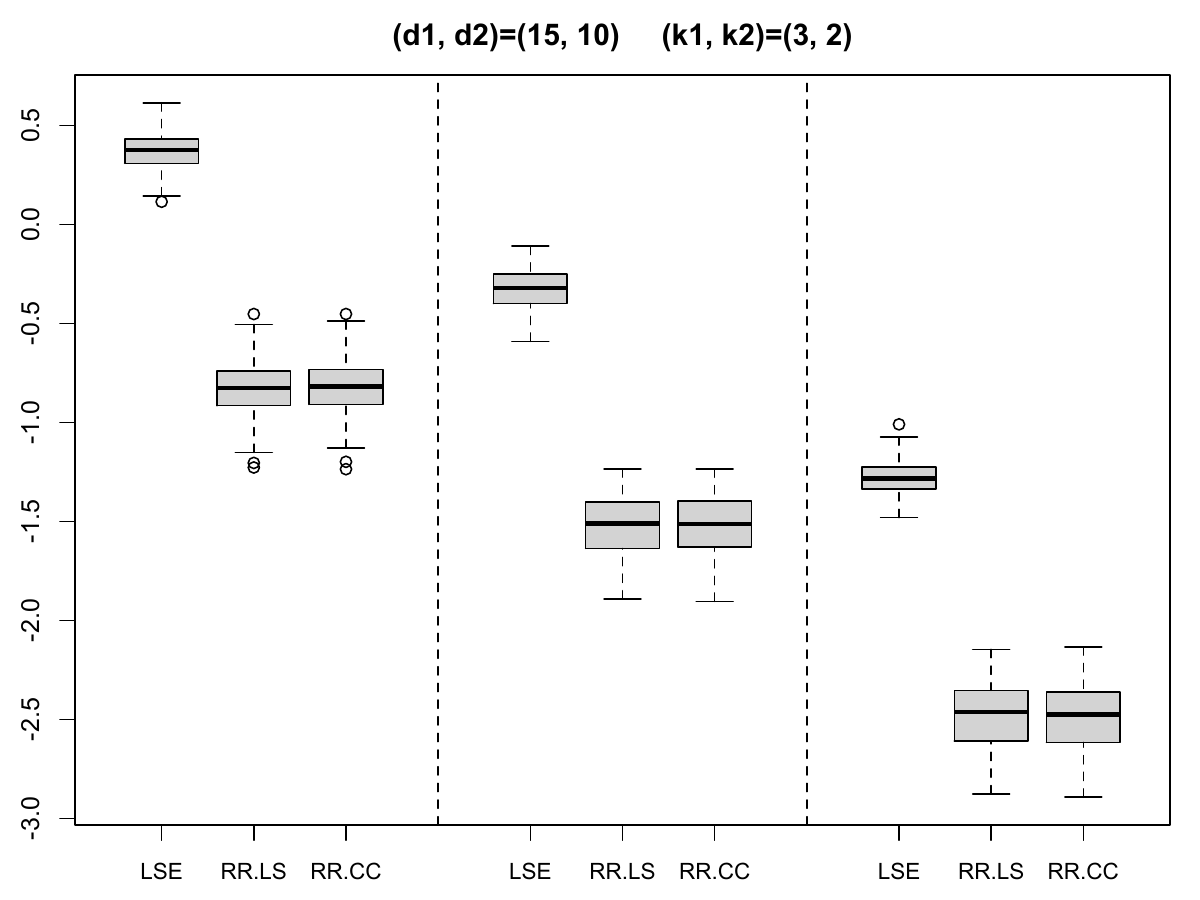}~~\includegraphics[width=2.1in]{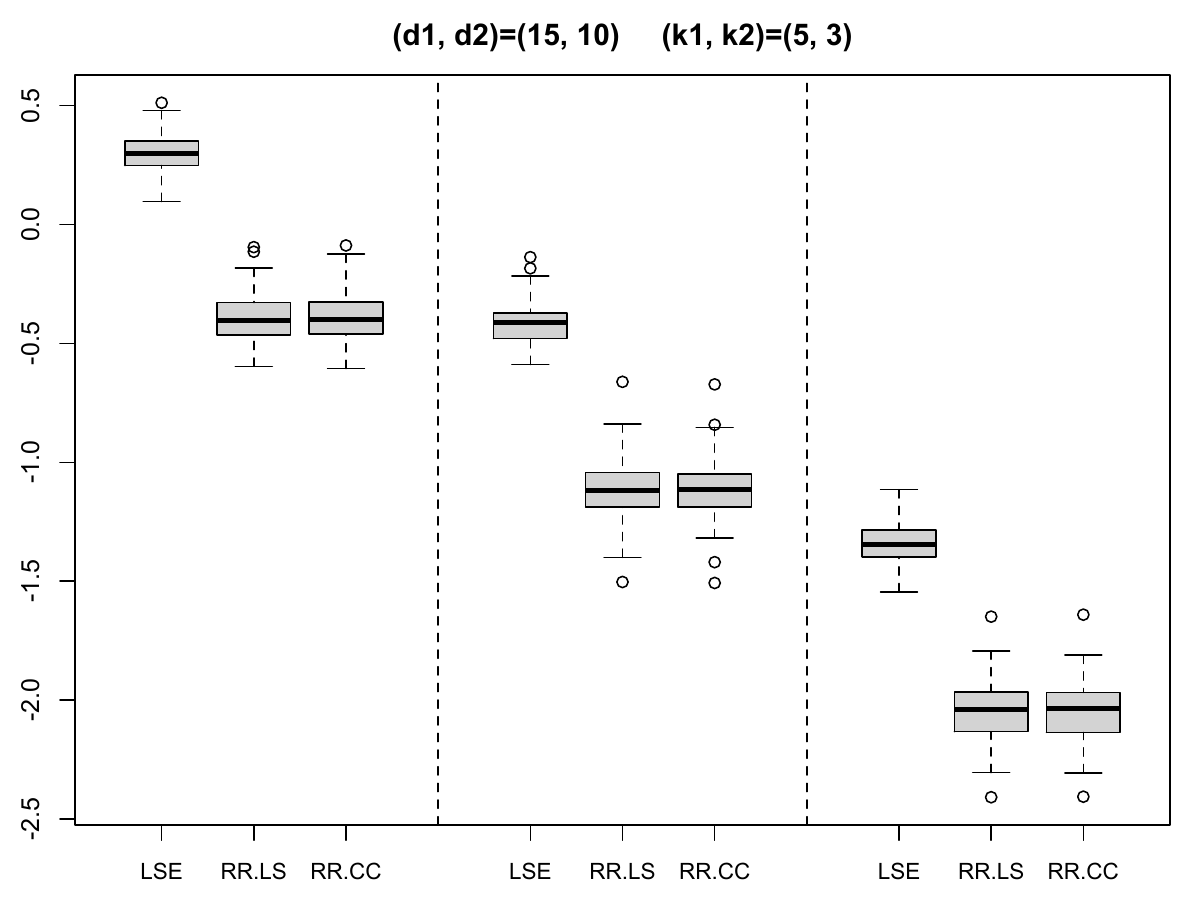}\\
    \caption{Comparison of LSE, RR.LS and RR.CC. The three panels in each figure correspond to sample sizes 200, 400 and 1000 respectively. The errors are generated according to Setting I.}
    \label{fig:simu1}
\end{figure}

\begin{figure}[h]
    \centering
    \includegraphics[width=2.1in]{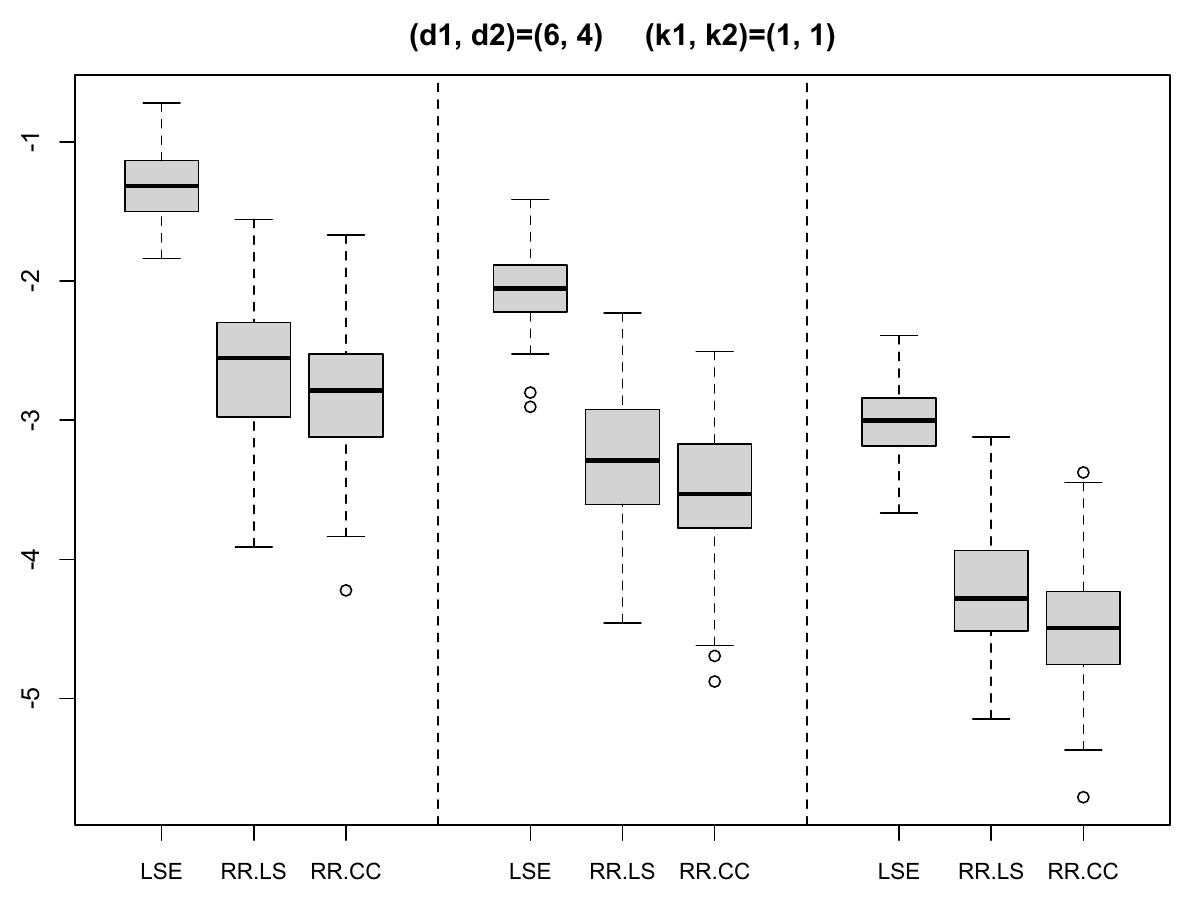}~~\includegraphics[width=2.1in]{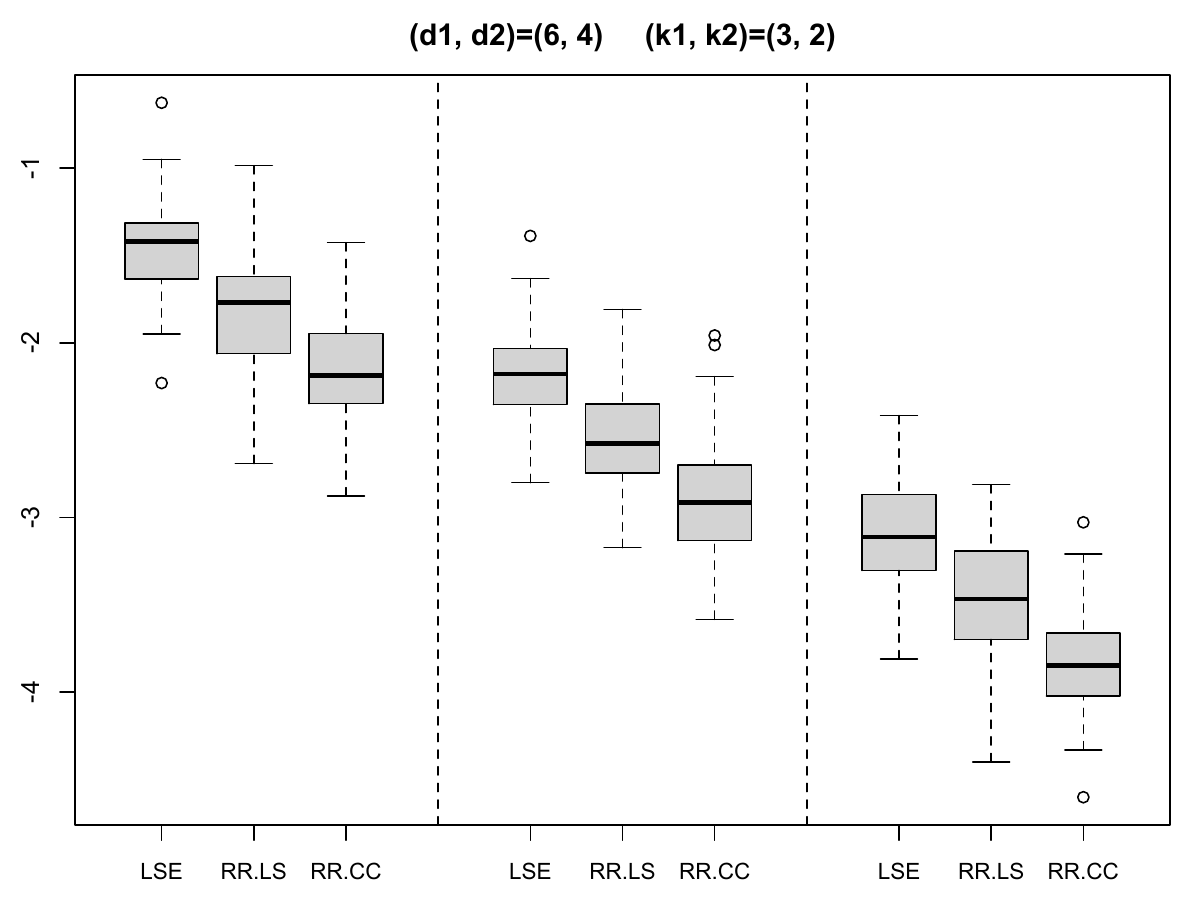}~~\includegraphics[width=2.1in]{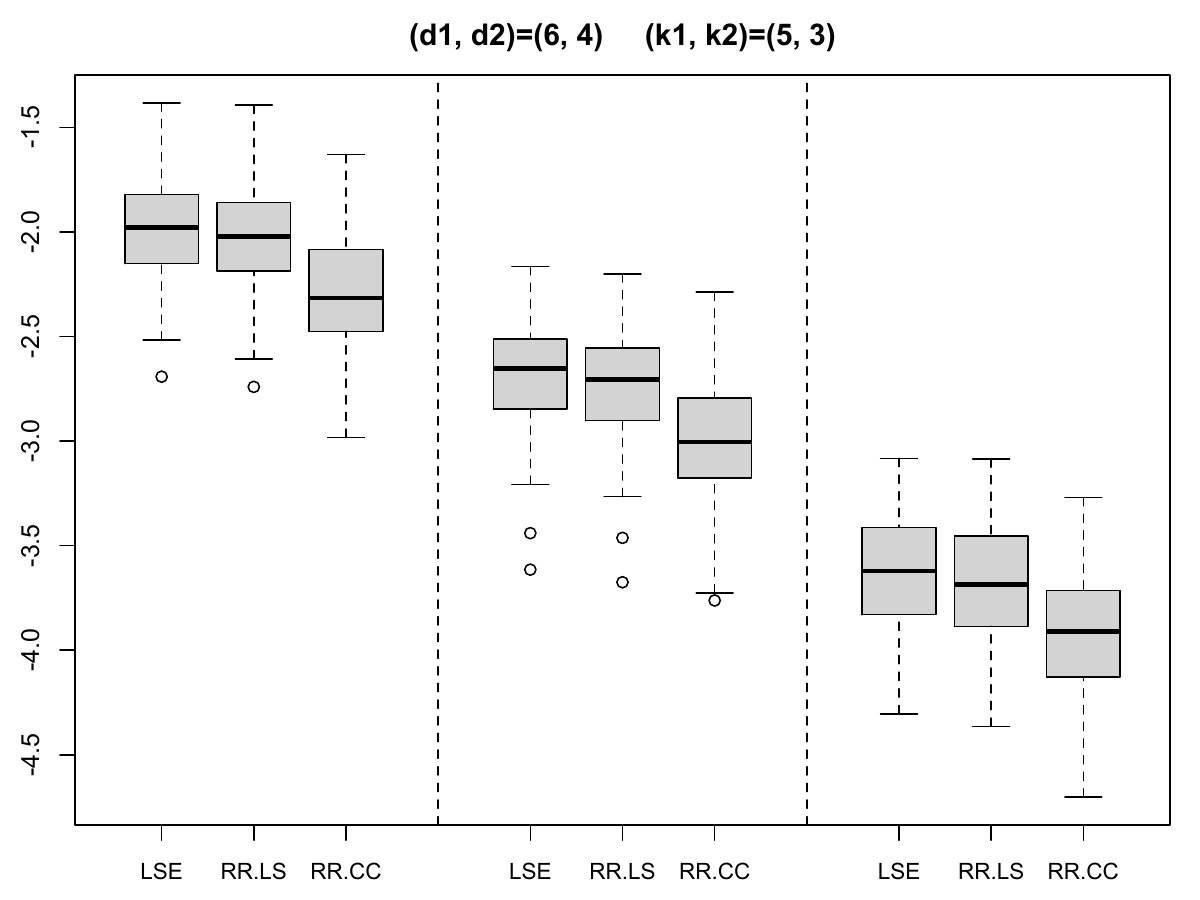}\\
    \includegraphics[width=2.1in]{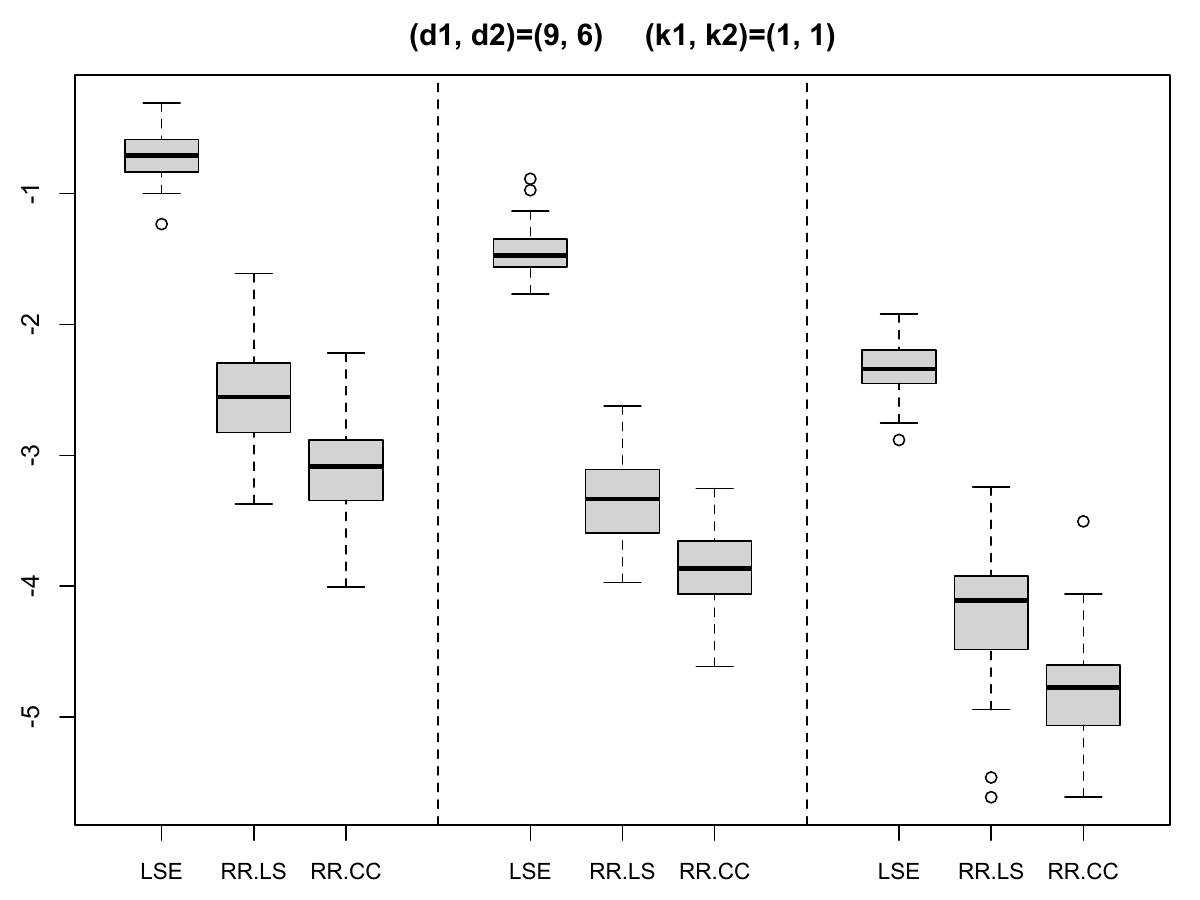}~~\includegraphics[width=2.1in]{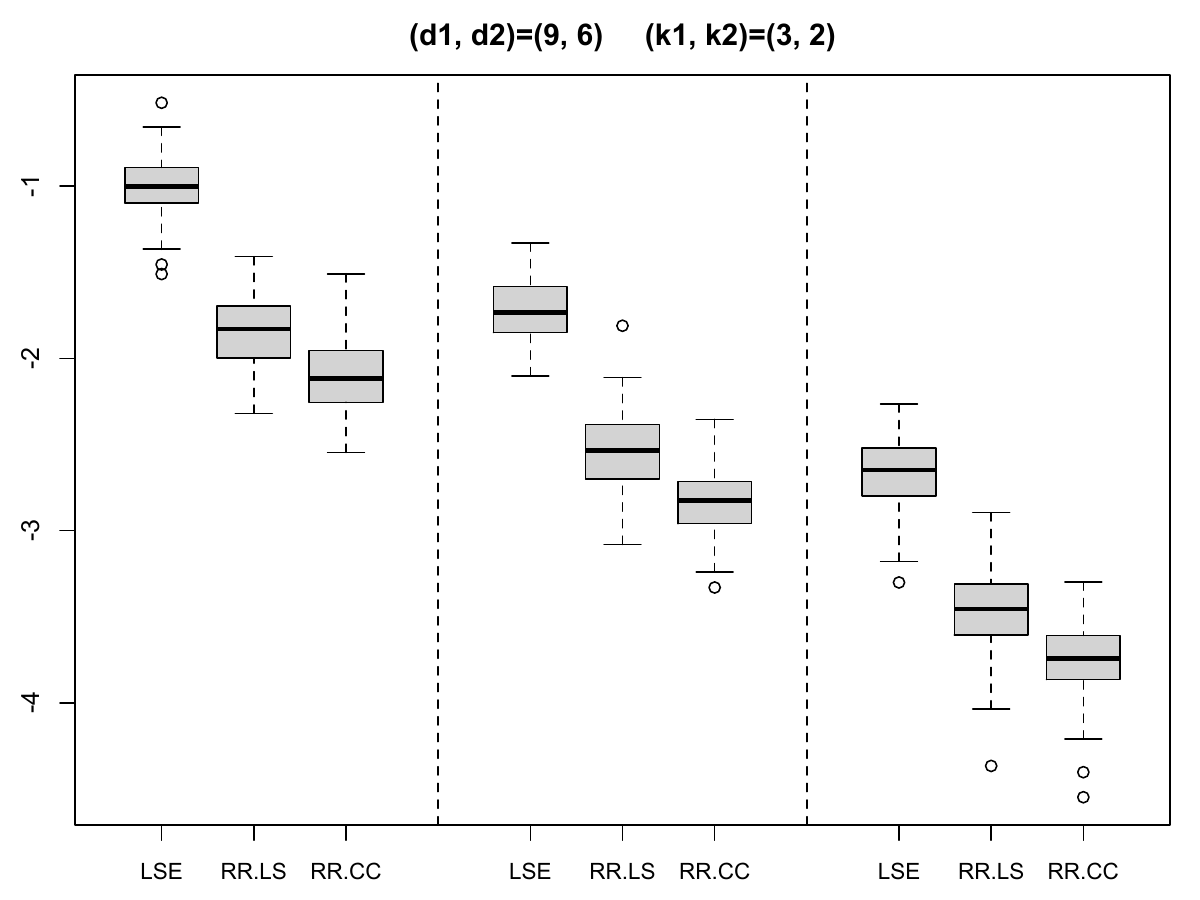}~~\includegraphics[width=2.1in]{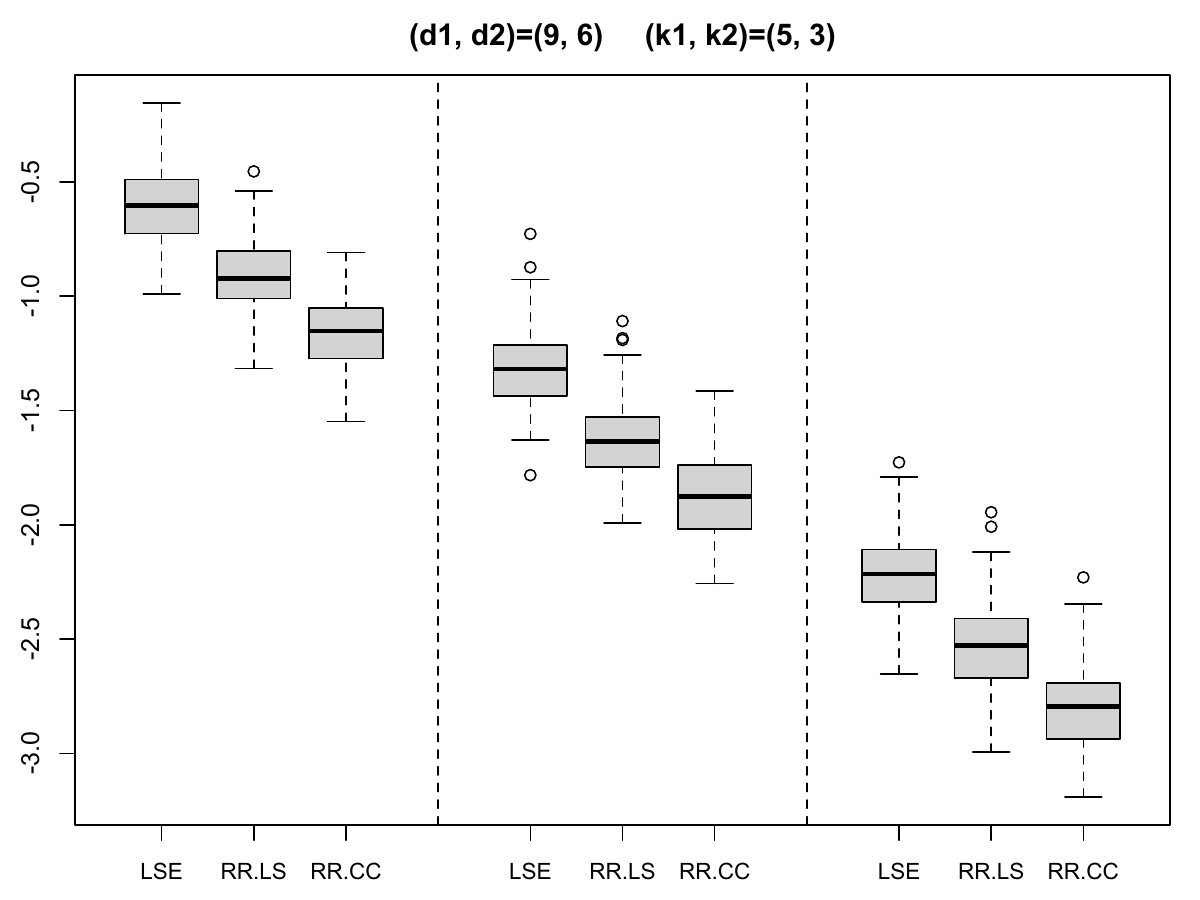}\\
    \includegraphics[width=2.1in]{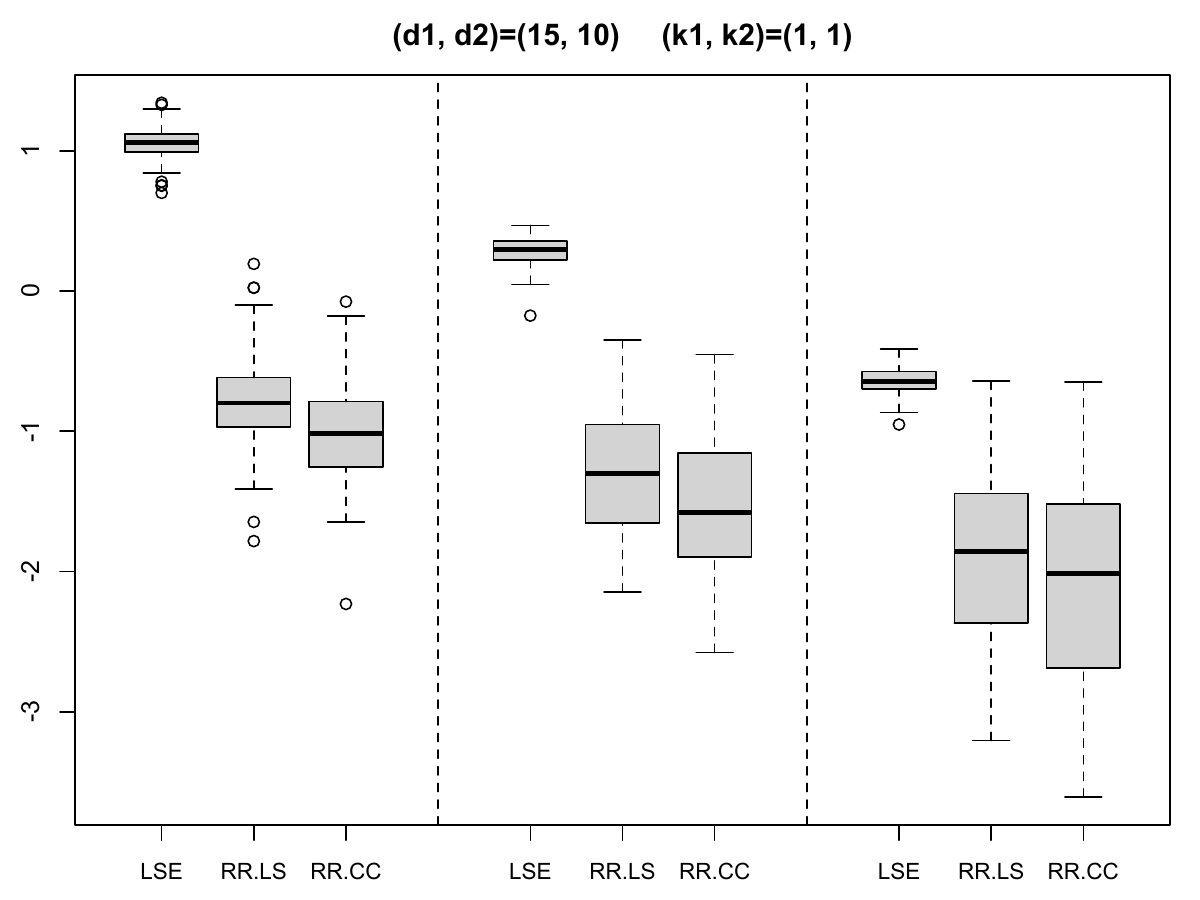}~~\includegraphics[width=2.1in]{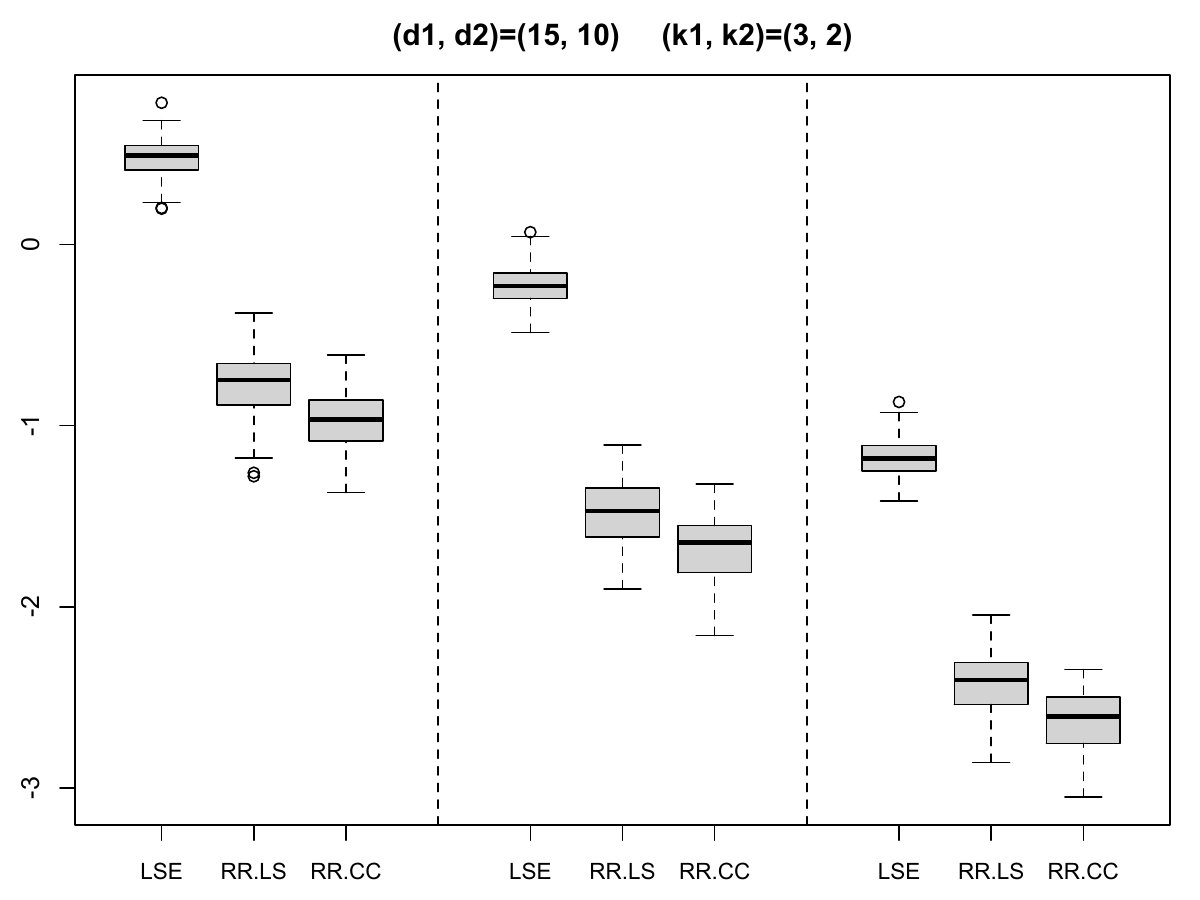}~~\includegraphics[width=2.1in]{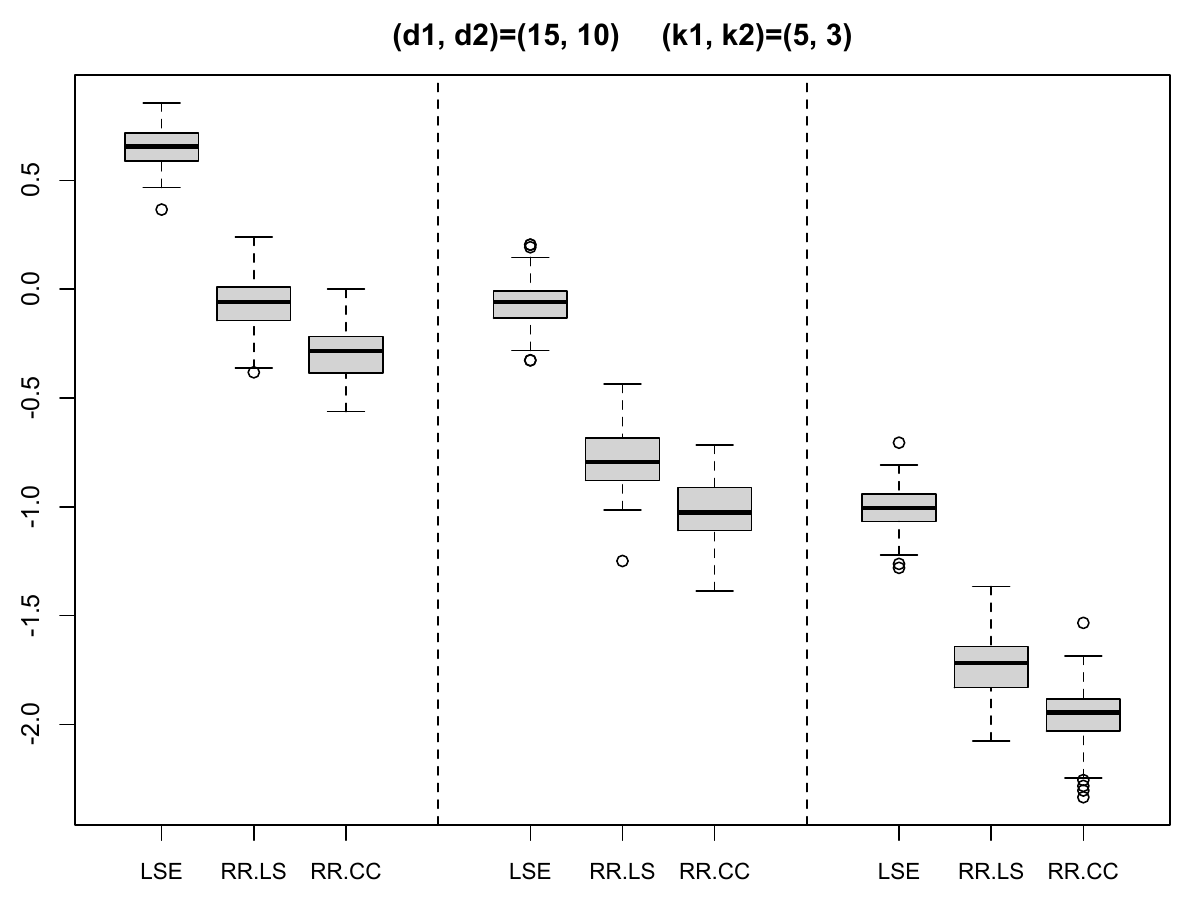}\\
    \caption{Comparison of LSE, RR.LS and RR.CC. The three panels in each figure correspond to sample sizes 200, 400 and 1000 respectively. The errors are generated according to Setting II.}
    \label{fig:simu2}
\end{figure}

In the second part, we consider the coverage probabilities of the confidence intervals based on Theorem~\ref{thm:ls}, Theorem~\ref{thm:cc} and Corollary~\ref{cor:svd}. In this experiment the true ranks are fixed at $k_1=3$ and $k_2=2$. We run simulations 1000 times for sample size $T=200,400,1000$, and consider the cases of
dimension $(d_1,d_2)=(6,4),(9,6),(15,10)$ and $\rho=0.25,0.5,0.75$. For RR.LS, the 
error covariance matrix settings (I) and (II) are considered. For RR.CS, we consider two `correct' settings (II') and (II), where in (II') we use $\Sigma_e=\hI_{d_2}\otimes \hI_{d_1}$.
The confidence intervals of the entries of the matrices $\hA_1,\hA_2, \hU_1,\hV_1, \hU_2,\hV_2$ are constructed. Table \ref{tab:ci} shows the percentage that the true parameters fall within their corresponding marginal 95\% confidence intervals. Each percentage records the average empirical coverage over all involved matrix entries.
It can be seen from the table that the coverage is quite accurate, especially when the sample size is large ($T=1000$). The empirical coverage probabilities are closer to the nominal ones under Setting (I) for RR.LS and Setting (II') for RR.CS than those under Setting (II). This is probably due to the fact that the singulars values of $\Sigma_e$ under Setting (II) are more spread out, making $\Sigma_e$ more different from the scalar matrix. 

\begin{table}[htbp]
\centering
\small
\begin{tabular}{c|cc|rrr | rrr | rrr | rrr}
\toprule
& & & \multicolumn{6}{c|}{RR.LS} & \multicolumn{6}{c}{RR.CC} \\
\hline 
& & {Setting} &\multicolumn{3}{c|}{I} & \multicolumn{3}{c|}{II} &\multicolumn{3}{c|}{II$'$} & \multicolumn{3}{c}{II} \\ 
\hline
& \multicolumn{2}{r|}{$T$}  & $200$ & $400$ & $1000$ & $200$ & $400$ & $1000$ & $200$ & $400$ & $1000$ & $200$ & $400$ & $1000$\\ 
&$\rho$ &$(d_1,d_2)$ & &&&&&&&&&&&\\ 
\hline
\multirow{9}{*}{$(\hA_1,\hA_2)$} &\multirow{3}{*}{$0.75$} 
& $(6,4)$ & 93.5 & 93.8 & 94.3 & 92.1 & 92.3 & 92.5 & 93.9 & 94.2 & 95.0 & 91.0 & 91.3 & 91.3 \\  
& & $(9,6)$ & 94.3 & 94.6 & 95.3 & 92.2 & 92.6 & 92.8 & 94.0 & 94.4 & 95.0 & 91.8 & 92.1 & 92.4 \\   
& & $(15,10)$ & 94.5 & 94.7 & 95.5 & 92.5 & 92.8 & 93.3 & 93.8 & 94.1 & 94.9 & 92.2 & 92.5 & 93.0 \\  
\cline{2-15}
& \multirow{3}{*}{$0.5$} 
& $(6,4)$ & 93.5 & 93.8 & 94.5 & 91.4 & 91.9 & 92.1 & 93.7 & 94.0 & 94.9 & 90.5 & 91.0 & 91.3 \\
& & $(9,6)$ & 94.3 & 94.6 & 95.4 & 91.7 & 92.4 & 92.6 & 93.8 & 94.3 & 95.0 & 91.2 & 92.0 & 92.2 \\    
& & $(15,10)$ & 94.4 & 94.7 & 95.4 & 92.3 & 92.8 & 93.2 & 93.7 & 94.1 & 94.9 & 92.0 & 92.5 & 92.9 \\ 
\cline{2-15}
& \multirow{3}{*}{$0.25$}
& $(6,4)$ & 92.8 & 93.7 & 94.6 & 88.0 & 90.1 & 91.3 & 92.6 & 93.6 & 94.7 & 87.2 & 89.5 & 90.8 \\  
& & $(9,6)$ & 93.3 & 94.4 & 95.3 & 88.8 & 90.9 & 91.9 & 92.8 & 93.8 & 94.8 & 88.4 & 90.7 & 91.8 \\   
& & $(15,10)$ & 93.8 & 94.5 & 95.3 & 90.6 & 92.0 & 92.9 & 93.1 & 93.8 & 94.8 & 90.4 & 91.8 & 92.6 \\  
\hline
\multirow{9}{*}{$(\hU_1,\hV_1)$} &\multirow{3}{*}{$0.75$} 
& $(6,4)$ & 94.4 & 94.2 & 94.2 & 92.4 & 92.7 & 92.5 & 94.4 & 94.3 & 94.5 & 91.3 & 91.5 & 91.1 \\ 
& & $(9,6)$ & 95.0 & 95.0 & 95.1 & 92.4 & 92.8 & 92.6 & 94.5 & 94.4 & 94.8 & 92.0 & 92.2 & 92.3 \\
& & $(15,10)$ & 94.9 & 95.2 & 95.3 & 92.7 & 92.8 & 93.3 & 94.2 & 94.5 & 94.8 & 92.8 & 93.0 & 93.3 \\  
\cline{2-15}
& \multirow{3}{*}{$0.5$} 
& $(6,4)$ & 94.6 & 94.2 & 94.1 & 92.9 & 92.5 & 92.2 & 94.5 & 94.3 & 94.2 & 91.8 & 91.5 & 91.1 \\ 
& & $(9,6)$ & 95.1 & 95.0 & 95.1 & 92.7 & 93.0 & 92.6 & 94.8 & 94.5 & 94.9 & 92.2 & 92.4 & 92.3 \\ 
& & $(15,10)$ & 94.9 & 95.0 & 95.4 & 92.8 & 92.7 & 93.0 & 94.1 & 94.4 & 94.8 & 92.7 & 92.6 & 92.8 \\ 
\cline{2-15}
& \multirow{3}{*}{$0.25$}
& $(6,4)$ & 95.2 & 94.8 & 94.7 & 93.0 & 92.6 & 92.4 & 95.0 & 94.6 & 94.4 & 92.3 & 91.7 & 91.6 \\  
& & $(9,6)$ & 95.3 & 95.3 & 95.3 & 92.5 & 92.8 & 92.6 & 94.9 & 94.8 & 94.9 & 92.0 & 92.4 & 92.2 \\  
& & $(15,10)$ & 95.0 & 95.0 & 95.1 & 92.4 & 92.4 & 92.6 & 94.3 & 94.2 & 94.5 & 92.3 & 92.3 & 92.6 \\  
\hline
\multirow{9}{*}{$(\hU_2,\hV_2)$} &\multirow{3}{*}{$0.75$} 
& $(6,4)$ & 94.1 & 94.3 & 94.2 & 91.6 & 91.7 & 91.1 & 94.5 & 94.5 & 94.3 & 91.2 & 91.3 & 90.6 \\ 
& & $(9,6)$ & 94.7 & 94.6 & 95.5 & 92.0 & 91.9 & 92.5 & 94.4 & 94.2 & 94.8 & 91.9 & 92.1 & 92.9 \\
& & $(15,10)$ & 95.1 & 95.6 & 95.3 & 92.6 & 93.2 & 93.1 & 94.3 & 95.0 & 94.7 & 93.2 & 93.7 & 93.6 \\  
\cline{2-15}
& \multirow{3}{*}{$0.5$} 
& $(6,4)$ & 94.3 & 94.4 & 94.6 & 90.9 & 91.4 & 91.0 & 94.1 & 94.7 & 94.6 & 90.8 & 91.1 & 90.7 \\    
& & $(9,6)$ & 94.8 & 94.7 & 95.3 & 91.8 & 91.8 & 92.3 & 94.4 & 94.2 & 94.5 & 91.5 & 91.9 & 92.4 \\  
& & $(15,10)$ & 95.0 & 95.5 & 95.2 & 92.4 & 93.1 & 92.8 & 94.2 & 94.9 & 94.6 & 92.7 & 93.4 & 93.1 \\  
\cline{2-15}
& \multirow{3}{*}{$0.25$}
& $(6,4)$ & 93.7 & 93.7 & 94.8 & 88.8 & 89.8 & 90.6 & 93.3 & 93.9 & 94.7 & 88.3 & 89.9 & 91.0 \\   
& & $(9,6)$ & 93.8 & 94.7 & 95.3 & 89.4 & 91.0 & 91.9 & 93.5 & 94.1 & 94.4 & 89.2 & 91.2 & 92.1 \\  
& & $(15,10)$ & 94.6 & 95.6 & 95.1 & 91.0 & 92.5 & 92.4 & 93.7 & 94.8 & 94.5 & 91.5 & 92.9 & 92.7 \\  
\bottomrule
\end{tabular}
\caption{Empirical coverage probabilities (in percentage) of the 95\% confidence intervals.}
\label{tab:ci}
\end{table}

The third part of the simulation considers the performance of the rank determination procedure using the joint EBIC \eqref{eq:bic} and the separate EBIC. Simulations are conducted under various configurations of the sample size, dimensions, ranks and signal strength $\rho$, and the empirical probabilities of selecting the correct ranks out of 100 repetitions
are recorded. When the true ranks are $k_1=k_2=1$, or when the signal strength $\rho$ is not too small ($\rho\geq 0.1$), both selection procedures are able to determine the ranks 
perfectly. Therefore, we choose to report in Table~\ref{tab:bic} only the results for the configurations that are more challenging, with $\rho=.15$ when the true ranks are $(3,2)$, and $\rho=.25$ when the true ranks are $(5,3)$. A closer look of the simulation results (not shown in the table) reveals that, in these very low signal to noise ratio cases, both EBIC procedures tend to select ranks smaller than the true ranks. 
However, larger sampling sizes significantly enhance the performance. Moreover, when the autocorrelation strength $\rho$ is larger than those reported in Table~\ref{tab:bic}, both procedures make nearly perfect choices of the ranks for all configurations and covariance settings. We also note that the performances of the joint and separate procedures are almost the same.

It is also observed that the performance under the error covariance setting (II) is worse than that under Setting (I).
This is due to the design of $\Sigma_e$ in these two settings. The eigenvalues of $\Sigma_e$ spread over $[1,10]$ in Setting (I), and over $[1,25]$ in Setting (II). Therefore, both the estimation and model selection are more challenging under Setting (II).

We also experiment with using rolling forecasting to choose the ranks. We consider the range $1\le r_i \le \min\{d_i,k_i+2\}$, $i=1,2$, as the candidate set of $\rk(\hA_i)$. For each configuration of $(\rho,T,k_1,k_2,d_1,d_2)$, we choose $T/4$ as the rolling forecast origin, calculate the entry-wise squared forecast error (SFE) of the one-step ahead prediction of $\hX_{s+1}$, $T/4\le s\le T$, then take the average over the $d_1d_2$ series and over the time, 
\[
\hbox{MSFE}(r_1,r_2)=\frac{1}{d_1d_2(3T/4)}\sum_{s=T/4}^{T-1}\sum_{i=1}^{d_1}\sum_{j=1}^{d_2} 
\left|\hat \hX_{s+1}^{(r_1,r_2)}[i,j]-\hX_{s+1}[i,j]\right|^2.
\]
The estimated ranks $(\hat k_1,\hat k_2)$ is the pair $(r_1,r_2)$ with the smallest MSFE$(r_1,r_2)$. Table~\ref{tab:forecast} shows the proportion of the correct selection out of 100 repetitions. 
It is seen that, although rolling forecast criterion still performs very well in most cases, it has a much higher variability than EBIC. It performs better than EBIC in the case $\rho=.25$ and $(k_1,k_2)=(5,3)$, when the sample size is small.

\begin{table}[H]
    \centering
    \begin{tabular}{cc r rrr r rrr}
    \toprule
        \multirow{2}{*}{}&\multirow{2}{*}{} && \multicolumn{3}{c}{$(r_1,r_2)=(3,2)$, $\rho=.15$} && \multicolumn{3}{c}{$(r_1,r_2)=(5,3)$, $\rho=.25$}  \\ 
	    \cmidrule(lr){4-6}\cmidrule(lr){8-10}
		& $(d_1,d_2)$ && 200 & 400 & \phantom{100}1000 && 200 & 400 & \phantom{100}1000 \\ 
		\midrule
		\multirow{3}{*}{I} 
		& $(6,4)$   && (.01, .01) & (.67, .68) & (1, 1) && (.36, .36) & (.96, .96) & (1, 1) \\ 
		& $(9,6)$   && (.15, .15) & (.90, .92) & (1, 1) && (.07, .07) & (.71, .71) & (1, 1) \\ 
		& $(15,10)$ && (1, 1) & (1, 1) & (1, 1) && (.00, .00) & (.47, .48) & (1, 1) \\ 
	\midrule
		\multirow{3}{*}{II} 
		& $(6,4)$   && (.02, .04) & (.38, .39) & (.99, .99) && (.25, .23) & (.92, .92) & (1, 1) \\ 
		& $(9,6)$   && (.00, .00) & (.45, .47) & (1, 1) && (.10, .09) & (.70, .69) & (1, 1) \\ 
		& $(15,10)$ && (.98, .99) & (1, 1) & (1, 1) && (.00, .00) & (.03, .04) & (1, 1) \\ 
	\bottomrule
    \end{tabular}
    \caption{Empirical probabilities of the correct rank selection by the EBIC. I, II stand for different covariance structures of $\hE_t$. For each cell, two numbers correspond to the joint and separate selections respectively. The second row shows the sample sizes.}
    \label{tab:bic}
\end{table}

\begin{table}[H]
    \centering
    \begin{tabular}{ccc r rrr r rrr r rrr}
    \toprule
		\multirow{2}{*}{}&\multirow{2}{*}{}&\multirow{2}{*}{} && \multicolumn{3}{c}{$(1,1)$} && \multicolumn{3}{c}{$(3,2)$} && \multicolumn{3}{c}{$(5,3)$}  \\ 
		\cmidrule(lr){5-7}\cmidrule(lr){9-11}\cmidrule(lr){13-15}
		&& $(d_1,d_2)$ && 200 & 400 & 1000 && 200 & 400 & 1000 && 200 & 400 & 1000 \\ 
		\midrule
		\multirow{6}{*}{$\rho=.5$} &\multirow{3}{*}{I} 
		 & $(6,4)$ && 0.94 & 0.98 & 0.98 && 0.86 & 0.91 & 0.92 && 0.68 & 0.74 & 0.78 \\ 
		&& $(9,6)$ && 0.99 & 1 & 1 && 0.97 & 0.99 & 0.99 && 0.94 & 0.98 & 1 \\ 
		&& $(15,10)$ && 1 & 1 & 1 && 0.99 & 1 & 1 && 1 & 1 & 1 \\ 
		\cline{2-15}
		&\multirow{3}{*}{II}
		 & $(6,4)$ && 0.95 & 0.97 & 0.99 && 0.87 & 0.88 & 0.92 && 0.74 & 0.77 & 0.79 \\ 
		&& $(9,6)$ && 0.99 & 0.99 & 1 && 0.96 & 0.98 & 0.98 && 0.94 & 0.95 & 0.96 \\ 
		&& $(15,10)$ && 1 & 1 & 1 && 0.99 & 1 & 1 && 1 & 1 & 1 \\ 
		\midrule
		\multirow{6}{*}{$\rho=.25$}&\multirow{3}{*}{I} 
		 & $(6,4)$ && 0.95 & 0.98 & 0.98 && 0.81 & 0.88 & 0.90 && 0.62 & 0.68 & 0.74 \\ 
		&& $(9,6)$ && 1 & 1 & 1 && 0.96 & 0.98 & 0.99 && 0.94 & 0.98 & 0.98 \\ 
		&& $(15,10)$ && 1 & 1 & 1 && 1 & 1 & 1 && 0.96 & 1 & 1 \\ 
		\cline{2-15}
		&\multirow{3}{*}{II}
		 & $(6,4)$ && 0.96 & 0.97 & 0.98 && 0.44 & 0.79 & 0.93 && 0.64 & 0.76 & 0.82 \\ 
		&& $(9,6)$ && 0.99 & 1 & 1 && 0.69 & 0.98 & 0.98 && 0.95 & 0.95 & 0.95 \\ 
		&& $(15,10)$ && 1 & 1 & 1 && 0.99 & 1 & 1 && 0.94 & 1 & 1 \\ 
		\bottomrule
    \end{tabular}
    \caption{Empirical probabilities of the correct rank selection by rolling forecast. I, II stands for different covariance structures of $\hE_t$}
    \label{tab:forecast}
\end{table}


\subsection{Example}
\label{sec:example}

We use the RRMAR model to study the eight key short term economic indicators (116 quarters, from 1991 Q1 to 2019 Q4) from ten countries. The data is downloaded from Organisation for Economic Co-operation and Development (OECD, {\tt https://www.oecd.org/}). The 8 indicators are Consumer Price Index (CPI, growth rate), GDP (growth rate), 3-month interbank Interest Rate (IR3, difference), Long Term government bond yield (IRLT, difference), International Trade total Export Value (ITEX, growth rate) and Import Value (ITIM, growth rate), Total Industrial PRoduction excludingn construction (PRTI, growth rate) and Total Manufacturing PRoduction (PRTM, growth rate). The 10 countries are Australia (AUS), Austria (AUT), Canada (CAN), France (FRA), Germany (DEU), Netherlands (NLD), Norway (NOR), Sweden (SWE), United Kingdom (GBR) and United States (USA). All the 80 series have been centered before attempting the model. We also standardize each indicator across all the countries, i.e. the 10 series corresponding to each indicator have an overall standard deviation 1.

The EBIC \eqref{eq:bic} selects  the ranks as $k_1=1$ and $k_2=4$ for this data set. Using the RR.CC approach, the
estimated leading singular vectors $\hat\hU_i$ and $\hat\hV_i$ of ${\hA_{i}}$ ($i=1,2$), and their corresponding 
estimated standard errors are shown in Tables~\ref{table:A1} and \ref{table:A2}. 
Entries which are not significant at 10\% level are shown in light gray color. 

\begin{table}[H] 
\begin{center}
\begin{tabular}{rrrrrrrrrrr}
  \hline
  & AUS & AUT & CAN & DEU & FRA & GBR & NLD & NOR & SWE & USA\\ 
 $\hat\hU_1'$ & 0.25 & 0.30 & 0.35 & 0.35 & 0.31 & 0.28 & 0.30 & 0.28 & 0.40 & 0.32 \\  
  {s.e.} & 0.02 & 0.01 & 0.01 & 0.01 & 0.01 & 0.01 & 0.01 & 0.02 & 0.01 & 0.01 \\ \hline
  $\hat\hV_1'$ & {\color{lightgray}0.09} & 0.49 & {\color{lightgray}0.05} & {\color{lightgray}0.08} & {\color{lightgray}-0.04} & 0.61 & -0.27 & {\color{lightgray}0.01} & 0.27 & 0.47\\ 
  {s.e.} & 0.10 & 0.12 & 0.13 & 0.15 & 0.17 & 0.10 & 0.11 & 0.07 & 0.1 & 0.13 \\ \hline
\end{tabular}
\caption{Estimated singular vectors of the coefficient matrix $\hA_1$, with their corresponding estimated standard errors.}\label{table:A1}
\end{center}
\end{table}

\begin{table}[H] 
\begin{center}
\begin{tabular}{rrrrrrrrr}
  \hline
  & CPI & GDP & IR3 & IRLT & ITEX & ITIM & PRTI & PRTM\\ 
 $\hat\hU_2'[1,]$& 0.25 & 0.44 & 0.15 & 0.21 & 0.29 & 0.29 & 0.45 & 0.56\\ 
 {s.e.} & 0.05 & 0.03 & 0.05 & 0.04 & 0.03 & 0.03 & 0.03 & 0.02  \\ \hline
 $\hat\hV_2'[1,]$ & {\color{lightgray} -0.05} & 0.84 & -0.43 & 0.21 & 0.17 & {\color{lightgray} 0.14}  & {\color{lightgray} -0.03} & {\color{lightgray} 0.08}\\
 {s.e.} & 0.06 & 0.04 & 0.05 & 0.07 & 0.10 & 0.10 & 0.08 & 0.09\\ \hline
\end{tabular}
\caption{Estimated leading singular vectors of the coefficient matrix $\hA_2$, with their corresponding estimated standard errors.}\label{table:A2}
\end{center}
\end{table}

The implication of using $k_1=1$ and $k_2=4$ is that the observations in the previous quarter form a 4 composite indexes $\h{f}_{t-1}:=\hV_1'\hX_{t-1}\hV_2$, and the conditional expectation $\E(\hX_t\mid\hX_{t-1})$ is given by $d_{11}\cdot \hU_1 \h{f}_{t-1}\hD_2\hU_2'$, where $d_{11}$ is the largest singular value of $\hA_1$, and $\hD_2$ is the $4\times 4$ diagonal matrix containing the singular values of $\hA_2$, see \eqref{eq:marrr_fac} as well. We report the estimated leading singular vectors of $\hU_2$ and $\hV_2$ in Table~\ref{table:A2}. It is very interesting to observe that when the indicators are combined to form the first element of $\h{f}_{t-1}$ using $\hat\hV_2[,1]$, GDP is most dominant, followed by IR3, IRLT and ITEX, while CPI, ITIM, PRTI, PRTM have less importance. It is also worth noting that GDP has a positive coefficient, and IR has a negative one in $\hat{\hV}_2[,1]$. When the countries are combined using $\hat\hV_1$ (see Table~\ref{table:A1}), 5 countries play more important roles (i.e. the 5 significant entries in $\hat\hV_1$), and both USA and GBR are among them. At time $t$, all indicators from all countries significantly load on $\h{f}_{t-1}$.

\begin{table}[H]
\begin{center}
  \begin{tabular}{lrrrrrrrrrr}\hline
    &iAR(1)&VAR(1) & PROJ&LSE&MLE&RR.LS & RR.CC\\\hline
    MSE & 0.5258 & 3.6988 & 1.6362 & 0.5734 & 0.5187 & 0.5816 & 0.5016  \\ 
    \# par & 80 & 6,400 & 163 & 163 & 163 & 102 & 102  \\\hline

  \end{tabular}
  \caption{Out-sample prediction performance comparison of various models for the matrix series of 8 indicators from 10 OECD countries.}\label{table:FFpred}
\end{center}
\end{table}

The model selected by the EBIC does not leading to the best rolling forecast performance. For this purpose, we consider the model with $k_1=5$ and $k_2=2$, which leads to almost the smallest mean sqaured rolling forecast error, but is still much more parsimonious than the full rank MAR model. The mean squared errors of the one-step rolling forecast of the last 8 years are summarized in Table~\ref{table:FFpred}, in which we compare the following seven methods.
\renewcommand{\labelenumi}{(\roman{enumi})}
\begin{enumerate}
    \item {\bf iAR(1):\ } Fit an AR(1) model to each individual series.
    \item {\bf VAR(1):\ } Fit a VAR(1) model to $\vect(\hX_t)$.
    \item {\bf PROJ, LSE, MLE:\ } Fit the MAR(1) model (without rank constraint) to $\hX_t$ using projection, least squares and MLE methods. See \cite{chen2021autoregressive} for details.
    \item {\bf RR.LS:\ } Reduced-rank MAR(1) model, fitted by least squares.
    \item {\bf RR.CC:\ } Reduced-rank MAR(1) model, fitted by MLE under the assumption \eqref{eq:ocov}.
\end{enumerate}

From Table~\ref{table:FFpred}, it is seen that VAR(1) model involves a $80\times 80$ coefficient matrix and significantly overfits the data, with the worst out-sample prediction performance. The MAR and RRMAR models with MLE, have better performance than fitting each individual series separately (iAR(1)).
Comparing to the MAR model without rank constraint, the reduced-rank model estimated by MLE (RR.CC) has the smallest rolling forecast error with less parameters (163 vs 102).

\section{Conclusion}
\label{sec:con}

We introduce the reduced-rank matrix autoregressive model, which relies on an autoregressive term involving bilinear coefficient matrices, and assumes rank deficiency of the coefficient matrices. 
Comparing with the MAR model without the low rank structure \citep{chen2021autoregressive}, the RRMAR model involves a greatly reduced number of parameters and leads to more efficient estimation. On the other hand, we use the MAR estimates as the warm-start initial values for the estimation of the RRMAR model.
Both LSE and MLE are studied, where the latter is considered under an additional assumption that the covariance tensor of the error matrix is separable. We propose to use extended BIC to select the ranks of the coefficient matrices. Our numerical analysis suggests that even if the separability assumption on the covariance tensor does not hold, MLE still has reasonable and almost equally good performance, comparing with LSE. On the other hand, MLE can perform much better when that assumption does stand. Therefore, we would recommend the use of MLE in practice.

There are a number of directions to extend the study of the reduced-rank autoregressive model. For example the conditional mean can involve multiple terms of the form $\sum_{j=1}^J \hA_{j1}\hX_{t-1}\hA_{j2}'$, and multiple lagged terms $\hX_{t-1},\ldots,\hX_{t-p}$. The model can be extended for tesor time series as well. More importantly, the asymptotic analysis has been carried out for the fixed dimensional case in the current paper. It is interesting and important to study the model under the high dimensional paradigm. In particular, we would like to understand: (i) what are the convergence rates of $\hat\hA_i$; and (ii) how to obtain initial estimates of $\hA_i$ to start the alternating algorithm. To select the ranks, either the information criterion based procedure can be adapted to account for the high dimensionality, or the singular (eigen-)value based approach \citep{lam:2012,wang:2018} can be employed. The relationship between the reduced-rank tensor autoregressive model and the dynamic tensor factor model \citep{chen2019factor} is also worth exploring.

{In this paper the theoretical results are obtained under the fix data dimension assumption. It is an interesting and important problem to expand the theoretical results to high and diverging dimension setting. It is a challenging problem and may require additional structure of the model. Due to the stationary condition $\rho(\hA_2\otimes \hA_1)<1$ needed for the autoregressive model, the signal to noise ratio is constrained, different from typical regression models. This is similar to the simple AR(1) model $x_t=\phi x_{t-1}+e_t$, in which the signal to noise ratio is always $\phi^2/(1-\phi^2)$ no matter how large or small the noise variance is. In order to achieve consistency results for the diverging dimensional setting, the reduced-rank structure is not sufficient. It seems that additional sparsity structure or other type of structure is needed. We are currently investigating this problem. }

\bigskip

\noindent {\bf Acknowledgement.}\ \ We thank the Editor, Associate Editor and Referees for their constructive comments and suggestions.

\bibliographystyle{apalike}
\bibliography{mybib}

@book {brockwell:1991,
    AUTHOR = {Brockwell, Peter J. and Davis, Richard A.},
     TITLE = {Time series: theory and methods},
    SERIES = {Springer Series in Statistics},
   EDITION = {Second},
 PUBLISHER = {Springer-Verlag},
   ADDRESS = {New York},
      YEAR = {1991},
     PAGES = {xvi+577},
}

@preamble{
   "\def\polhk#1{\setbox0=\hbox{#1}{\ooalign{\hidewidth
    \lower1.5ex\hbox{`}\hidewidth\crcr\unhbox0}}} "
}

@article {lam:2012,
    AUTHOR = {Lam, Clifford and Yao, Qiwei},
     TITLE = {Factor modeling for high-dimensional time series: inference
              for the number of factors},
   JOURNAL = {Annals of Statistics},
    VOLUME = {40},
      YEAR = {2012},
    NUMBER = {2},
     PAGES = {694--726},
      ISSN = {0090-5364},
   MRCLASS = {62M10 (62H25)},
MRREVIEWER = {Manfred Deistler},
       DOI = {10.1214/12-AOS970},
       URL = {http://dx.doi.org/10.1214/12-AOS970},
}

@book {hannan:1970,
    AUTHOR = {Hannan, E. J.},
     TITLE = {Multiple time series},
 PUBLISHER = {John Wiley and Sons, Inc., New York-London-Sydney},
      YEAR = {1970},
     PAGES = {xi+536},
   MRCLASS = {62.85},
MRREVIEWER = {L. Tornqvist},
}

@article {han:2015,
    AUTHOR = {Han, Fang and Lu, Huanran and Liu, Han},
     TITLE = {A direct estimation of high dimensional stationary vector
              autoregressions},
   JOURNAL = {Journal of Machine Learning Research},
    VOLUME = {16},
      YEAR = {2015},
     PAGES = {3115--3150},
      ISSN = {1532-4435},
   MRCLASS = {62M10 (62H12 62J07)},
}

@book {lutkepohl:2005,
    AUTHOR = {L{\"u}tkepohl, Helmut},
     TITLE = {New introduction to multiple time series analysis},
 PUBLISHER = {Springer-Verlag, Berlin},
      YEAR = {2005},
     PAGES = {xxii+764},
      ISBN = {3-540-40172-5},
   MRCLASS = {62-01 (62M10)},
MRREVIEWER = {Ra{\'u}l Pedro Mentz},
       DOI = {10.1007/978-3-540-27752-1},
       URL = {http://dx.doi.org/10.1007/978-3-540-27752-1},
}

@book {tsay:2014,
    AUTHOR = {Tsay, Ruey S.},
     TITLE = {Multivariate time series analysis},
    SERIES = {Wiley Series in Probability and Statistics},
      OPTNOTE = {With R and financial applications},
 PUBLISHER = {John Wiley \& Sons, Inc., Hoboken, NJ},
      YEAR = {2014},
     PAGES = {xviii+492},
      ISBN = {978-1-118-61790-8},
   MRCLASS = {62M10 (62Hxx 62Mxx 62P05)},
MRREVIEWER = {Biljana {\v{C}}. Popovi{\'c}},
}

@article {negahban:2011,
    AUTHOR = {Negahban, Sahand and Wainwright, Martin J.},
     TITLE = {Estimation of (near) low-rank matrices with noise and
              high-dimensional scaling},
   JOURNAL = {Annals of Statistics},
    VOLUME = {39},
      YEAR = {2011},
    NUMBER = {2},
     PAGES = {1069--1097},
      ISSN = {0090-5364},
   MRCLASS = {62F30 (60B20 62H12)},
  OPTMRNUMBER = {2816348},
MRREVIEWER = {Miguel Fonseca},
       DOI = {10.1214/10-AOS850},
       URL = {http://dx.doi.org.proxy.libraries.rutgers.edu/10.1214/10-AOS850},
}

@article {kock:2015,
    AUTHOR = {Kock, Anders Bredahl and Callot, Laurent},
     TITLE = {Oracle inequalities for high dimensional vector
              autoregressions},
   JOURNAL = {Journal of Econometrics},
    VOLUME = {186},
      YEAR = {2015},
    NUMBER = {2},
     PAGES = {325--344},
      ISSN = {0304-4076},
   MRCLASS = {62J02 (91G70)},
  OPTMRNUMBER = {3343790},
       DOI = {10.1016/j.jeconom.2015.02.013},
       URL = {http://dx.doi.org.proxy.libraries.rutgers.edu/10.1016/j.jeconom.2015.02.013},
}

@article{bunea:2011,
author = "Bunea, Florentina and She, Yiyuan and Wegkamp, Marten H.",
doi = "10.1214/11-AOS876",
journal = "Annals of Statistics",
month = "04",
number = "2",
pages = "1282--1309",
publisher = "The Institute of Mathematical Statistics",
title = "Optimal selection of reduced rank estimators of high-dimensional matrices",
url = "https://doi.org/10.1214/11-AOS876",
volume = "39",
year = "2011"
}

@article{yuan:2007,
 author = {Ming Yuan and Ali Ekici and Zhaosong Lu and Renato Monteiro},
 journal = {Journal of the Royal Statistical Society: Series B (Statistical Methodology)},
 number = {3},
 pages = {329-346},
 publisher = {[Royal Statistical Society, Wiley]},
 title = {Dimension Reduction and Coefficient Estimation in Multivariate Linear Regression},
 volume = {69},
 year = {2007}
}

@article {anderson:1951,
    AUTHOR = {Anderson, T. W.},
     TITLE = {Estimating linear restrictions on regression coefficients for
              multivariate normal distributions},
   FJOURNAL = {Ann. Math. Statistics},
   JOURNAL = {Annals of Mathematical Statistics},
    VOLUME = {22},
      YEAR = {1951},
     PAGES = {327--351},
      ISSN = {0003-4851},
   MRCLASS = {62.0X},
MRREVIEWER = {P. Whittle},
}

@article {izenman:1975,
    AUTHOR = {Izenman, Alan Julian},
     TITLE = {Reduced-rank regression for the multivariate linear model},
   JOURNAL = {Journal of Multivariate Analysis},
    VOLUME = {5},
      YEAR = {1975},
     PAGES = {248--264},
      ISSN = {0047-259x},
   MRCLASS = {62J05 (62H25)},
MRREVIEWER = {Peter M. Robinson},
       URL = {https://doi-org.proxy.libraries.rutgers.edu/10.1016/0047-259X(75)90042-1},
}

@book {reinsel:1998,
    AUTHOR = {Reinsel, Gregory C. and Velu, Raja P.},
     TITLE = {Multivariate reduced-rank regression},
    SERIES = {Lecture Notes in Statistics},
    VOLUME = {136},
      OPTNOTE = {Theory and applications},
 PUBLISHER = {Springer-Verlag, New York},
      YEAR = {1998},
     PAGES = {xiv+258},
      ISBN = {0-387-98601-4},
   MRCLASS = {62J05 (62F30 62H12 62M10 62P20)},
MRREVIEWER = {C. M. Theobald},
       URL = {https://doi-org.proxy.libraries.rutgers.edu/10.1007/978-1-4757-2853-8},
}

@article {rohde:2011,
    AUTHOR = {Rohde, Angelika and Tsybakov, Alexandre B.},
     TITLE = {Estimation of high-dimensional low-rank matrices},
   JOURNAL = {Annals of Statistics},
    VOLUME = {39},
      YEAR = {2011},
    NUMBER = {2},
     PAGES = {887--930},
      ISSN = {0090-5364},
   MRCLASS = {62G05 (62F10)},
MRREVIEWER = {Ou Zhao},
       URL = {https://doi-org.proxy.libraries.rutgers.edu/10.1214/10-AOS860},
}

@article{wang:2018,
author = "Wang, Dong and Liu, Xialu and Chen, Rong",
title = "Factor models for matrix-valued high-dimensional time series",
journal = "Journal of Econometrics",
volume = "208",
number = "1",
pages = "231 - 248",
year = "2019",
}

@article {tiao:1981,
    AUTHOR = {Tiao, G. C. and Box, G. E. P.},
     TITLE = {Modeling multiple time series with applications},
   JOURNAL = {Journal of the American Statistical Association},
    VOLUME = {76},
      YEAR = {1981},
    NUMBER = {376},
     PAGES = {802--816},
      ISSN = {0162-1459},
   MRCLASS = {62M10},
  OPTMRNUMBER = {650891},
MRREVIEWER = {Bert M. Steece},
       URL =
              {http://links.jstor.org.proxy.libraries.rutgers.edu/sici?sici=0162-1459(198112)76:376<802:MMTSWA>2.0.CO;2-E&origin=MSN},
}

@article{zhao:2014,
  title = {{Structured lasso for regression with matrix covariates}},
  author = {Zhao, J. and Leng, C.},
  journal = {Statistica Sinica},
  volume = {24},
  pages = {799--814},
  year = {2014},
}

@article{zhou:2013,
    author = {Zhou, H. and Li, L. and Zhu, H.},
    title = {{Tensor Regression with Applications in Neuroimaging Data Analysis}},
    journal =  {Journal of the American Statistical Association},
    volume = {108},
    number = {502},
    pages = {540-552},
    year = {2013},
}

@article{ghosh2019high,
  title={High-Dimensional Posterior Consistency in Bayesian Vector Autoregressive Models},
  author={Ghosh, Satyajit and Khare, Kshitij and Michailidis, George},
  journal={Journal of the American Statistical Association},
  volume={114},
  number={526},
  pages={735--748},
  year={2019},
  publisher={Taylor \& Francis}
}

@article{allen2010transposable,
  title={Transposable regularized covariance models with an application to missing data imputation},
  author={Allen, Genevera I and Tibshirani, Robert},
  journal={Annals of Applied Statistics},
  volume={4},
  number={2},
  pages={764},
  year={2010},
  publisher={NIH Public Access}
}

@article{chen2019factor,
  title={Factor models for high-dimensional tensor time series},
  author={Chen, Rong and Yang, Dan and Zhang, Cun-Hui},
  journal={Journal of the American Statistical Association},
  volume={just-accepted},
  pages={1--59},
  year={2021},
  publisher={Taylor \& Francis}
}

@article{chen:2019a,
  title={Modeling Dynamic Transport Network with Matrix Factor Models: with an Application to International Trade Flow},
  author={Chen, Elynn Y and Chen, Rong},
  journal={arXiv preprint arXiv:1901.00769},
  year={2019}
}

@book {anderson:2013,
    AUTHOR = {Anderson, T. W.},
     TITLE = {An introduction to multivariate statistical analysis},
    SERIES = {Wiley Series in Probability and Statistics},
   EDITION = {Third},
 PUBLISHER = {Wiley-Interscience, Hoboken, NJ},
      YEAR = {2003},
     PAGES = {xx+721},
      ISBN = {0-471-36091-0},
   MRCLASS = {62-01 (62Hxx)},
  OPTMRNUMBER = {1990662},
}

@book{horn2012matrix,
  title={Matrix analysis},
  author={Horn, Roger A and Johnson, Charles R},
  year={2012},
  publisher={Cambridge university press}
}

@article{de2008tensor,
  title={Tensor rank and the ill-posedness of the best low-rank approximation problem},
  author={De Silva, Vin and Lim, Lek-Heng},
  journal={SIAM Journal on Matrix Analysis and Applications},
  volume={30},
  number={3},
  pages={1084--1127},
  year={2008},
  publisher={SIAM}
}

@article{de2000best,
  title={On the best rank-1 and rank-($r_1$, $r_2$,..., $r_n$) approximation of higher-order tensors},
  author={De Lathauwer, Lieven and De Moor, Bart and Vandewalle, Joos},
  journal={SIAM Journal on Matrix Analysis and Applications},
  volume={21},
  number={4},
  pages={1324--1342},
  year={2000},
  publisher={SIAM}
}

@article{de2000multilinear,
  title={A multilinear singular value decomposition},
  author={De Lathauwer, Lieven and De Moor, Bart and Vandewalle, Joos},
  journal={SIAM Journal on Matrix Analysis and Applications},
  volume={21},
  number={4},
  pages={1253--1278},
  year={2000},
  publisher={SIAM}
}

@article{chen2012sparse,
  title={Sparse reduced-rank regression for simultaneous dimension reduction and variable selection},
  author={Chen, Lisha and Huang, Jianhua Z},
  journal={Journal of the American Statistical Association},
  volume={107},
  number={500},
  pages={1533--1545},
  year={2012},
  publisher={Taylor \& Francis Group}
}

@article{moench2013dynamic,
  title={Dynamic hierarchical factor models},
  author={Moench, Emanuel and Ng, Serena and Potter, Simon},
  journal={Review of Economics and Statistics},
  volume={95},
  number={5},
  pages={1811--1817},
  year={2013},
  publisher={MIT Press}
}

@article{diebold2008global,
  title={Global yield curve dynamics and interactions: a dynamic Nelson--Siegel approach},
  author={Diebold, Francis X and Li, Canlin and Yue, Vivian Z},
  journal={Journal of Econometrics},
  volume={146},
  number={2},
  pages={351--363},
  year={2008},
  publisher={Elsevier}
}

@article{giannone2008nowcasting,
  title={Nowcasting: The real-time informational content of macroeconomic data},
  author={Giannone, Domenico and Reichlin, Lucrezia and Small, David},
  journal={Journal of Monetary Economics},
  volume=55,
  number=4,
  pages={665--676},
  year=2008,
  publisher={Elsevier}
}

@article{kohn1979asymptotic,
  title={Asymptotic estimation and hypothesis testing results for vector linear time series models},
  author={Kohn, Robert},
  journal={Econometrica: Journal of the Econometric Society},
  pages={1005--1030},
  year={1979},
  publisher={JSTOR}
}

@article{tiao1989model,
  title={Model specification in multivariate time series},
  author={Tiao, George C and Tsay, Ruey S},
  journal={Journal of the Royal Statistical Society: Series B (Methodological)},
  volume={51},
  number={2},
  pages={157--195},
  year={1989},
  publisher={Wiley Online Library}
}

@article{tsay1985use,
  title={Use of canonical analysis in time series model identification},
  author={Tsay, Ruey S and Tiao, George C},
  journal={Biometrika},
  volume={72},
  number={2},
  pages={299--315},
  year={1985},
  publisher={Oxford University Press}
}

@article{schwarz1978estimating,
  title={Estimating the dimension of a model},
  author={Schwarz, Gideon},
  journal={Annals of Statistics},
  volume={6},
  number={2},
  pages={461--464},
  year={1978},
  publisher={Institute of Mathematical Statistics}
}

@article{haughton1988choice,
  title={On the choice of a model to fit data from an exponential family},
  author={Haughton, Dominique MA},
  journal={Annals of Statistics},
  volume={16},
  number={1},
  pages={342--355},
  year={1988},
  publisher={Institute of Mathematical Statistics}
}

@article{mukherjee2015degrees,
  title={On the degrees of freedom of reduced-rank estimators in multivariate regression},
  author={Mukherjee, Ashin and Chen, Kun and Wang, Naisyin and Zhu, Ji},
  journal={Biometrika},
  volume={102},
  number={2},
  pages={457--477},
  year={2015},
  publisher={Oxford University Press}
}

@article{yuan2016degrees,
  title={Degrees of freedom in low rank matrix estimation},
  author={Yuan, Ming},
  journal={Science China Mathematics},
  volume={59},
  number={12},
  pages={2485--2502},
  year={2016},
  publisher={Springer}
}

@article{raskutti2019convex,
  title={Convex regularization for high-dimensional multiresponse tensor regression},
  author={Raskutti, Garvesh and Yuan, Ming and Chen, Han},
  journal={Annals of Statistics},
  volume={47},
  number={3},
  pages={1554--1584},
  year={2019},
  publisher={Institute of Mathematical Statistics}
}

@article{han2020tensor,
  title={Tensor Factor Model Estimation by Iterative Projection},
  author={Han, Yuefeng and Chen, Rong and Yang, Dan and Zhang, Cun-hui},
  journal={arXiv preprint arXiv:2006.02611},
  year={2020}
}

@article{han2020number,
  title={Rank Determination in Tensor Factor Model},
  author={Han, Yuefeng and Zhang, Cun-hui and Chen, Rong},
  journal={arXiv preprint arxiv:2011.07131},
  volume={},
  number={},
  pages={},
  year={2020},
  publisher={}
}

@inproceedings{loh2011high,
  title={High-dimensional regression with noisy and missing data: Provable guarantees with non-convexity},
  author={Loh, Po-Ling and Wainwright, Martin J},
  booktitle={Advances in Neural Information Processing Systems},
  pages={2726--2734},
  year={2011}
}

@article{davis2016sparse,
  title={Sparse vector autoregressive modeling},
  author={Davis, Richard A and Zang, Pengfei and Zheng, Tian},
  journal={Journal of Computational and Graphical Statistics},
  volume={25},
  number={4},
  pages={1077--1096},
  year={2016},
  publisher={Taylor \& Francis}
}

@article{nicholson2017varx,
  title={VARX-L: Structured regularization for large vector autoregressions with exogenous variables},
  author={Nicholson, William B and Matteson, David S and Bien, Jacob},
  journal={International Journal of Forecasting},
  volume={3},
  number={33},
  pages={627--651},
  year={2017}
}

@inproceedings{melnyk2016estimating,
  title={Estimating structured vector autoregressive models},
  author={Melnyk, Igor and Banerjee, Arindam},
  booktitle={International Conference on Machine Learning},
  pages={830--839},
  year={2016}
}

@article{ghosh2020strong,
  title={STRONG SELECTION CONSISTENCY OF BAYESIAN VECTOR AUTOREGRESSIVE MODELS BASED ON A PSEUDO-LIKELIHOOD APPROACH},
  author={Ghosh, Satyajit and Khare, Kshitij and Michailidis, George},
  journal={Annals of Statistics},
  volume={in press},
  number={},
  pages={},
  year={2021+},
  publisher={Institute of Mathematical Statistics}
}

@TECHREPORT{han2020sparse,
  title={Sparse Nonlinear Vector Autoregressive Models},
  author={Han, Yuefeng and Chen, Likai and Wu, Weibiao},
  journal={},
  volume={},
  number={},
  pages={},
  year={2020+},
  publisher={Technical Report}
}

@article{basu2015regularized,
  title={Regularized estimation in sparse high-dimensional time series models},
  author={Basu, Sumanta and Michailidis, George},
  journal={Annals of Statistics},
  volume={43},
  number={4},
  pages={1535--1567},
  year={2015},
  publisher={Institute of Mathematical Statistics}
}

@article{basu2019low,
  title={Low rank and structured modeling of high-dimensional vector autoregressions},
  author={Basu, Sumanta and Li, Xianqi and Michailidis, George},
  journal={IEEE Transactions on Signal Processing},
  volume={67},
  number={5},
  pages={1207--1222},
  year={2019},
  publisher={IEEE}
}

@article{hall2018learning,
  title={Learning high-dimensional generalized linear autoregressive models},
  author={Hall, Eric C and Raskutti, Garvesh and Willett, Rebecca M},
  journal={IEEE Transactions on Information Theory},
  volume={65},
  number={4},
  pages={2401--2422},
  year={2018},
  publisher={IEEE}
}

@article{lin2017regularized,
  title={Regularized estimation and testing for high-dimensional multi-block vector-autoregressive models},
  author={Lin, Jiahe and Michailidis, George},
  journal={Journal of Machine Learning Research},
  volume={18},
  number={1},
  pages={4188--4236},
  year={2017},
  publisher={JMLR. org}
}

@article{lin2020regularized,
  title={Regularized Estimation of High-dimensional Factor-Augmented Vector Autoregressive ({FAVAR}) Models},
  author={Lin, Jiahe and Michailidis, George},
  journal={Journal of Machine Learning Research},
  volume={21},
  number={117},
  pages={1--51},
  year={2020}
}

@article{cichocki2015tensor,
  title={Tensor decompositions for signal processing applications: From two-way to multiway component analysis},
  author={Cichocki, Andrzej and Mandic, Danilo and De Lathauwer, Lieven and Zhou, Guoxu and Zhao, Qibin and Caiafa, Cesar and Phan, Huy Anh},
  journal={IEEE Signal Processing Magazine},
  volume={32},
  number={2},
  pages={145--163},
  year={2015},
  publisher={IEEE}
}

@article{sidiropoulos2017tensor,
  title={Tensor decomposition for signal processing and machine learning},
  author={Sidiropoulos, Nicholas D and De Lathauwer, Lieven and Fu, Xiao and Huang, Kejun and Papalexakis, Evangelos E and Faloutsos, Christos},
  journal={IEEE Transactions on Signal Processing},
  volume={65},
  number={13},
  pages={3551--3582},
  year={2017},
  publisher={IEEE}
}

@book{cichocki2009nonnegative,
  title={Nonnegative matrix and tensor factorizations: applications to exploratory multi-way data analysis and blind source separation},
  author={Cichocki, Andrzej and Zdunek, Rafal and Phan, Anh Huy and Amari, Shun-ichi},
  year={2009},
  publisher={John Wiley \& Sons}
}

@article{anandkumar2014tensor,
  title={Tensor decompositions for learning latent variable models},
  author={Anandkumar, Animashree and Ge, Rong and Hsu, Daniel and Kakade, Sham M and Telgarsky, Matus},
  journal={Journal of Machine Learning Research},
  volume={15},
  pages={2773--2832},
  year={2014},
  publisher={Journal of Machine Learning Research}
}

@article{tsiligkaridis2013covariance,
  title={Covariance estimation in high dimensions via kronecker product expansions},
  author={Tsiligkaridis, Theodoros and Hero, Alfred O},
  journal={IEEE Transactions on Signal Processing},
  volume={61},
  number={21},
  pages={5347--5360},
  year={2013},
  publisher={IEEE}
}

@article{hoff2011separable,
  title={Separable covariance arrays via the Tucker product, with applications to multivariate relational data},
  author={Hoff, Peter D},
  journal={Bayesian Analysis},
  volume={6},
  number={2},
  pages={179--196},
  year={2011},
  publisher={International Society for Bayesian Analysis}
}

@article{hafner2020estimation,
  title={Estimation of a multiplicative correlation structure in the large dimensional case},
  author={Hafner, Christian M and Linton, Oliver B and Tang, Haihan},
  journal={Journal of Econometrics},
  volume={217},
  number={2},
  pages={431--470},
  year={2020},
  publisher={Elsevier}
}

@article{linton2019estimation,
  title={Estimation of the Kronecker Covariance Model by Partial Means and Quadratic Form},
  author={Linton, Oliver B and Tang, Haihan},
  journal={arXiv preprint arXiv:1906.08908},
  year={2019}
}

@article{zhou2014gemini,
  title={Gemini: Graph estimation with matrix variate normal instances},
  author={Zhou, Shuheng},
  journal={Annals of Statistics},
  volume={42},
  number={2},
  pages={532--562},
  year={2014},
  publisher={Institute of Mathematical Statistics}
}

@article{gao2021two,
  title={A Two-Way Transformed Factor Model for Matrix-Variate Time Series},
  author={Gao, Zhaoxing and Tsay, Ruey S},
  journal={Econometrics and Statistics},
  volume={in press},
  year={2021},
  publisher={Elsevier}
}

@article{hoff2015multilinear,
  title={Multilinear tensor regression for longitudinal relational data},
  author={Hoff, Peter D},
  journal={Annals of Applied Statistics},
  volume={9},
  number={3},
  pages={1169},
  year={2015},
  publisher={NIH Public Access}
}

@article{ding2018matrix,
author = {Ding, Shanshan and Dennis Cook, R.},
title = {Matrix variate regressions and envelope models},
journal = {Journal of the Royal Statistical Society: Series B (Statistical Methodology)},
volume = {80},
number = {2},
pages = {387-408},
year = {2018}
}

@article{chen2021autoregressive,
  title={Autoregressive models for matrix-valued time series},
  author={Chen, Rong and Xiao, Han and Yang, Dan},
  journal={Journal of Econometrics},
  volume={222},
  number={1},
  pages={539--560},
  year={2021},
  publisher={Elsevier}
}

@article{velu1986reduced,
  title={Reduced rank models for multiple time series},
  author={Velu, Raja P and Reinsel, Gregory C and Wichern, Dean W},
  journal={Biometrika},
  volume={73},
  number={1},
  pages={105--118},
  year={1986},
  publisher={Oxford University Press}
}

@article{camba2003tests,
  title={Tests of rank in reduced rank regression models},
  author={Camba-Mendez, Gonzalo and Kapetanios, George and Smith, Richard J and Weale, Martin R},
  journal={Journal of Business \& Economic Statistics},
  volume={21},
  number={1},
  pages={145--155},
  year={2003},
  publisher={Taylor \& Francis}
}

@article{al2019testing,
  title={Testing subspace Granger causality},
  author={Al-Sadoon, Majid M},
  journal={Econometrics and Statistics},
  volume={9},
  pages={42--61},
  year={2019},
  publisher={Elsevier}
}

@article{hallin2007,
  title={Determining the number of factors in the general dynamic factor model},
  author={Hallin, Marc and Li{\v{s}}ka, Roman},
  journal={Journal of the American Statistical Association},
  volume={102},
  number={478},
  pages={603--617},
  year={2007},
  publisher={Taylor \& Francis}
}

@article{shao1997asymptotic,
  title={An asymptotic theory for linear model selection},
  author={Shao, Jun},
  journal={Statistica sinica},
  volume={7},
  pages={221--242},
  year={1997},
  publisher={JSTOR}
}

@article{wedin1972perturbation,
  title={Perturbation bounds in connection with singular value decomposition},
  author={Wedin, Per-{\AA}ke},
  journal={BIT Numerical Mathematics},
  volume={12},
  number={1},
  pages={99--111},
  year={1972},
  publisher={Springer}
}

@book{stewart1990matrix,
  title={Matrix Perturbation Theory},
  author={Stewart, G.W. and Sun, J.},
  isbn={9780126702309},
  lccn={lc90033378},
  series={Computer Science and Scientific Computing},
  year={1990},
  publisher={Elsevier Science}
}

\clearpage
\section*{Appendix}
\label{sec:app}


\newcommand{\hlambda}{\h{\lambda}}

\bigskip 

\noindent{\large\bf Appendix I: Proofs}

\bigskip

The proof of Theorem~\ref{thm:ls} is similar to that of Theorem~\ref{thm:cc}, and is much simpler since it does not involve the matrices $\Sigma_i$ and $\hat\Sigma_i$. Therefore, we will present the proof of Theorem~\ref{thm:cc} first and then point out the major difference for the proof of Theorem~\ref{thm:ls}.


\begin{proof}[Proof of Theorem~\ref{thm:cc}]
Let $\hat\hA_i^{\hbox{\tiny cc}}$ and $\hat\Sigma_i$ be the MLE under the model \eqref{eq:marrr} and \eqref{eq:ocov}. First, using the arguments of the proof of Theorem~4 in \cite{chen2021autoregressive}, we have that $\hat \hA_i^{\hbox{\tiny cc}} = \hA_i + O_P(T^{-1/2})$, and $\hat\Sigma_i = \Sigma_i + o_P(1)$. For the rest of the proof, we will drop the superscript $^{\hbox{\tiny cc}}$ to simplify the notation. Based on the likelihood function \eqref{eq:loglik}, similar to \eqref{eq:mle-A1} (also see Equation (2.15) of \cite{reinsel:1998}), the gradient condition for $\hat\hA_1$ is given by:
\begin{equation}
\label{eq:cc_grad1}
    \hat \hA_1\hat\hS_{1xx}=\hat\Sigma_{1}^{1/2}\hat\hU_1\hat\hU_1'\hat\Sigma_{1}^{-1/2}\hat\hS_{1yx},
\end{equation}
where 
\begin{align*}
    \hat\hS_{1xx} & =\sum_t \hX_{t-1}\hat\hA_2'\hat\Sigma_2^{-1}\hat\hA_2\hX_{t-1}',\\
    \hat\hS_{1yx} & =\sum_t \hX_t\hat\Sigma_2^{-1}\hat\hA_2\hX_{t-1}',\\
    \hat\Sigma_1 & = \frac{1}{T-1} \sum_t\left(\hX_t-\hat\hA_1\hX_{t-1}\hat\hA_2'\right)\hat\Sigma_2^{-1}\left(\hX_t-\hat\hA_1\hX_{t-1}\hat\hA_2'\right)',
\end{align*}
and $\hat\hU_1$ is the $d_1\times k_1$ matrix consisting of the first $k_1$ leading eigenvectors (all normalized to have unit length) of $\hat\Sigma_{1}^{-1/2}\hat\hS_{1yx}\hat\hS_{1xx}^{-1}\hat\hS_{1xy}\hat\Sigma_{1}^{-1/2}$.
With similarly defined quantities (by swapping $\hA_1$ and $\hA_2$, $\Sigma_1$ and $\Sigma_2$, and $\hX_t$ and $\hX_t'$ respectively), we have
\begin{align*}
    \hat \hA_2\hat\hS_{2xx}=\hat\Sigma_{2}^{1/2}\hat\hU_2\hat\hU_2'\hat\Sigma_{2}^{-1/2}\hat\hS_{2yx}.
\end{align*}

Note that $\hat\hU_1$ can also be viewed as the first $k_1$ leading left singular vectors (all normalized to have unit length) of $\hat\hM_1:=T^{-1/2}\hat\Sigma_{1}^{-1/2}\hat\hS_{1yx}\hat\hS_{1xx}^{-1/2}$.  Intuitively, the proof should rely on the expansion of $\hat\hU_1\hat\hU_1'$ around the true value $\hU_1\hU_1'$. However, since the estimates $(\hat\hA_1,\,\hat\hU_1,\,\hat\Sigma_1)$ and $(\hat\hA_2,\,\hat\hU_2,\,\hat\Sigma_2)$ are intertwined, we introduce an intermediate
$$\tilde\hM_1:=T^{-1/2}\hat\Sigma_1^{-1/2}\hA_1\left(\sum_t \hX_{t-1}\hA_2'\hat\Sigma_2^{-1}\hat\hA_2\hX_{t-1}'\right)\hat \hS_{1xx}^{-1/2},$$ and let $\tilde \hU_1$ be the orthogonal matrix consisting of the normalized left singular vectors of $\tilde \hM_1$. The fact that $\hat\Sigma_1^{1/2}\tilde\hU_1\tilde\hU_1'\hat\Sigma_1^{-1/2}\hA_1=\hA_1$ will be of critical importance later.

Let $\tilde\hM_1=\tilde\hU_1\tilde\hD_1\tilde\hV_1'$ be the SVD of $\tilde\hM_1$ (which is of rank $k_1$), and $\hat\hU_1\hat\hD_1\hat\hV_1'$ be the leading rank-$k_1$ SVD component of $\hat\hM_1$. 
Since 
\begin{equation}
    \label{eq:add1}
    \hat\hM_1 - \tilde\hM_1 = T^{-1/2}\hat\Sigma_1^{-1/2}\left(\sum_t \hE_t\hat\Sigma_2^{-1}\hat\hA_2\hX_{t-1}'\right)\hat \hS_{1xx}^{-1/2},
\end{equation}
the convergence rates of $\hat\hA_i$ and $\hat\Sigma_i$ imply that $\hat \hM_1 = \tilde\hM_1 + O_P(1/\sqrt{T})$ by the Slutsky's Theorem. It follows that $\hat\hU\hat\hU'=\tilde\hU_1\tilde\hU_1'+ O_P(1/\sqrt{T})$ and $\hat\hV\hat\hV'=\tilde\hV_1\tilde\hV_1'+ O_P(1/\sqrt{T})$, by Wedin's $\sin \theta$ Theorem \citep{wedin1972perturbation}. It further follows that 
\begin{equation*}
    \hat\hU_1\hat\hD_1\hat\hV_1' - \tilde\hM_1 = \hat\hU_1\hat\hU_1'\hat\hM_1 - \tilde\hM_1 = (\hat\hU\hat\hU'-\tilde\hU_1\tilde\hU_1')\tilde\hM_1 + \hat\hU\hat\hU'(\hat\hM_1-\tilde\hM_1) = O_P(1/\sqrt{T}),
\end{equation*}
and thus $\hat\hU_1\hat\hD_1\hat\hV_1'$ is an acute (see Section~3.1 of \cite{stewart1990matrix} for the definition of {\it acuteness}) perturbation of $\tilde \hM_1$. Therefore, according to Theorem~4.1 of \cite{stewart1990matrix},
\begin{equation}
\label{eq:u_proj_1}  
\begin{aligned}
    \hat\hU_1\hat\hU_1' & = \tilde\hU_1\tilde\hU_1' + (\hI-\tilde\hU_1\tilde\hU_1')(\hat\hM_1 - \tilde\hM_1)\tilde\hV_1\tilde\hD_1^{-1}\tilde\hU_1' \\
    & \qquad + \tilde\hU_1\tilde\hD_1^{-1}\tilde\hV_1'(\hat\hM_1 - \tilde\hM_1)'(\hI-\tilde\hU_1\tilde\hU_1') + o_P(1/\sqrt{T}).
\end{aligned}
\end{equation}
Let $\hM_1:=\Sigma_1^{-1/2}\hA_1\Gamma_1^{1/2}$ and $\hM_1=\hU_1\hD_1\hV_1'$ be its SVD. The consistencies of of $\hat\hA_i$ and $\hat\Sigma_i$ also imply that $\tilde\hM_1$, $\tilde\hU_1\tilde\hU_1'$ are consistent for $\hM_1$ and $\hU_1\hU_1'$ respectively. This fact, combined with \eqref{eq:u_proj_1}, yields
\begin{align*}
    \hat\hU_1\hat\hU_1' = \tilde\hU_1\tilde\hU_1' & + (\hI-\hU_1\hU_1') \Sigma_1^{-1/2}\left( \sum_t \hE_t\Sigma_2^{-1}\hA_2\hX_{t-1}' \right) \Gamma_1^{-1/2}\left(\Sigma_1^{-1/2} \hA_1\Gamma_1^{1/2}\right)^{+} \\
    & + \left(\Gamma_1^{1/2}\hA_1'\Sigma_1^{-1/2}\right)^{+}\Gamma_1^{-1/2} \left( \sum_t \hE_t\Sigma_2^{-1}\hA_2\hX_{t-1}' \right)' \Sigma_1^{-1/2}(\hI-\hU_1\hU_1') + o_P(1/\sqrt{T}).
\end{align*}
Note that $\hU_1\hU_1'=\mathbb{P}_1$. Let $\hP_i=\Sigma_i^{1/2}\mathbb{P}_i\Sigma_i^{-1/2}$. Plugging the preceding equation into \eqref{eq:cc_grad1}, and using the facts that $(\hI-\hU_1\hU_1')\Sigma_1^{-1/2}\hA_1\Gamma_1=\h{0}$ and 
$\hat\Sigma_1^{1/2}\tilde\hU_1\tilde\hU_1'\hat\Sigma_1^{-1/2}\hA_1=\hA_1$ {(this critical step was mentioned earlier in the proof)}, we get
\begin{equation*}
\begin{aligned}
    & \hat\hA_1\left( \sum_t \hX_{t-1}\hat\hA_2'\hat\Sigma_2^{-1}\hat\hA_2\hX_{t-1}'\right) -
    \hA_1\left(\sum_t \hX_{t-1}\hA_2'\hat\Sigma_2^{-1}\hat\hA_2\hX_{t-1}' \right)\\
    = & (\hI - \hP_1) \left( \sum_t \hE_t\Sigma_2^{-1}\hA_2\hX_{t-1}' \right) \Gamma_1^{-1/2}\left(\Sigma_1^{-1/2} \hA_1\Gamma_1^{1/2}\right)^{+}\Sigma_1^{-1/2}\hA_1\Gamma_1 \\
    & \qquad+ \hP_1 \left( \sum_t \hE_t\Sigma_2^{-1}\hA_2\hX_{t-1}' \right) + o_P(\sqrt{T})\\
    = & (\hI - \hP_1) \left( \sum_t \hE_t\Sigma_2^{-1}\hA_2\hX_{t-1}' \right) \hA_1'(\hA_1\Gamma_1\hA_1')^+\hA_1\Gamma_1 + \hP_1 \left( \sum_t \hE_t\Sigma_2^{-1}\hA_2\hX_{t-1}' \right) + o_P(\sqrt{T}),
\end{aligned}
\end{equation*}
and
\begin{equation}
\label{eq:cc_grad_A1}
\begin{aligned}
    & (\hat\hA_1-\hA_1)\left( \sum_t \hX_{t-1}\hA_2'\Sigma_2^{-1}\hA_2\hX_{t-1}'\right) + 
    \hA_1 \left[ \sum_t \hX_{t-1}(\hat\hA_2-\hA_2)'\Sigma_2^{-1}\hA_2\hX_{t-1}'\right] \\
    = & 
    (\hI - \hP_1) \left( \sum_t \hE_t\Sigma_2^{-1}\hA_2\hX_{t-1}' \right) \hA_1'(\hA_1\Gamma_1\hA_1')^+\hA_1\Gamma_1 + \hP_1 \left( \sum_t \hE_t\Sigma_2^{-1}\hA_2\hX_{t-1}' \right) + o_P(\sqrt{T}).
\end{aligned}
\end{equation}
A similar formula holds for $\hat\hA_2$:
\begin{equation}
\label{eq:cc_grad_A2}
\begin{aligned}
    & \left(\sum_t \hX_{t-1}'\hA_1'\Sigma_1^{-1}\hA_1\hX_{t-1}\right) (\hat\hA_2-\hA_2)' + 
    \left[\sum_t \hX_{t-1}'\hA_1'\Sigma_1^{-1}(\hat\hA_1-\hA_1)\hX_{t-1}\right]\hA_2' \\
    = & 
    \Gamma_2\hA_2'(\hA_2\Gamma_2\hA_2')^+\hA_2 \left(\sum_t \hX_{t-1}'\hA_1'\Sigma_1^{-1}\hE_t\right)(\hI - \hP_2)' 
    + \left(\sum_t \hX_{t-1}'\hA_1'\Sigma_1^{-1}\hE_t\right)\hP_2' + o_P(\sqrt{T}).
\end{aligned}
\end{equation}
Combining \eqref{eq:cc_grad_A1} and \eqref{eq:cc_grad_A2}, it holds that after vectorization
\begin{equation}
\label{eq:final}
    \begin{aligned}
        & \sum_t\begin{pmatrix}
        (\hX_{t-1}\hA_2'\Sigma_2^{-1}\hA_2\hX_{t-1}') \otimes \hI & ( \hX_{t-1}\hA_2'\Sigma_2^{-1})\otimes(\hA_1\hX_{t-1}) \\
        (\hA_2\hX_{t-1}')\otimes(\hX_{t-1}'\hA_1'\Sigma_1^{-1}) & \hI\otimes( \hX_{t-1}'\hA_1'\Sigma_1^{-1}\hA_1\hX_{t-1})
        \end{pmatrix}
        \begin{pmatrix}
        \vect\left(\hat\hA_1-\hA_1\right) \\
        \vect\left(\hat\hA_2'-\hA_2'\right)
        \end{pmatrix} \\
        = & \sum_t \begin{pmatrix}
          \hX_{t-1}\hA_2'\Sigma_2^{-1}\otimes \hP_1+[\Gamma_1\hA_1'(\hA_1\Gamma_1\hA_1')^+\hA_1\hX_{t-1}\hA_2']\Sigma_2^{-1}\otimes(\hI-\hP_1) \\
          \hP_2\otimes \hX_{t-1}'\hA_1'\Sigma_1^{-1}+(\hI-\hP_2)\otimes [\Gamma_2\hA_2'(\hA_2\Gamma_2\hA_2')^+\hA_2\hX_{t-1}'\hA_1'\Sigma_1^{-1}] 
        \end{pmatrix}\vect(\hE_t) + o_P(\sqrt{T}).
    \end{aligned}
\end{equation}
Note that $\Sigma_e=\Sigma_2\otimes\Sigma_1$. Multiplying both sides of \eqref{eq:final} by the matrix 
\begin{equation*}
    \begin{pmatrix}
    \hI\otimes\Sigma_1^{-1} & \hzero \\
    \hzero & \Sigma_2^{-1}\otimes\hI
    \end{pmatrix},
\end{equation*}
and note that $\Sigma_i^{-1}\hP_i=\mathcal P_i\Sigma_i^{-1}$, then \eqref{eq:final} becomes
\begin{equation}
\label{eq:final1}
    \sum_t(\hW_t\Sigma_e^{-1}\hW_t') \begin{pmatrix}
        \vect\left(\hat\hA_1-\hA_1\right) \\
        \vect\left(\hat\hA_2'-\hA_2'\right)
        \end{pmatrix} = \sum_t \hQ_{t-1}\Sigma_e^{-1}\vect(\hE_t) + o_P(\sqrt{T}).
\end{equation}
Since $\|\hA_1\|_F=\|\hat\hA_1\|_F=1$, it holds that $\halpha'\vect(\hat\hA_1-\hA_1)=O_P(1/T)$. By the ergodic theorem, \eqref{eq:final1} implies that
\begin{equation*}
    \hH\begin{pmatrix}
      \vect\left(\hat\hA_1-\hA_1\right) \\
      \vect\left(\hat\hA_2'-\hA_2'\right)
    \end{pmatrix} = \frac{1}{T}\sum_t \hQ_{t-1}\Sigma_e^{-1}\vect(\hE_t) + o_P(1/\sqrt{T}),
\end{equation*}
and the proof is completed by an application of the martingale central limit theorem.
\end{proof}

\bigskip
\begin{proof}[Proof of Theorem~\ref{thm:ls}]
The proof of Theorem~\ref{thm:ls} is the same as that of Theorem~\ref{thm:cc} until \eqref{eq:final}, with the exception (and simplification) that all $\Sigma_i$ and $\hat\Sigma_i$ should be replaced by $\hI$. Using the definition of $\hW_t$ given in the table at the beginning of Section~\ref{sec:asymp}, it is immediately seen that \eqref{eq:final} becomes
\begin{equation}
\label{eq:final2}
    \sum_t(\hW_t\hW_t') \begin{pmatrix}
        \vect\left(\hat\hA_1^{\hbox{\tiny ls}}-\hA_1\right) \\
        \vect\left((\hat\hA_2^{\hbox{\tiny ls}})'-\hA_2'\right)
        \end{pmatrix} = \sum_t \hQ_t\vect(\hE_t) + o_P(\sqrt{T}).
\end{equation}
Since $\|\hat\hA_1^{\hbox{\tiny ls}}\|_F=\|\hA_1\|_F=1$, it holds that $\halpha'\vect(\hat\hA_1^{\hbox{\tiny ls}}-\hA_1)=O_P(1/T)$. By the ergodic theorem, \eqref{eq:final2} implies that
\begin{equation*}
    \hH\begin{pmatrix}
      \vect\left(\hat\hA_1^{\hbox{\tiny ls}}-\hA_1\right) \\
      \vect\left((\hat\hA_2^{\hbox{\tiny ls}})'-\hA_2'\right)
    \end{pmatrix} = \frac{1}{T}\sum_t \hQ_{t-1}\vect(\hE_t) + o_P(1/\sqrt{T}),
\end{equation*}
and the proof is completed by an application of the martingale central limit theorem.
\end{proof}

\bigskip
\begin{proof}[Proof of Theorem~\ref{thm:efficiency}]
If the estimation of $\hA_i$ is done by the MLE without imposing the low rank constraints, the asymptotic covariance matrix would take the same form as $\Xi^{\hbox{\tiny cc}}$, by setting $\mathcal P_i=\hI$ in the definition of $\hQ_t$, which we denote by $\tilde\hQ_t$. Comparing Theorem~\ref{thm:cc} with Theorem~3 of \cite{chen2021autoregressive}, it suffices to show that $\E(\tilde\hQ_t\Sigma_e^{-1}\tilde\hQ_t') \succeq \E(\hQ_t\Sigma_e^{-1}\hQ_t')$. By the definition of $\hQ_t$ and $\tilde \hQ_t$, it holds that
\begin{equation*}
    \tilde\hQ_t = \hQ_t + \begin{pmatrix}
          [(\hI-\Gamma_1\hA_1'(\hA_1\Gamma_1\hA_1')^+\hA_1)\hX_t\hA_2']\otimes(\hI-\PPP_1) \\
          (\hI-\PPP_2)\otimes [(\hI - \Gamma_2\hA_2'(\hA_2\Gamma_2\hA_2')^+\hA_2)\hX_t'\hA_1']
        \end{pmatrix}=:\hQ_t+\bar\hQ_t.
\end{equation*}
It now suffices to show that $\E(\bar\hQ_t\Sigma_e^{-1}\hQ_t')=0$. We write $\bar\hQ_t\Sigma_e^{-1}\hQ_t'$ as a $2\times 2$ block matrix. The top-left corner equals to 
\begin{equation}
\label{eq:topleft}
\begin{aligned}
    & \left\{[(\hI-\Gamma_1\hA_1'(\hA_1\Gamma_1\hA_1')^+\hA_1)\hX_t\hA_2']\otimes(\hI-\PPP_1)\right\}\left(\Sigma_2^{-1}\otimes\Sigma_1^{-1}\right) \\
    & \qquad \left\{ \hA_2\hX_t'\otimes \PPP_1'+[\hA_2\hX_t'\hA_1'(\hA_1\Gamma_1\hA_1')^+\hA_1\Gamma_1]\otimes(\hI-\PPP_1') \right\}.
\end{aligned}
\end{equation}
Obviously,
\begin{align*}
    & \left\{[(\hI-\Gamma_1\hA_1'(\hA_1\Gamma_1\hA_1')^+\hA_1)\hX_t\hA_2']\otimes(\hI-\PPP_1)\right\}\left(\Sigma_2^{-1}\otimes\Sigma_1^{-1}\right) \left( \hA_2\hX_t'\otimes \PPP_1' \right) \\
    = & \left\{[(\hI-\Gamma_1\hA_1'(\hA_1\Gamma_1\hA_1')^+\hA_1)\hX_t\hA_2']\Sigma_2^{-1} \hX_t\hA_2'\right\}\otimes \left\{(\hI-\PPP_1)\Sigma_1^{-1}\PPP_1'\right\}=\hzero.
\end{align*}
For the second term in \eqref{eq:topleft}, note that
\begin{align*}
    & \E\left\{[(\hI-\Gamma_1\hA_1'(\hA_1\Gamma_1\hA_1')^+\hA_1)\hX_t\hA_2']\Sigma_2^{-1}[\hA_2\hX_t'\hA_1'(\hA_1\Gamma_1\hA_1')^+\hA_1\Gamma_1]\right\} \\
    = & [\hI-\Gamma_1\hA_1'(\hA_1\Gamma_1\hA_1')^+\hA_1]\Gamma_1\hA_1'(\hA_1\Gamma_1\hA_1')^+\hA_1\Gamma_1 \\
    = & \Gamma_1^{1/2}\left[\hI-\Gamma_1^{1/2}\hA_1'(\hA_1\Gamma_1\hA_1')^+\hA_1\Gamma_1^{1/2}\right]\left[\Gamma_1^{1/2}\hA_1'(\hA_1\Gamma_1\hA_1')^+\hA_1\Gamma_1^{1/2}\right]\Gamma_1^{1/2}
    = \hzero,
\end{align*}
where the last identity is due to the fact that $\Gamma_1^{1/2}\hA_1'(\hA_1\Gamma_1\hA_1')^+\hA_1\Gamma_1^{1/2}$ is the orthogonal projection matrix to the row space of $\hA_1\Gamma_1^{1/2}$. Therefore, the expectation of \eqref{eq:topleft} is zero. The expectation of the other three blocks of $\bar\hQ_t\Sigma_e^{-1}\hQ_t'$ can be shown to be zero similarly. The proof is complete.
\end{proof}

\bigskip
The proof of Corollary~\ref{cor:svd} is a direct application of the following lemma and Theorems~\ref{thm:ls}, \ref{thm:cc}. Lemma~\ref{lem:svd_clt} is regarding the central limit theorems of singular vectors under the fixed dimensional setting. Although some cases are available in the literature, we have not seen any formulation that is exactly the same. Therefore, we provide Lemma~\ref{lem:svd_clt} and a proof here for the completeness. The proof essentially relies on the matrix perturbation theory.

\begin{lemma}
\label{lem:svd_clt}
Suppose $\hM$ is a $p\times p$ symmetric matrix of rank $r\leq p$, and let $\hM=\hU\Lambda\hU'$ be its spectral decomposition, where $\hU$ is a $p\times r$ ortho-normal matrix and $\Lambda$ is a $r\times r$ diagonal matrix. Denote the $j$-th diagonal element of $\Lambda$ by $\lambda_j$, and define $\hlambda=(\lambda_1,\ldots,\lambda_r)'$. Assume that the $\lambda_j$'s are distinct. Suppose $\{\hat\hM_n\}$ is a sequence of random matrices and $\{a_n\}$ is a sequence of diverging positive numbers such that
\begin{equation*}
    a_n\vect(\hat\hM_n-\hM)\Rightarrow N(\hzero,\Theta).
\end{equation*}
Let $\hat\hM_n=\hat\hU_n\hat\Lambda_n\hat\hU_n'$ be the spectral decomposition of $\hat\hM$ corresponding to the $r$ leading eigenvalues. 
Define the matrix $\hR$ as
\begin{equation*}
    \hR=(\hI_{r}\otimes\hU,\hI_{r}\otimes \hU^{\perp})
    \begin{pmatrix}
    (\Lambda\otimes\hI_{r}-\hI_{r}\otimes\Lambda+\hL_{r}\hL_{r}')^{-1}(\hI_{r^2}-\hL_{r}\hL_{r}')(\hU'\otimes\hU') \\
    (\Lambda^{-1}\hU')\otimes(\hU^\perp)'
    \end{pmatrix}.
\end{equation*}
Then
\begin{equation*}
    a_n\vect(\hat\hU_n-\hU)\Rightarrow N(\hzero,\hR\Theta\hR'),
\end{equation*}
and
\begin{equation*}
    a_n(\hat\hlambda -\hlambda) \Rightarrow N\left[\hzero,\hL_{r}'(\hU'\otimes\hU')\Theta(\hU\otimes\hU)\hL_{r}\right].
\end{equation*}
\end{lemma}

\begin{proof}
Due to the assumption that the $\lambda_j$'s are distinct, by the Wedin's Theorem \citep{wedin1972perturbation}, it holds that $\hat\hU=\hU+O_P(1/a_n)$ and $\hat\hlambda=\hlambda+O_P(1/a_n)$.

Expand $\hat\hM\hat\hU=\hat\hU\hat\Lambda$ around the true values and omit small order terms, we have
\begin{equation}
\label{eq:evd_expansion}
    (\hat\hU-\hU)\Lambda + \hU(\hat\Lambda-\Lambda) - \hM(\hat \hU-\hU) = (\hat\hM-\hM)\hU + o_P(1/a_n).
\end{equation}
\noindent Multiplying both sides of \eqref{eq:evd_expansion} by $\hU'$ leads to
\begin{equation*}
    \label{eq:U1}
    \hU'(\hat\hU-\hU)\Lambda - \Lambda\hU'(\hat \hU-\hU) + (\hat\Lambda-\Lambda)= \hU'(\hat\hM-\hM)\hU + o_P(1/a_n).
\end{equation*}
Using the properties of the $\hL$ matrices introduced at the end of Section~\ref{sec:intro}, it follows that 
\begin{equation}
\label{eq:evd_lambda}
    \hat\hlambda-\hlambda = \hL_r'(\hU'\otimes\hU')\vect(\hat\hM-\hM),
\end{equation}
and
\begin{equation}
    \label{eq:U2}
    \begin{aligned}
    & \left(\Lambda\otimes\hI_r-\hI_r\otimes\Lambda+\hL_r\hL_r'\right)\vect\left[\hU'(\hat\hU-\hU)\right] \\
    & \qquad\qquad = (\hI_{r^2}-\hL_r\hL_r')(\hU'\otimes\hU')\vect(\hat\hM-\hM) + o_P(1/a_n).
    \end{aligned}
\end{equation}
The asymptotic distribution of $\hat\hlambda$ follows \eqref{eq:evd_lambda} immediately. In deriving \eqref{eq:evd_lambda} and \eqref{eq:U2}, we have implicitly used the fact that \begin{equation*}
    \hU'(\hat\hU-\hU) + (\hat\hU-\hU)'\hU = o_P(1/a_n).
\end{equation*}
From \eqref{eq:U2} we deduce that
\begin{equation}
    \label{eq:U3}
    \begin{aligned}
    & \vect\left[\hU'(\hat\hU-\hU)\right]  = (\hI_r\otimes\hU')\vect(\hat\hU-\hU) \\
    & \qquad\qquad = \left(\Lambda\otimes\hI_r-\hI_r\otimes\Lambda+\hL_r\hL_r'\right)^{-1}(\hI_{r^2}-\hL_r\hL_r')(\hU'\otimes\hU')\vect(\hat\hM-\hM) + o_P(1/a_n).
    \end{aligned}
\end{equation}
Multiplying both sides of \eqref{eq:evd_expansion} by $(\hU^\perp)'$ gives
\begin{equation*}
    (\hU^\perp)'(\hat\hU-\hU)\Lambda = (\hU^\perp)'(\hat\hM-\hM)\hU + o_P(1/a_n),
\end{equation*}
and therefore,
\begin{equation*}
    (\hU^\perp)'(\hat\hU-\hU) = (\hU^\perp)'(\hat\hM-\hM)\hU\Lambda^{-1} + o_P(1/a_n),
\end{equation*}
and
\begin{equation}
    \label{eq:Uc}
    \begin{aligned}
    \vect\left[(\hU^\perp)'(\hat\hU-\hU)\right]  & = \left[\hI_r\otimes(\hU^\perp)'\right]\vect(\hat\hU-\hU) \\
    & = \left[(\Lambda^{-1}\hU')\otimes(\hU^\perp)'\right]\vect(\hat\hM-\hM) + o_P(1/a_n).
    \end{aligned}
\end{equation}
Combining \eqref{eq:U3} and \eqref{eq:Uc}, it holds that
\begin{equation*}
\begin{aligned}
    & \begin{pmatrix}
    \hI_r\otimes\hU'\\
    \hI_r\otimes(\hU^\perp)'
    \end{pmatrix}\vect(\hat\hU-\hU) \\
    & \qquad = \begin{pmatrix}
    (\Lambda\otimes\hI_{r}-\hI_{r}\otimes\Lambda+\hL_{r}\hL_{r}')^{-1}(\hI_{r^2}-\hL_{r}\hL_{r}')(\hU'\otimes\hU') \\
    (\Lambda^{-1}\hU')\otimes(\hU^\perp)'
    \end{pmatrix}\vect(\hat\hM-\hM) + o_P(1/a_n).
\end{aligned}
\end{equation*}
Since
\begin{equation*}
    \begin{pmatrix}
    \hI_r\otimes\hU'\\
    \hI_r\otimes(\hU^\perp)'
    \end{pmatrix}^{-1}=(\hI_{r}\otimes\hU,\hI_{r}\otimes \hU^{\perp}),
\end{equation*}
it follows that
\begin{equation*}
    \vect(\hat\hU-\hU)=\hR\,\vect(\hat\hM-\hM) + o_P(1/a_n),
\end{equation*}
and the proof is complete.
\end{proof}

\bigskip
\begin{proof}[Proof of Theorem~\ref{thm:bic}]
We give the proof for the joint EBIC($r_1,r_2$). The proof for separate EBIC follows similar arguments, and will be skipped.

Let $\sigma_0^2:=\E \|\hE_t\|_F^2/(d_1d_2)$. It is straightforward to show that when $(r_1,r_2)=(k_1,k_2)$
\begin{equation*}
    \frac{1}{Td_1d_2}\sum_{t=2}^T\|\hX_t-\hat\hA_1^{\hbox{\tiny ls}}(r_1,r_2)\hX_{t-1}(\hat\hA_2^{\hbox{\tiny ls}}(r_1,r_2))'\|_F^2 \stackrel{p}{\rightarrow} \sigma_0^2;
\end{equation*}
when $r_1<k_1$ or $r_2<k_2$,
\begin{equation*}
    \frac{1}{Td_1d_2}\sum_{t=2}^T\|\hX_t-\hat\hA_1^{\hbox{\tiny ls}}(r_1,r_2)\hX_{t-1}(\hat\hA_2^{\hbox{\tiny ls}}(r_1,r_2))'\|_F^2 \stackrel{p}{\rightarrow} \sigma_1^2>\sigma_0^2;
\end{equation*}
and when $r_1\geq k_1,\,r_2\geq k_2$, and at least one of the inequalities is strict,
\begin{equation*}
    \frac{1}{Td_1d_2}\sum_{t=2}^T\|\hX_t-\hat\hA_1^{\hbox{\tiny ls}}(r_1,r_2)\hX_{t-1}(\hat\hA_2^{\hbox{\tiny ls}}(r_1,r_2))'\|_F^2 = \sigma_0^2 + O_p(1/T).
\end{equation*}
Then a direct calculation shows that the joint EBIC does not under select or over select the ranks with probability approaching one.
\end{proof}

\bigskip

\noindent{\large\bf Appendix II: RRMAR($p$) Models}

\bigskip

We consider the extension to the RRMAR($p$) model \eqref{eq:marrrp}. Assume that $\rk(\hA_{ji})=k_{ji}$, and $\|\hA_{j1}\|_F=1$ for $1\leq j\leq p$. We also assume that the model \eqref{eq:marrrp} is causal and stationary, i.e. the roots of the polynomial
\begin{equation}
\label{eq:stat}
    \det\left[\hI-(\hA_{12}\otimes\hA_{11})z-(\hA_{22}\otimes\hA_{21})z^2-\cdots-(\hA_{p2}\otimes\hA_{p1})z^p\right]=0
\end{equation}
are all strictly outside the unit circle.
The least squares and maximum likelihood estimators are extended in a straightforward way from the RRMAR(1) case, and so are the corresponding iterative algorithms. 

We next give the extensions of Theorems~\ref{thm:ls} and \ref{thm:cc}. We begin with the extension of notations.  First of all, recall that we assume $\|\hA_{j1}\|_F=1$ for the parameter identifiability, so we rescale
$\hat\hA_{ji}^{\hbox{\tiny ls}}$ and $\hat\hA_{ji}^{\hbox{\tiny cc}}$ so that
$\|\hat\hA_{j1}^{\hbox{\tiny ls}}\|_F=1$ and $\|\hat\hA_{j1}^{\hbox{\tiny cc}}\|_F=1$. 
We update the table of notations at the beginning of Section~\ref{sec:asymp} as follows, where the index $j$ runs from 1 to $p$, and $i=1,2$.
\begin{center}
\begin{tabular}{l|l|l}
\hline
    Notations & LSE & MLE  \\\hline
    $\Gamma_{j1}$ & $\E(\hX_t\hA_{j2}'\hA_{j2}\hX_t')$ & $\E(\hX_t\hA_{j2}'\Sigma_2^{-1}\hA_{j2}\hX_t')$ \\
    $\Gamma_{j2}$ & $\E(\hX_t'\hA_{j1}'\hA_{j1}\hX_t)$ & $\E(\hX_t'\hA_{j1}'\Sigma_1^{-1}\hA_{j1}\hX_t)$\\
    $\mathbb{P}_{ji}$ & \hbox{orthogonal projection to} $\mathrm{col}(\hA_{ji})$ & \hbox{orthogonal projection to} $\mathrm{col}(\Sigma_i^{-1/2}\hA_{ji})$\\
    $\PPP_{ji}$ & $\mathbb{P}_{ji}$ & $\Sigma_i^{-1/2}\mathbb{P}_{ji}\Sigma_i^{1/2}$ \\
    $\mathfrak P_{ji}$ & $\hA_{ji}'(\hA_{ji}\Gamma_{ji}\hA_{ji}')^+\hA_{ji}\Gamma_{ji}$ & Same \\
    $\hW_{jt}$ & $[(\hA_{j2}\hX_t')\otimes\hI_{d_1},\hI_{d_2}\otimes(\hA_{j1}\hX_t)]'$ & {Same} \\
    $\hQ_t$ & Equation~\eqref{eq:Qt} & Same \\
    $\hC_t$ & Remove $\Sigma_e^{-1}$ in \eqref{eq:Ct} & Equation \eqref{eq:Ct} \\\hline
\end{tabular}
\end{center}
The $\hQ_t$ is defined as
\begin{equation}
    \label{eq:Qt}
    \hQ_t:=[\hQ_{1,t-1}',\hQ_{2,t-2}',\ldots,\hQ_{p,t-p}']',
\end{equation}
where
\begin{equation*}
    \hQ_{jt}=\hW_{jt} - \begin{pmatrix}
        (\hI-\frak P_{j1})'\hX_t\hA_{j2}'\otimes(\hI-\mathcal P_{j1}) \\
        (\hI-\mathcal P_{j2})\otimes(\hI-\frak P_{j2})\hX_t'\hA_{j1}'
    \end{pmatrix}.
\end{equation*}
The $\hC_t$ is defined as
\begin{equation}
\label{eq:Ct}
    \hC_t=\begin{pmatrix}
                    \hW_{1,t-1} \Sigma^{-1}_e\hW'_{1,t-1} &  \hQ_{1,t-1} \Sigma^{-1}_e\hW'_{2,t-2} & \cdots &\hQ_{1,t-1} \Sigma^{-1}_e\hW'_{p,t-p} \\
                    \hQ_{2,t-2} \Sigma^{-1}_e\hW'_{1,t-1} & \hW_{2,t-2} \Sigma^{-1}_e\hW'_{2,t-2}  & \cdots &\hQ_{2,t-2} \Sigma^{-1}_e\hW'_{p,t-p}\\
                    \vdots & \vdots & \ddots & \vdots\\
                    \hQ_{p,t-p} \Sigma^{-1}_e\hW'_{1,t-1} & \hQ_{p,t-p} \Sigma^{-1}_e\hW'_{2,t-2}  & \cdots & \hW_{p,t-p} \Sigma^{-1}_e\hW'_{p,t-p}
                \end{pmatrix},
\end{equation}
Let $\hgamma_j$ be a block vector with $2p$ blocks, of size $d_1^2$ and $d_2^2$ alternatively, whose $(2j-1)$-th block is $\vect(\hA_{j1})$, and everywhere else zero. Define $\hH=\E \hC_t + \sum_{j=1}^p \hgamma_j\hgamma_j'$.

\begin{theorem}
\label{thm:rrmarp}
  Assume that $\{\hE_t\}$ are i.i.d. with mean zero and finite second
  moments. Assume that
  $0<\rk(\hA_{ji})=k_{ji}\leq d_i$ and the stationarity condition defined through the roots of \eqref{eq:stat}. 
  \begin{enumerate}
      \item Assume $\Sigma_e$ is non-singular, then
  \begin{equation*}
    \sqrt{T}\begin{pmatrix}
      \vect\left[\hat\hA_{11}^{\hbox{\tiny ls}}-\hA_{11}\right] \\
      \vect\left[(\hat\hA_{12}^{\hbox{\tiny ls}})'-\hA_{12}'\right]\\
      \cdots \\
      \vect\left[\hat\hA_{p1}^{\hbox{\tiny ls}}-\hA_{p1}\right] \\
      \vect\left[(\hat\hA_{p2}^{\hbox{\tiny ls}})'-\hA_{p2}'\right]
    \end{pmatrix}
    \Rightarrow N(\hzero, \Xi^{\hbox{\tiny ls}}),
  \end{equation*}
  where
  \begin{equation*}
    \Xi^{\hbox{\tiny ls}} :=\hH^{-1}\E(\hQ_t\Sigma_e\hQ_t')\hH^{-1}.
  \end{equation*}      
      \item If in addition, $\Sigma_e$ is of the form \eqref{eq:ocov}, and is non-singular, then 
  \end{enumerate}
  \begin{equation*}
    \sqrt{T}\begin{pmatrix}
      \vect\left[\hat\hA_{11}^{\hbox{\tiny cc}}-\hA_{11}\right] \\
      \vect\left[(\hat\hA_{12}^{\hbox{\tiny cc}})'-\hA_{12}'\right]\\
      \cdots \\
      \vect\left[\hat\hA_{p1}^{\hbox{\tiny cc}}-\hA_{p1}\right] \\
      \vect\left[(\hat\hA_{p2}^{\hbox{\tiny cc}})'-\hA_{p2}'\right]
    \end{pmatrix}
    \Rightarrow N(\hzero, \Xi^{\hbox{\tiny cc}}),
  \end{equation*}
  where
  \begin{equation*}
    \Xi^{\hbox{\tiny cc}} :=\hH^{-1}\E(\hQ_t\Sigma_e^{-1}\hQ_t')\hH^{-1}.
  \end{equation*}
\end{theorem}

\bigskip

\begin{proof}[Proof of Theorem~\ref{thm:rrmarp}]
The proof of Theorem~\ref{thm:rrmarp} follows essentially the same lines as those of Theorem~\ref{thm:cc}. We will only prove (ii) since the proof (i) is almost the same but much simpler. Throughout the proof we are concerned with the MLE $\hat\hA_{ji}^{\hbox{\tiny cc}}$, but we will omit the superscript {\footnotesize cc} for notational simplicity. We begin with the gradient condition for $\hat\hA_{11}$:
\begin{equation}
\label{eq:cc_grad1}
    \hat \hA_{11}\hat\hS_{11xx}=\hat\Sigma_{1}^{1/2}\hat\hU_{11}\hat\hU_{11}'\hat\Sigma_{1}^{-1/2}\hat\hS_{11yx},
\end{equation}
where 
\begin{align*}
    \hat\hS_{11xx} & =\sum_t \hX_{t-1}\hat\hA_{12}'\hat\Sigma_2^{-1}\hat\hA_{12}\hX_{t-1}',\\
    \hat\hS_{11yx} & =\sum_t \left(\hX_t-\sum_{j=2}^p\hat\hA_{j1}\hX_{t-j}\hat\hA_{j2}'\right)\hat\Sigma_2^{-1}\hat\hA_{12}\hX_{t-1}', \\
    \hat\Sigma_1 & = \frac{1}{T-1} \sum_t\left(\hX_t-\sum_{j=1}^p\hat\hA_{j1}\hX_{t-j}\hat\hA_{j2}'\right)\hat\Sigma_2^{-1}\left(\hX_t-\sum_{j=1}^p\hat\hA_{j1}\hX_{t-j}\hat\hA_{j2}'\right)',
\end{align*}
and $\hat\hU_{11}$ is the $d_1\times k_{11}$ matrix consisting of the first $k_{11}$ leading eigenvectors (all normalized to have unit length) of $\hat\Sigma_{1}^{-1/2}\hat\hS_{11yx}\hat\hS_{1xx}^{-1}\hat\hS_{11xy}\hat\Sigma_{1}^{-1/2}$. The equations for other $\hat\hA_{ji}$ ($1\leq j\leq p$, $i=1,2$) are similar.

Define 
\begin{align*}
  \hat\hM_{11}&=T^{-1/2}\hat\Sigma_{1}^{-1/2}\hat\hS_{11yx}\hat\hS_{11xx}^{-1/2},  \\
  \tilde\hM_{11}&=T^{-1/2}\hat\Sigma_1^{-1/2}\hA_{11}\left(\sum_t \hX_{t-1}\hA_{12}'\hat\Sigma_2^{-1}\hat\hA_{12}\hX_{t-1}'\right)\hat \hS_{11xx}^{-1/2}, \\
  \hM_{11}&=\Sigma_1^{-1/2}\hA_{11}\Gamma_{11}^{1/2},
\end{align*}
and the corresponding rank-$k_{11}$ matrices $\hat \hU_{11}$, $\tilde \hU_{11}$ and $\hU_{11}$ similarly as appearing in the proof of Theorem~\ref{thm:cc}, and let $\hP_{11}=\Sigma_1^{1/2}\hU_{11}\hU_{11}'\Sigma_1^{-1/2}$. Following similar arguments there, and note that
\begin{align*}
    \hat\hM_{11} - \tilde\hM_{11} =\; & T^{-1/2}\hat\Sigma_1^{-1/2}\left(\sum_t \hE_t\hat\Sigma_2^{-1}\hat\hA_{12}\hX_{t-1}'\right)\hat \hS_{11xx}^{-1/2} \\
    & + T^{-1/2}\hat\Sigma_1^{-1/2}\left[\sum_t\sum_{j=2}^p\left(\hA_{j1}\hX_{t-j}\hA_{j2}' - \hat\hA_{j1}\hX_{t-j}\hat\hA_{j2}'\right)\hat\Sigma_2^{-1}\hat\hA_{12}\hX_{t-1}'\right]\hat \hS_{11xx}^{-1/2},
\end{align*}
we have
\begin{align*}
   & (\hat\hA_{11}-\hA_{11})\left( \sum_t \hX_{t-1}\hA_{12}'\Sigma_2^{-1}\hA_{12}\hX_{t-1}'\right) + 
    \hA_{11} \left[ \sum_t \hX_{t-1}(\hat\hA_{12}-\hA_{12})'\Sigma_2^{-1}\hA_{12}\hX_{t-1}'\right] \\
    = & 
    (\hI - \hP_{11}) \left( \sum_t \hE_t\Sigma_2^{-1}\hA_{12}\hX_{t-1}' \right) \hA_{11}'(\hA_{11}\Gamma_{11}\hA_{11}')^+\hA_{11}\Gamma_{11} + \hP_{11} \left( \sum_t \hE_t\Sigma_2^{-1}\hA_{12}\hX_{t-1}' \right) \\
    & + (\hI - \hP_{11}) \sum_{j=2}^p\left[\sum_t(\hA_{j1}-\hat\hA_{j1})\hX_{t-j}\hA_{j2}'\Sigma_2^{-1}\hA_{12}\hX_{t-1}' \right]\hA_{11}'(\hA_{11}\Gamma_{11}\hA_{11}')^+\hA_{11}\Gamma_{11}  \\
    & + (\hI - \hP_{11}) \sum_{j=2}^p\left[\sum_t\hA_{j1}\hX_{t-j}(\hA_{j2}-\hat\hA_{j2})'\Sigma_2^{-1}\hA_{12}\hX_{t-1}' \right]\hA_{11}'(\hA_{11}\Gamma_{11}\hA_{11}')^+\hA_{11}\Gamma_{11}  \\
    & + \hP_{11}\sum_{j=2}^p\left[\sum_t(\hA_{j1}-\hat\hA_{j1})\hX_{t-j}\hA_{j2}'\Sigma_2^{-1}\hA_{12}\hX_{t-1}' \right] \\
    & + \hP_{11}\sum_{j=2}^p\left[\sum_t\hA_{j1}\hX_{t-j}(\hA_{j2}-\hat\hA_{j2})'\Sigma_2^{-1}\hA_{12}\hX_{t-1}' \right] + o_P(\sqrt{T}).
\end{align*}
Recall that $\mathfrak P_{11}:=\hA_{11}'(\hA_{11}\Gamma_{11}\hA_{11}')^+\hA_{11}\Gamma_{11}$ which is an oblique projection to the row space of $\hA_{11}$, then the previous equation can be simplified to
\begin{equation}
    \label{eq:grad_p}
\begin{aligned}
   & (\hat\hA_{11}-\hA_{11})\left( \sum_t \hX_{t-1}\hA_{12}'\Sigma_2^{-1}\hA_{12}\hX_{t-1}'\right) + 
    \hA_{11} \left[ \sum_t \hX_{t-1}(\hat\hA_{12}-\hA_{12})'\Sigma_2^{-1}\hA_{12}\hX_{t-1}'\right] \\
    & + \sum_{j=2}^p \sum_t(\hat\hA_{j1}-\hA_{j1})\hX_{t-j}\hA_{j2}'\Sigma_2^{-1}\hA_{12}\hX_{t-1}' \\
    & - (\hI - \hP_{11}) \sum_{j=2}^p\left[\sum_t(\hat\hA_{j1}-\hA_{j1})\hX_{t-j}\hA_{j2}'\Sigma_2^{-1}\hA_{12}\hX_{t-1}' \right] (\hI-\mathfrak P_{11})  \\
    & + \sum_{j=2}^p \sum_t\hA_{j1}\hX_{t-j}(\hat\hA_{j2}-\hA_{j2})'\Sigma_2^{-1}\hA_{12}\hX_{t-1}' \\
    & - (\hI - \hP_{11}) \sum_{j=2}^p\left[\sum_t\hA_{j1}\hX_{t-j}(\hat\hA_{j2}-\hA_{j2})'\Sigma_2^{-1}\hA_{12}\hX_{t-1}' \right] (\hI-\mathfrak P_{11})  \\
    = & 
    \sum_t \hE_t\Sigma_2^{-1}\hA_{12}\hX_{t-1}' - (\hI - \hP_{11}) \left( \sum_t \hE_t\Sigma_2^{-1}\hA_{12}\hX_{t-1}' \right) (\hI-\mathfrak P_{11})  + o_P(\sqrt{T}).
\end{aligned}
\end{equation}
Utilizing the quantites $\hW_{jt}$ and $\hQ_{jt}$ defined earlier,
the equation \eqref{eq:grad_p} and similar ones for other $\hat\hA_{ji}$ can be combined into
\begin{align*}
            & \sum_t
                \underbrace{\begin{pmatrix}
                    \hW_{1,t-1} \Sigma^{-1}_e\hW'_{1,t-1} &  \hQ_{1,t-1} \Sigma^{-1}_e\hW'_{2,t-2} & \cdots &\hQ_{1,t-1} \Sigma^{-1}_e\hW'_{p,t-p} \\
                    \hQ_{2,t-2} \Sigma^{-1}_e\hW'_{1,t-1} & \hW_{2,t-2} \Sigma^{-1}_e\hW'_{2,t-2}  & \cdots &\hQ_{2,t-2} \Sigma^{-1}_e\hW'_{p,t-p}\\
                    \vdots & \vdots & \ddots & \vdots\\
                    \hQ_{p,t-p} \Sigma^{-1}_e\hW'_{1,t-1} & \hQ_{p,t-p} \Sigma^{-1}_e\hW'_{2,t-2}  & \cdots & \hW_{p,t-p} \Sigma^{-1}_e\hW'_{p,t-p}
                \end{pmatrix}}_{\hC_t}
                \underbrace{\begin{pmatrix}
                    \vect(\hat{\hA}_{11} - \hA_{11}) \\
                    \vect(\hat{\hA}'_{12} - \hA'_{12}) \\
                    \vect(\hat{\hA}_{21} - \hA_{21}) \\
                    \vect(\hat{\hA}'_{22} - \hA'_{22}) \\
                    \vdots\\
                    \vect(\hat{\hA}_{p1} - \hA_{p1}) \\
                    \vect(\hat{\hA}'_{p2} - \hA'_{p2}) 
                \end{pmatrix}}_{\hat\halpha-\halpha}\\
                & =
                 \sum_t \hQ_t\Sigma_e^{-1}
                \vect(\hE_t) + o_P(\sqrt{T}).
            \end{align*}
Note that the notations $\hC_t$, $\halpha$ and $\hat\halpha$ are also defined in the preceding equation. Recall that $\hH=\E \hC_t + \sum_{j=1}^p \hgamma_j\hgamma_j'$, then the preceding equation becomes
\begin{equation*}
    \hH(\hat\halpha-\halpha)=
                 T^{-1}\sum_t \hQ_t\Sigma_e^{-1}
                \vect(\hE_t) + o_P(1/\sqrt{T}),
\end{equation*}
and the conclusion of the Part (ii) follows.
\end{proof}

\end{document}